\documentclass[12pt]{article}
\usepackage{amssymb,amsmath,amsthm,amsxtra,amsfonts}
\usepackage[margin=1in]{geometry}
\usepackage[onehalfspacing]{setspace}
\usepackage{aliascnt}
\usepackage{color}
\usepackage[usenames,dvipsnames,table]{xcolor}
\usepackage{subcaption}
\usepackage{float}
\usepackage{dsfont}
\usepackage{graphicx}
\usepackage{mathrsfs}
\usepackage{tikz}
\usetikzlibrary{plotmarks,arrows,calc,patterns}
\usepackage[round]{natbib}
\usepackage[colorlinks=true,citecolor=Blue,urlcolor=Blue,linkcolor=Blue,filecolor=Blue]{hyperref}
\usepackage{algorithm,algpseudocode}
\usepackage{threeparttable}
\usepackage{booktabs}
\usepackage{rotating}
\usepackage{lscape}
\usepackage{authblk}
	{\rule[-4mm]{0mm}{6mm}\noindent\vspace{3cm}\end{table}}

	{Table continued on next page ... \end{table}}

\allowdisplaybreaks
\theoremstyle{plain}
\newtheorem{theorem}{Theorem}[section]
\newtheorem{proposition}{Proposition}[section]
\newtheorem{lemma}{Lemma}[section]

\newtheorem{assumption}{Assumption}[section]

\theoremstyle{remark}
\newtheorem{example}{Example}[section]

\newcommand{\R}{\ensuremath{\mathbf{R}}}
\newcommand{\N}{\ensuremath{\mathbf{N}}}

\newcommand{\E}{\ensuremath{\mathbb{E}}}

\newcommand{\one}{\ensuremath{\mathds{1}}}
\newcommand{\Var}{\ensuremath{\mathrm{Var}}}
\newcommand{\dx}{\ensuremath{\mathrm{d}}}
\newcommand{\sign}{\ensuremath{\mathrm{sign}}}
\newcommand{\tablexplain}[1]{\rule[0mm]{0mm}{5mm}\noindent\footnotesize{#1}}

\newenvironment{ctabular}{\begin{center}\begin{tabular}}{\end{tabular}\end{center}}

\providecommand{\keywords}[1]
{\small	
\textbf{\textit{Keywords:}} #1
}

\usepackage{xr}
\begin{document}
\title{Binary Choice under Asymmetric Loss in a Data-Rich Environment:
Theory and an Application  to Algorithmic Fairness\footnote{We are grateful to seminar and conference participants at the University of Pennsylvania, Brown University, Universit\'{e} de Montr\'{e}al, National University of Singapore, Singapore Management University, University of Notre Dame, Boston University, HU Berlin, UNC Wilmington, 2022 CEME Conference for Young Econometricians, 2022 AEA/ASSA Annual Meeting, IAAE webinar series, 2021 North American/Asian/China/Australasia/European Meetings of the Econometric Society, Bristol ESG 2021, and the 41st International Symposium on Forecasting for helpful comments.}
}

\author[1]{Andrii Babii\thanks{Department of Economics. Email: \href{mailto:babii.andrii@gmail.com}{babii.andrii@gmail.com}.}}
\author[1]{Xi Chen\thanks{Department of Statistics. Email: \href{mailto:xich@live.unc.edu}{xich@live.unc.edu}.}}
\author[1]{Eric Ghysels\thanks{Department of Economics and Department of Finance, Kenan-Flagler Business School. Email: \href{mailto:eghysels@unc.edu}{eghysels@unc.edu}.}}
\author[2]{Rohit Kumar\thanks{HSS, IITD. Email: \href{mailto:sirohi.rohit@gmail.com}{sirohi.rohit@gmail.com}.}}
\affil[1]{UNC Chapel Hill}
\affil[2]{IIT Delhi}

\maketitle

\begin{abstract}
We study the binary choice problem in a data-rich environment with asymmetric loss functions. The econometrics literature covers nonparametric binary choice problems but does not offer computationally attractive solutions in data-rich environments. The machine learning literature has many algorithms but is focused mostly on loss functions that are independent of covariates. We show that theoretically valid decisions on binary outcomes with general loss functions can be achieved via a very simple loss-based reweighting of logistic regression or state-of-the-art machine learning techniques. We apply our analysis to algorithmic fairness in pretrial detentions.
\end{abstract}


\thispagestyle{empty}
\setcounter{page}{0}

\newpage

\section{Introduction}
Many economic decisions are inherently binary. This includes life-changing decisions such as pretrial release, college admission, job hiring, medical testing and treatment---areas increasingly influenced by automated algorithmic processes based on vast data inputs. It also extends to routine tasks such as loan or credit card approvals, fraud detection, or spam filters. The topic has gained interest from a diverse set of fields, ranging from economics and computer science, to machine learning (ML), among others, and depending on the discipline is also known as classification or screening problems.

In this paper, we focus on data-driven binary decision problems characterized as a solution to the following risk-minimization problem:
\begin{equation}\label{eq:decision_problem}
	\min_{f:\mathcal{X}\to\{-1,1\}}\E_{X,Y}[\ell(f(X),Y,X)],
\end{equation}
where $Y\in\{-1,1\}$ is a binary outcome, $X\in\mathcal{X}\subset\R^d$ is a vector of covariates, $f:\mathcal{X}\to\{-1,1\}$ is a binary decision (or prediction) function, $\ell:\{-1,1\}^2\times\mathcal{X}\to\R_+$ is a loss function, and the expected value $\E_{X,Y}$ is taken with respect to the distribution of $(X,Y)$. The problem in (\ref{eq:decision_problem}) is minimized over all measurable functions of covariates and the loss function $\ell$ describes the asymmetric consequences in different states of the world. For example, pretrial detention decisions may reflect the asymmetric cost of (a) keeping someone wrongly in jail before his/her trial, (b) releasing a recidivist, versus the social benefits of (a) keeping potential recidivists in jail, and (b) releasing non-recidivists who could be productive members of society. In the loan approval application, the losses may depend on loan size and other considerations such as equal credit opportunities.

The importance of asymmetries in economic decision problems has been recognized for a long time; see \cite{granger1969prediction}, \cite{manski1989estimation}, \cite{granger2000economic}, \cite{diebold1996optimal,diebold1997optimal}, \cite{elliott2016economic}, among others. The combination of high-dimensional data and general loss functions in binary decision problems is challenging and not well understood.\footnote{In contrast, the asymmetric decisions in the regression setting can be achieved via quantile regression, see \cite{koenker1978regression}, the asymmetric least-squares, see \cite{newey1987asymmetric}, or more generally with M-estimators based on a suitable asymmetric loss function; see \cite{elliott2016economic}.} Econometricians have studied nonparametric binary choice problems for a long time, but the literature does not offer computationally attractive solutions for high-dimensional datasets. For example, \cite{elliott2013predicting} consider the binary utility maximization problem similar to (\ref{eq:decision_problem}) and propose to solve its empirical counterpart by extending the maximum score estimator of \cite{manski1975maximum}. These examples typically involve \textit{non-smooth} and \textit{non-convex} combinatorial optimization which is known to be NP-hard and, hence, challenging to apply in modern data-rich environments.\footnote{Several efforts have been made to make these problems computationally feasible. For example, \cite{horowitz1992smoothed} proposes to use kernel smoothing which addresses the lack of smoothness, but not the lack of convexity. \cite{florios2008exact} considers the mixed-integer optimization that requires special solvers without guarantees of finding the global optimum.}

In this paper, we formally show that the valid binary decisions corresponding to the problem (\ref{eq:decision_problem}) can be obtained by solving simple weighted logistic regression or more generally weighted machine learning problems including deep learning, boosting, and LASSO, where the weights are computed directly from the loss function. Our approach can be understood as solving a convexified counterpart to the empirical utility maximization problem of \cite{elliott2013predicting} similarly to modern classification methods solving the convexified counterparts to the maximum score estimator of \cite{manski1975maximum}. 

The practical implementation of our method is remarkably simple, which is why it can be discussed here. The first ingredient comes from a policy/decision-maker who has to provide a quartet of loss functions pertaining to (a) true positives, denoted by $\ell_{1,1}(x)$, (b) true negatives, $\ell_{-1,-1}(x)$, (c) false positives, $\ell_{-1,1}(x)$, and (d) false negatives, $\ell_{1,-1}(x)$, where we define $\ell_{f,y}(x):=\ell(f,y,x)$. This requires a careful selection of relevant covariates, $x\in\mathcal{X}$, and functional forms describing the economic decision problem. These are usually determined by either a cost/utility function of the social planner (policy/decision-maker), or are obtained through some type of cost-benefit analysis. Given these economic inputs, the task of the econometrician is to compute weights from $\omega(y,x)$ with $y\in\{-1,1\}$ and $x\in\mathcal{X}$ that will be applied to the appropriate classification procedures. We show that this very simple loss-based reweighting leads to valid binary decisions without strong distributional assumptions. In some cases, e.g., for carefully crafted deep learning architectures, these decisions are optimal from the minimax point of view.

Our methodology can incorporate various preferences through the loss function including the treatment of protected groups which is related to the ongoing debate about algorithmic biases. As more and more decisions affecting our daily lives have become digitized and automated by data-driven algorithms, concerns have been raised regarding such biases. The potential for ML algorithmic outcomes to reproduce and reinforce existing discrimination against legally protected groups has been of great concern. Gender and race are two leading examples. \cite{cowgill2019economics} discuss the case of computer scientists at Amazon who developed powerful new technology to screen resumes and discovered the algorithm placed a negative coefficient on terms associated with women and as a result appeared to be amplifying and entrenching male dominance in the technology industry. Another example of gender discrimination is discussed by \cite{datta2015automated} who studied AdFisher, an automated tool that explores how user behaviors, Google's ads, and Ad Settings interact. They found that setting the gender to female resulted in fewer instances of an ad related to high-paying jobs than setting it to male.

Related to our empirical illustration, journalists at the news website ProPublica reported on a commercial software used by judges in Broward County, Florida, that helps to decide whether a person charged with a crime should be released from jail prior to the trial. They have found that the software tool called Correctional Offender Management Profiling for Alternative Sanctions (COMPAS) resulted in a disproportionate number of false positives for African American defendants who were classified as high risk but subsequently not charged with another crime. We apply our covariate-driven classification methods to pretrial decisions with an emphasis on algorithmic fairness and show how these biases can be eliminated naturally when the decision-maker incorporates preferences for false positive mistakes for protected groups.\footnote{It is also worth mentioning that human decision-makers can also be biased. For example, \cite{posner2025judge} reports that large language models tend to rely more on legal arguments than human judges who may be influenced by sympathy and emotions.}

Therefore, we make contributions to different strands of the literature. First, our paper is related to the general social planner setting studied by \cite{rambachan2020economic} as well as the special cases studied by \cite{corbett2017algorithmic} and \cite{kleinberg2018algorithmic}, among others. \cite{rambachan2020economic} argue that the algorithmic fairness constraints can be addressed using the tools of welfare economics, where a social planner has group-specific welfare weights. Our paper adopts the same point of view, but we take a step further showing how to map these welfare functions to the specific ML methods. Note that \cite{corbett2017algorithmic} and \cite{menon2018cost} suggest incorporating fairness constraints in ML algorithms which is equivalent to using the plug-in classifiers, where the fitted probability is compared to race-specific thresholds. \cite{agarwal2018reductions} and \cite{cotter2019two} consider the symmetric classifiers with non-convex fairness constraints which are challenging to use in practice. In contrast, our work focuses on convexified asymmetric classifiers that do not require fitting the conditional probabilities. \cite{kitagawa2023constrained} use convexification techniques similar to ours and show that the standard symmetric binary classification can incorporate various constraints with a Hinge convexifying function. In contrast, our work focuses on incorporating fairness through the loss function similarly to \cite{rambachan2020economic} which can be also done with the logistic and exponential convexifying functions.

Second, our paper is related to the literature on empirical utility and welfare maximization; see \cite{manski2004statistical}, \cite{hirano2009asymptotics}, \cite{bhattacharya2012inferring}, \cite{zhao2012estimating}, \cite{kitagawa2018should}, \cite{athey2021policy}, \cite{mbakop2021model}, \cite{kallus2021more}, \cite{viviano2023fair}, \cite{adjaho2022externally} among others. This literature usually relies on the potential outcomes framework and deals with continuous outcomes. In contrast to this work, we consider the unconstrained asymmetric binary classification. Our paper is also related to \cite{su2021model,su2021utility} who highlights connections between \cite{elliott2013predicting} and the cost-sensitive classification literature. \cite{su2021model} studies the model selection problem for the non-smooth non-convex empirical utility maximization of \cite{elliott2013predicting} while \cite{su2021utility} considers the utility-maximizing nearest neighborhood method. In contrast, our work proposes to solve a convexified problem and establishes the performance guarantee for asymmetric Logit, deep learning, LASSO, and boosting which has not been done in the aforementioned papers.

Regarding deep learning for economic applications, our paper is related to the recent work of \cite{farrell2018deep} who show that under weak assumptions the deep learning estimators of the regression function can achieve sufficiently fast convergence rates for semiparametric inference and establish supporting non-asymptotic results; see also \cite{chen2007sieves} for an earlier review of shallow neural networks and \cite{dell2024deep} for economic applications of deep learning.

Lastly, we also contribute to classification literature in statistics and computer science; see \cite{zhang2004statistical}, \cite{bartlett2006convexity}, \cite{boucheron2005theory}, \cite{koltchinskii2011oracle}, for the theory of symmetric binary classification. The classification with asymmetric losses is also known as cost-sensitive or example-dependent classification and has been developing largely independently from the related econometric literature, e.g.\ see \cite{elkan2001foundations}. While this literature is quite large, few formal results exist justifying proposed methods. For example,  \cite{bahnsen2014example} introduces weights in logistic regression heuristically without providing formal justifications that the weights actually correspond to the asymmetric decision problem. Several papers were written subsequently, see \cite{bahnsen2015example}, \cite{Mehta2020deepcatch} and \cite{xia2017cost} among others, mostly with practical computer science applications such as fraud detection, e-commerce, or marketing. There is also a body of literature on classification with imbalanced or corrupted classes, where weighting is also done heuristically to improve the accuracy of a classifier, which can be highly sensitive to over-represented and/or noisy classes; see \cite{ting1998inducing} and \cite{sun2007cost} among others. This literature considers weights that are different from ours, e.g., based on the empirical number of instances in a given class, or proportional to the inverse of expected costs that are simulated with the importance or rejection sampling. \cite{scott2011calibrated,scott2012calibrated}, and \cite{bao2020calibrated} are examples that attempt to deal rigorously with several instances of asymmetric classification problems. Neither of these papers address the problem (\ref{eq:decision_problem}) in full generality studied in our paper. For example, \cite{scott2012calibrated} covers a special case of our problem with \textit{covariate-independent} false positive/negative losses while \cite{scott2011calibrated} focuses on predicting a sign of a real-valued random variable. Moreover, these papers do not provide specific estimators and do not establish the supporting theoretical results. In contrast, our paper shows how the weights coming from the asymmetric loss function can be incorporated in the cost-sensitive asymmetric Logit, deep learning, boosting or LASSO.

The paper is organized as follows. Section \ref{sec:predicting} introduces notation, describes the binary decision problem in terms of risk/payoff, and illustrates a convenient characterization of the optimal binary decision. Examples of several asymmetric binary decision problems are also provided. Section \ref{sec:convexification} covers the main convexification theorem and provides examples of convexifying functions. In Section \ref{sec:riskbounds}, we provide the excess risk bounds for the binary decision rules constructed from the data and discuss several cases, including logistic regression, LASSO, and deep learning. Section~\ref{secapp:MCresults} presents Monte Carlo simulation results. Finally, we provide an empirical application to pretrial detention in Section \ref{sec:empirical} followed by conclusions. The proofs of key results appear in the Appendix, while the Online Appendix provides additional results for boosting, LASSO, and shallow learning as well as the detailed data description for the empirical application.

\section{Binary decisions \label{sec:predicting}}
Let $Y\in\{-1,1\}$ be the target variable and let $X\in\mathcal{X}\subset\R^d$ be covariates.\footnote{We use capitals for random variables and lowercase letters for realizations.} A measurable function $f:\mathcal{X}\to\{-1,1\}$ is called the binary decision/choice/prediction. The decision-making process amounts to minimizing a risk function that describes its consequences in different states of the world
\begin{equation*}
	\mathcal{R}(f) = \E_{X,Y}[\ell(f(X), Y,X)],
\end{equation*}
where $\ell:\{-1,1\}^2\times\mathcal{X}\to\R$ is a loss function specified by the decision-maker. Note that the decision $f$ can be random and the expectation is taken with respect to the distribution of $(Y,X)$ only. The loss function can be asymmetric and may also depend on the covariates $X$, which is economically a more realistic scenario faced by the decision-maker than the one provided by the standard classification setting.

\subsection{Some examples}

In the Introduction, we alluded to many examples in economic decision-making. To motivate the first example, we can think of a credit risk application, where with the false negative mistakes ($f(X)=-1$ and $Y=1$) the bank suffers a loss from the borrower's default, while with the false positive mistakes ($f(X)=1$ and $Y=-1$), the bank simply foregoes its potential earnings. Moreover, the size of the loan and other economic factors may determine the loss, see also \cite{lieli2010construction}:
\begin{example}[Lending decisions]\label{ex:lending}
	Let $y\in\{-1,1\}$ be the default status ($=1$ if defaulted) and $f\in\{-1,1\}$ be the lending decision ($=-1$ if approved). A lender's loss function is \begin{equation*}
		\ell(f,y,s,z) = L(s,z)\one_{f=-1,y=1} - \Pi(s,z)\one_{f=-1,y=-1},
	\end{equation*}
	where $L(s,z)$ and $\Pi(s,z)$ be the loss and the profit functions, $s\geq 0$ is the size of the loan, and $z\in\R^{d-1}$ some other economic factors. 
\end{example}

As noted in the Introduction, an important policy debate pertains to the fairness and the discrimination bias of machine learning algorithms towards, e.g., low income groups, gender, or race. The following example suggests that the unfair treatment of individuals can be eliminated if the social planner assigns group-specific weights in the loss function, see also \cite{rambachan2020economic}:
\begin{example}[Social planner with a disadvantaged group]\label{ex:fairness}
	Let $y\in\{-1,1\}$ be the recidivism status ($=1$ if recidivist) and $f\in\{-1,1\}$ be a binary bail decision ($=-1$ if kept in prison). The social planner's loss function is
	\begin{equation*}
		\ell(f,y,g) = \psi_g\one_{f=-1,y=1} + \varphi_g\one_{f=1,y=-1},
	\end{equation*}
	where $g\in\{0,1\}$ is a binary indicator of a disadvantaged group, $\psi_g,\varphi_g>0$ are group-specific weights for false negative and false positive mistakes.\footnote{In criminal justice, the algorithm may be perceived as unfair when the false positive rates for African American and white defendants are different. For instance, \cite{larson2016surya} report that the widely used COMPAS software tends to make false positive predictions of the recidivism for African American individuals more frequently compared to white individuals.}
\end{example}
Computer scientists have focused on defining fairness-aware algorithms by imposing restrictions on $f$.\footnote{Alternatively, one could fit two different classifiers for two groups and use the plug-in approach with group-specific cut-off. This is not the route taken in our paper. We provide some comparisons to plug-in classifiers in simulations.} Often a formal criterion of fairness is defined, and a decision rule is developed to satisfy the criterion, e.g., the statistical parity across groups. We follow the economist's arguments that fairness is naturally defined through welfare/losses. 

In the remainder of the paper, we will continue to work with settings involving a generic vector of covariates $X\in\R^d$ which may contain the binary group membership variable $G\in\{0,1\}$ and some other covariates $Z\in\R^{d-1}$. The point worth keeping in mind, however, is that our framework covers much discussed preference-based notions of fairness characterized by general covariate-driven loss functions.

\subsection{Optimal binary decision}
Following the decision-theoretic perspective, we define the optimal binary decision $f^*:\mathcal{X}\to\{-1,1\}$ as a solution to
\begin{equation*}
	\inf_{f:\mathcal{X}\to\{-1,1\}}\E[\ell(f(X), Y,X)],
\end{equation*}
where the minimization is done over all measurable functions $f:\mathcal{X}\to\{-1,1\}$. We also use $\mathcal{R}^*=\mathcal{R}(f^*)$ to denote the corresponding smallest achievable risk. Note that this framework also covers the utility/welfare maximization problems, in which case, we can define $\ell(f(X), Y,X)=-U(f(X),Y,X)$, where $U(f(x),y,x)$ is a utility/welfare function of a binary decision $f(x)$ when the outcome is $Y=y$ and covariates are $X=x$.

The following proposition provides an alternative convenient for us characterization of the optimal binary decision, also known as the Bayes' decision; see Appendix for a proof.

\begin{proposition}\label{prop:risk_representation}
	The optimal binary decision $f^*$ solves
	\begin{equation*}
		\inf_{f:\mathcal{X}\to \{-1,1\}}\E\left[\omega(Y,X)\one_{-Yf(X)\geq 0}\right]
	\end{equation*}
	with $\omega(Y,X)\triangleq Ya(X)+b(X)$, $a(x) = \ell_{-1,1}(x)-\ell_{1,1}(x) + \ell_{-1,-1}(x) - \ell_{1,-1}(x)$, $b(x) = \ell_{-1,1}(x) - \ell_{1,1}(x) + \ell_{1,-1}(x) - \ell_{-1,-1}(x)$, and $\ell_{f,y}(x)\triangleq\ell(f,y,x)$.
\end{proposition}
Proposition~\ref{prop:risk_representation} is related to \cite{elliott2013predicting}, but in contrast to their work, it recasts the asymmetric classification problem in a way that is amenable to convexifiction. Note that $a$ can be interpreted as a difference of net losses $\ell_{-1,1}-\ell_{1,1}$ and $\ell_{1,-1}-\ell_{-1,-1}$ from wrong decisions when $Y=1$ and $Y=-1$ respectively. Similarly, $b$ can be interpreted as a sum of two net losses. Proposition~\ref{prop:risk_representation} shows that the optimal binary decision minimizes the objective function involving the indicator function $z\mapsto \one_{z\geq 0}$, which is discontinuous and not convex. This leads to a difficult NP-hard empirical risk minimization problem
\begin{equation*}
	\inf_{f:\mathcal{X}\to\{-1,1\}}\frac{1}{n}\sum_{i=1}^n\omega(Y_i,X_i)\one_{-Y_if(X_i)\geq 0 }.
\end{equation*}

\section{Convexification \label{sec:convexification}}
\subsection{Convexified excess risk}
The purpose of this section is to convexify the binary decision problem with a generic loss function making it amenable to modern ML algorithms. It is easy to see that the risk function of a binary decision rule $f:\mathcal{X}\to\{-1,1\}$ can be written as
\begin{equation*}
	\mathcal{R}(f) = 0.5\E\left[\omega(Y,X)\one_{-Yf(X)\geq 0}\right] + \E[d(Y,X)]
\end{equation*}
with $d(y,x) = 0.25(\ell_{1,1}(x) + \ell_{-1,1}(x))(1+y) + 0.25(\ell_{1,-1}(x) + \ell_{-1,-1}(x))(1-y) - 0.25\omega(y,x)$; see the proof of Proposition 2.1 in the Appendix. Replacing the indicator function with a convex function $\phi:\R\to\R$, we obtain the convexified risk
\begin{equation*}
	\begin{aligned}
		\mathcal{R}_\phi(f) & = 0.5\E[\omega(Y,X)\phi(-Yf(X))] + \E[d(Y,X)].
	\end{aligned}
\end{equation*}
Therefore, minimizing the convexified risk amounts to solving
\begin{equation*}
	\inf_{f:\mathcal{X}\to\R}\E[\omega(Y,X)\phi(-Yf(X))].
\end{equation*}
Let $f^*_\phi$ be a solution to the convexified problem and let $\mathcal{R}^*_\phi = \inf_{f:\mathcal{X}\to\R}\mathcal{R}_\phi(f)$ be the optimal convexified risk. Next, we can define the excess convexified risk of a decision $f:\mathcal{X}\to\{-1,1\}$ as
\begin{equation}\label{eq:excess_convex}
	\mathcal{R}_\phi(f) - \mathcal{R}_\phi^* = 0.5\E\left[\omega(Y,X)\left(\phi(-Yf(X)) - \phi(-Yf^*_\phi(X))\right)\right].
\end{equation}

The excess convexified risk measures the deviation of the convexified risk of a given decision rule $f:\mathcal{X}\to\{-1,1\}$ from the optimal convexified risk and can be controlled. Unfortunately, the convexified excess risk tells us little about the performance of the binary decision in terms of the actual excess risk that the decision-maker cares about, namely:
\begin{equation*}
	\mathcal{R}(f) - \mathcal{R}^* = 0.5\E\left[\omega(Y,X)\left(\one_{-Yf(X)\geq 0} - \one_{-Yf^*(X)\geq 0}\right)\right].
\end{equation*}
In the following subsection, we show that the excess risk is bounded from above by the convexified excess risk, a result that we refer to as the main convexification theorem.

\subsection{Assumptions and main convexification theorem}
Let $\eta(x)=\Pr(Y=1|X=x)$ be the conditional probability. We impose the following assumption on the conditional probability $\eta$ and the loss function:
\begin{assumption}\label{as:losses}
	(i) $\ell_{-1,1}(x)-\ell_{1,1}(x)\geq c_b$ and $\ell_{1,-1}(x)- \ell_{-1,-1}(x)\geq c_b$ a.s.\ over $x\in\mathcal{X}$ for some $c_b>0$; (ii) there exist $\epsilon>0$ such that $\epsilon\leq \eta(x)\leq 1-\epsilon$  a.s.\ over $x\in\mathcal{X}$; (iii) there exists $M<\infty$ such that $|\ell_{f,y}(x)|\leq M$ a.s.\ over $x\in\mathcal{X}$ for all $f,y\in\{-1,1\}$, where $\ell_{f,y}(x)=\ell(f,y,x)$.
\end{assumption}
Assumption~\ref{as:losses} (i) requires that the losses from wrong decisions exceed those of correct decisions. It ensures that weights are $\omega(Y,X)\geq 0$ a.s., so that the convexification is possible (ii) requires that $\eta(x)$ is non-degenerate for almost all states of the world $x\in\mathcal{X}$ and is often imposed in econometrics literature. (iii) requires that the loss function is bounded and is satisfied in our empirical application.

\noindent Next, we restrict the class of convexifying function $\phi$.
\begin{assumption}\label{as:phi}
	(i) $\phi:\R\to[0,\infty)$ is a convex and nondecreasing function with $\phi(0)=1$; (ii) there exists some $L<\infty$ such that $|\phi(z)-\phi(z')|\leq L|z-z'|$ for all $z,z'$; (iii) there exist $C>0$ and $\gamma\in(0,1]$ such that for all $x,c\in(0,1)$, $\left|x - c\right|\leq C\left(x+c-2xc-\inf_{y\in\R}Q_c(x,y)\right)^\gamma$, where $Q_c(x,y) = x(1-c)\phi(-y)+ (1-x)c\phi(y),x,y\in\R$.	
\end{assumption}
Assumption~\ref{as:phi} (i) and (ii) are relatively mild restrictions imposed on the convexifying function. (iii) is a technical condition which we verify for the logistic, exponential, and hinge convexifications; see Lemmas~\ref{lemma:logistic}, \ref{lemma:hinge}, and \ref{lemma:exponential} in the Online Appendix. We focus on these three convexifications since they cover the majority of ML algorithms used in practice, including the (penalized) logistic regression, boosting, deep learning, and support vector machines. Note also that Assumption~\ref{as:phi} does not impose any restrictions on the loss function $\ell$.

\noindent We will also use the following covariate-driven threshold rule defined by the asymmetry of the loss function $\ell:$
\begin{equation*}
	\label{eq:c-function}
	c(x)= \frac{\ell_{1,-1}(x) - \ell_{-1,-1}(x)}{\ell_{-1,1}(x) - \ell_{1,1}(x) +\ell_{1,-1}(x) - \ell_{-1,-1}(x)},
\end{equation*}
where $\ell_{f,y}(x)=\ell(f,y,x)$. 

For $z\in\R$, put $\sign(z)=\one_{z\geq 0} - \one_{z<0}$. Our first result establishes the link between the optimal decision $f^*$ and the solution to the convexified risk minimization problem $f^*_\phi.$ 
\begin{proposition}\label{prop:bayes}
	Suppose that Assumption~\ref{as:losses} (i) is satisfied. Then the optimal decision is $f^*(x)=\sign(\eta(x)-c(x))$. Moreover, under Assumptions~\ref{as:losses} (ii) and \ref{as:phi} (i) if $\phi$ is differentiable,\footnote{The result also holds for the nondifferentiable hinge function, $\phi(z)=\max\{1+z,0\}$.} then $\sign(f^*_\phi(x))=f^*(x)$.
\end{proposition}

\noindent The optimal decision rule $f^*(x)=\sign(\eta(x)-c(x))$ is well known, see \cite{elliott2013predicting} and references therein.\footnote{For completeness of presentation, we provide a concise proof in the Online Appendix.} More importantly, Proposition~\ref{prop:bayes} shows that the optimal decision rule for the convexified problem that can be easily solved in practice coincides with $f^*$. The threshold $c(x)$ corresponds to the fraction of net losses from false positives in the total net losses, hence, the decision $f(x)=1$ is made whenever the choice probability $\Pr(Y=1|X=x)$ exceeds the fraction of false positive losses. Note that in the symmetric binary classification case $\ell_{1,-1}(x)=\ell_{-1,1}(x)=1$ and $\ell_{1,1}(x)=\ell_{-1,-1}(x)=0$, so that $c(x)=1/2$ and the optimal decision is $f^*(x)=1$ if $\Pr(Y=1|X=x)$ exceeds $\Pr(Y=0|X=x)$ and $f^*(x)=-1$ otherwise. Interestingly, the optimal decision rules depend on $\eta(x)=\Pr(Y=1|X=x)$, which is related to the heteroskedasticity $\Var(Y|X=x)=4\eta(x)(1-\eta(x))$ and all other moments; cf. \cite{diebold1996optimal,diebold1997optimal}.

It is worth mentioning that our starting point is a fixed decision problem described by the loss function $\ell$, which in conjunction with $\eta$ determines the boundary separating the two decisions $\{x\in\mathcal{X}:\; \eta(x)-c(x) = 0 \}$. The following condition generalizes the so-called Tsybakov noise or margin condition, see \cite{boucheron2005theory}, and quantifies the complexity of the decision problem.
\begin{assumption}\label{as:tsybakov}
	Suppose that there exist $\alpha\geq 0$ and $C_m>0$ such that
	\begin{equation*}
		P_X\left(\{x:\; |\eta(x)-c(x)|\leq u\}\right)\leq C_mu^\alpha,\qquad \forall u>0,
	\end{equation*} 
	where $P_X$ is the distribution of $X$.
\end{assumption}
Assumption~\ref{as:tsybakov} describes the complexity of the classification problem around the decision boundary $\{x:\eta(x)=c(x) \}$. If $\alpha=0$, then it does not impose any restrictions as we can always take $C_m=1$, while larger values of $\alpha$ are more advantageous for classification. In the special case of the symmetric binary classification, $c(x)=1/2$, and the extreme case of $\alpha=\infty$ corresponds to $\eta(x)=\Pr(Y=1|X=x)$ being bounded away from $1/2$ which roughly means that there are no observations close to the decision boundary. In this case, we can interpret this as a perfect separation between recidivists and non-recidivists in the space of covariates. We refer to \cite{boucheron2005theory}, Section 5.2, for additional discussion of the margin condition as well as for equivalent formulations in the special case of symmetric binary classification; see also \cite{ponomarev2024lower}.

Our next result relates the convexified excess risk to the excess risk of the binary decision problem a generic loss function under the margin condition.

\begin{theorem}\label{thm:convexified_risk_bound_margin}
	Suppose that Assumptions~\ref{as:losses}, \ref{as:phi} (i) and (iii), and \ref{as:tsybakov} are satisfied. Then there exists $C_\phi<\infty$ such that for every measurable function $f:\mathcal{X}\to\R$
	\begin{equation*}
		\mathcal{R}(\sign(f)) - \mathcal{R}^* \leq C_\phi\left[\mathcal{R}_\phi(f) - \mathcal{R}_\phi^*\right]^\frac{\gamma(\alpha + 1)}{\gamma\alpha + 1}.
	\end{equation*}
\end{theorem}
\noindent The formal proof of this result with the explicit expression for the constant $C_\phi$ appears in the Appendix. It is worth noting that our result nests the symmetric binary classification as a special case; see e.g.\ \cite{zhang2004statistical}, Theorem 2.1. Theorem~\ref{thm:convexified_risk_bound_margin} relates the excess risk of the sign of $f$ to the convexified risk of $f$ (note that the convexified risk is well-defined for arbitrary $f:\mathcal{X}\to\R$). Therefore, instead of solving the NP-hard empirical risk minimization problem, we can solve the following weighted classification problem
\begin{equation*}
	\inf_{f\in\mathscr{F}}\frac{1}{n}\sum_{i=1}^n\omega(Y_i,X_i)\phi(-Y_if(X_i)),
\end{equation*}
where $\mathscr{F}$ is a class of measurable functions $f:\mathcal{X}\to[-1,1]$, called a soft decision. Note that we restrict the range of the soft decision to $[-1,1]$ since the risk is determined by the sign of $f$ only.

\subsection{Examples}\label{sec:examples}
We consider three examples of convexifying functions widely used in empirical applications. They are the logistic, exponential, and hinge functions. These convexifying functions appear in various ML algorithms including boosting, support vector machines, neural networks, and simple logistic regression; see Figure~\ref{fig:convexifications}. Applying Theorem \ref{thm:convexified_risk_bound_margin} allows us to implement suitably reweighted versions.

\begin{figure}
	\centering
	\includegraphics[width=0.5\textwidth]{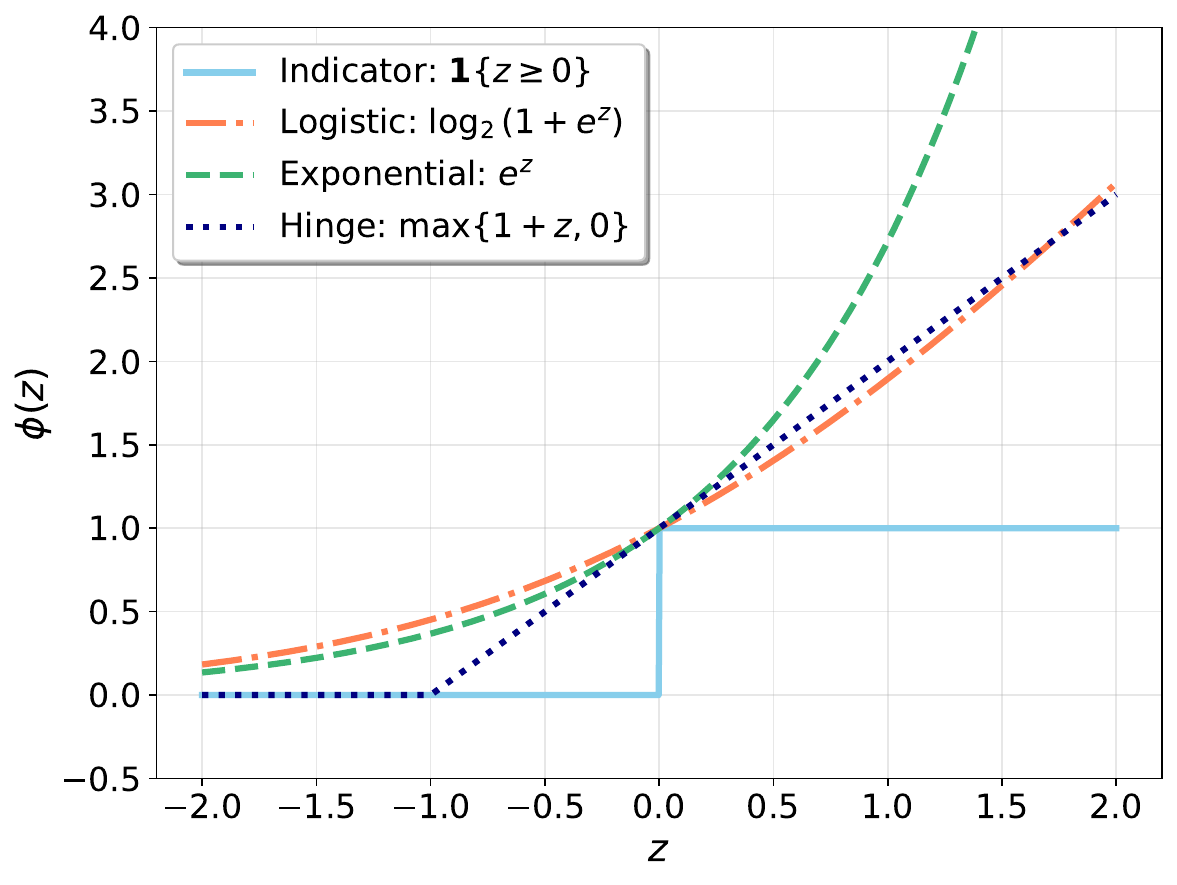}
	\caption{Indicator function with convexifications corresponding to logistic regression (logistic), boosting (exponential), and support vector machines (hinge).}
	\label{fig:convexifications}
\end{figure}

\begin{example}[Logistic convexification]
	Lemma~\ref{lemma:logistic} in the Online Appendix shows that the function $\phi(z)$ = $\log_2(1 + e^{z})$ satisfies Assumption~\ref{as:phi} with $\gamma=1/2$ and $C=\sqrt{2\log 2}$. The convexified objective is to minimize
	\begin{equation}\label{eq:logit}
		f\mapsto \frac{1}{n}\sum_{i=1}^n\omega(Y_i,X_i)\log_2\left(1 + e^{-Y_if(X_i)}\right).
	\end{equation}
\end{example}

\noindent For the symmetric loss function and linear function class, $f(X_i) = X_i^\top\theta$, the problem in equation~(\ref{eq:logit}) becomes the standard logistic regression. The deep learning consists of minimizing the problem in equation (\ref{eq:logit}) over the class of deep neural networks. Lastly, the state-of-the-art Extreme Gradient Boosting (XGBoost) algorithm of \cite{chen2016xgboost} can be understood as a functional gradient descent algorithm minimizing the problem in equation~(\ref{eq:logit}) with some additional regularization; see also \cite{friedman2001elements}, Ch. 10. 

\begin{example}[Exponential convexification]
	Lemma~\ref{lemma:exponential} in the Online Appendix shows that $\phi(z)$ =  $\exp(z),$ the exponential convexifying function,  satisfies Assumption~\ref{as:phi} on compact intervals with $\gamma = 1/2$ and $C = {2}$. The objective is to minimize 
	$$f\mapsto \frac{1}{n}\sum_{i=1}^n\omega(Y_i,X_i)e^{-Y_if(X_i)}.$$
\end{example}
For the symmetric loss function, this objective function is usually minimized iteratively in the adaptive boosting (AdaBoost); see also \cite{friedman2001elements}, Ch. 10.

\begin{example}[Hinge convexification]
	Lemma~\ref{lemma:hinge} in the Online Appendix shows that $\phi(z)$ = $(1+z)_+$,\footnote{For $a\in\mathbb{R}$, $(a)_+ = \max\{a,0\}$.} the hinge convexifying function, satisfies Assumption~\ref{as:phi} with $\gamma=1$ and $C=1$. The objective is to minimize $$f\mapsto \frac{1}{n}\sum_{i=1}^n\omega(Y_i,X_i)(1-Y_if(X_i))_+.$$
\end{example}
For the symmetric loss function, this objective function is usually minimized subject to a reproducing kernel Hilbert space constraint within support vector machines; see \cite{friedman2001elements}, Ch. 12.

\section{Excess risk bounds \label{sec:riskbounds}}
The convexified empirical risk minimization problem consists of minimizing the empirical risk
\begin{equation*}
	\widehat{\mathcal{R}}_{\phi}(f) = \frac{1}{n}\sum_{i=1}^n\omega(Y_i,X_i)\phi(-Y_if(X_i))
\end{equation*}
over some class of functions $f:\mathcal{X}\to[-1,1]$, denoted $\mathscr{F}_n$. Let $\hat f_n$ be a solution to $\inf_{f\in\mathscr{F}_n}\widehat{\mathcal{R}}_{\phi}(f)$, and let $f_n^*$ be a solution to $\inf_{f\in\mathscr{F}_n}{\mathcal{R}}_{\phi}(f)$. Put also $\|f\|_q = (\E|f(X)|^q)^{1/q}$ with $q\geq 1$ and $\|.\|_2=\|.\|$.
The following assumption requires that the risk function is sufficiently curved around the minimizer.
\begin{assumption}\label{as:loss_modulus}
	There exist some $c_\phi>0$ and $\kappa\geq 1$ such that for every $f\in\mathscr{F}_n$
	\begin{equation*}
		\mathcal{R}_\phi(f) - \mathcal{R}_\phi^* \geq c_\phi\|f-f^*_\phi\|^{2\kappa}.
	\end{equation*}
\end{assumption}
\noindent It is worth mentioning that in light of Assumption~\ref{as:losses} imposed on the loss function $\ell$, Assumption~\ref{as:loss_modulus} is mainly about the curvature of the convexifying function $\phi$. In particular, it is typically satisfied with $\kappa=1$ when $\phi''$ is strictly positive, which is the case for the logistic and exponential convexifications; see Lemma~\ref{lemma:curvature_exp_log} in the Online Appendix. On the other hand, it holds with $\kappa=1 + 1/\alpha$ in the hinge case, where $\alpha$ is the margin parameter from the Assumption~\ref{as:tsybakov}; see Lemma~\ref{lemma:curvature_hinge}. It is worth mentioning that the curvature condition in Assumption~\ref{as:loss_modulus} is typically used to establish bounds on $\|f-f_\phi^*\|$ and may not be needed to establish the \textit{slow} convergence rates of the excess risk. However, this condition is typically used to obtain the \textit{fast} convergence rates for the excess risk in the symmetric classification and regression cases; see \cite{koltchinskii2011oracle}.

We first state the oracle inequality for the excess risk in terms of a fixed point of the \textit{local Rademacher complexity} of the class $\mathscr{F}_n$ defined as
\begin{equation*}
	\psi_n(\delta;\mathscr{F}_n) \triangleq \E\left[\sup_{f\in\mathscr{F}_n:\|f - f^*_n\|^2\leq\delta}|R_n(f - f^*_n)|\right],
\end{equation*}
where $R_n(f-f_n^*) = \frac{1}{n}\sum_{i=1}^n\varepsilon_i(f(X_i)-f_n^*(X_i))$ is a Rademacher process, i.e., $(\varepsilon_i)_{i=1}^n$ are i.i.d. in $\{-1,1\}$ with probabilities $1/2$. An attractive feature of this complexity measure is that it only depends on the local complexity of the parameter space in the neighborhood of the minimizer and provides a sharp description of the learning problem; see \cite{koltchinskii2011oracle}. At the same time, the local Rademacher complexities are general enough to provide a unified theoretical treatment for different methods and can be used to deduce oracle inequalities in many interesting examples. For a function $\psi:\R_+\to\R_+$, put $\psi^\flat(\sigma) = \sup_{\delta\geq\sigma} [\psi(\delta)/\delta]$ and for a constant $\kappa\geq 1$, put $\psi_\kappa^{\sharp}(\epsilon) = \inf\left\{\sigma>0:\;\sigma^{1/\kappa-1}\psi^\flat(\sigma^{1/\kappa})\leq \epsilon \right\}$. The transform $\psi^\sharp_{\kappa}$ is a generalization of the $\sharp$-transform considered in \cite{koltchinskii2011oracle} and describes the fixed point of the local Rademacher complexity in our setting. The following result holds:\footnote{Proofs for all results in the Appendix.}

\begin{theorem}\label{thm:oracle inequality}
	Suppose that Assumptions~\ref{as:losses}, \ref{as:phi}, \ref{as:tsybakov}, and \ref{as:loss_modulus} are satisfied and  $(Y_i,X_i)_{i=1}^n$ is an i.i.d. sample. Then there exist constants $c>0,\epsilon>0$ such that for every $t>0$ with probability at least $1-ce^{-t}$
	\begin{equation*}
		\mathcal{R}(\sign(\hat f_n)) - \mathcal{R}^* \lesssim \left[\psi_{n,\kappa}^{\sharp}(\epsilon) + \left(\frac{t}{n}\right)^\frac{\kappa}{2\kappa - 1} + \frac{t}{n} + \inf_{f\in\mathscr{F}_n}{\mathcal{R}}_{\phi}(f) - \mathcal{R}_\phi^*\right]^\frac{\gamma(\alpha + 1)}{\gamma\alpha+1}.
	\end{equation*}
\end{theorem}
\noindent The proof of this result appears in the Appendix and provides the explicit expression for all constants. Theorem~\ref{thm:oracle inequality} tells us that the accuracy of the binary decision $\sign(\hat f_n)$ depends on the fixed point of the local Rademacher complexity of the class, $\psi^\sharp_{n,\kappa}$, and the approximation error to the convexified risk of the optimal decision. In the remaining part of this section, we will specialize this result to parametric decisions and deep neural networks. In each case it can be seen that for a fixed $\epsilon$, we have $\psi_n^\sharp(\varepsilon)\to 0$ as long as $n\to\infty$ and the complexity of the class $\mathcal{F}_n$ does not increase ``too fast". Therefore, Theorem~\ref{thm:oracle inequality} illustrates the variance/bias trade-off, i.e. more complex classes of decision rules reduce the approximation error (bias) at costs of increasing the Rademacher complexity (variance) and vice versa. Note that in the special case when $\kappa=1$, Theorem~\ref{thm:oracle inequality} recovers \cite{koltchinskii2011oracle}, Proposition 4.1. In the following subsections, we illustrate this result for parametric and nonparametric binary decision rules. Additional results on boosting and LASSO are available in the Online Appendix.

\subsection{Parametric decisions and logistic regression \label{sec:param}}
We start with illustrating our risk bounds for parametric binary decision rules. The decision rule is defined $\sign(f_{\hat\theta})$ with $\hat\theta$ solving
\begin{equation*}
	\inf_{\theta\in\Theta}\frac{1}{n}\sum_{i=1}^n\omega(Y_i,X_i)\phi(-Y_if_\theta(X_i)),
\end{equation*}
where $f_\theta(x)=\sum_{j=1}^p\theta_j\varphi_j(x)$, $\theta\in\Theta\subset\R^p$, and $(\varphi_j)_{j\geq 1}$ is a collection of functions in $L_2(P_X)$, called the dictionary. This covers the linear functions $f_\theta(x)=x^\top\theta$ as well as nonlinear functions of covariates provided that the dictionary $(\varphi_j)_{j\geq 1}$ contains nonlinear transformations. The most popular choice of convexifying function is the logistic function $\phi(z)=\log_2(1+e^z)$. More generally, we have the following result for any convexifying function satisfying Assumption~\ref{as:phi}.

\begin{theorem}\label{thm:parametric_predictions}
	Under assumptions of Theorem~\ref{thm:oracle inequality}
	\begin{equation*}
		\E\left[\mathcal{R}(\sign(\hat f_n)) - \mathcal{R}^*\right] \lesssim \left[\left(\frac{p}{n}\right)^\frac{\kappa}{2\kappa - 1} + \inf_{f\in\mathscr{F}_n}{\mathcal{R}}_{\phi}(f) - \mathcal{R}_\phi^*\right]^\frac{\gamma(\alpha+1)}{\gamma\alpha+1}.
	\end{equation*}
\end{theorem}
\noindent It follows from Lemmas~\ref{lemma:logistic}, \ref{lemma:hinge}, \ref{lemma:exponential}, \ref{lemma:curvature_exp_log}, and \ref{lemma:curvature_hinge} that for the logistic and the exponential functions $\gamma=1/2$ and $\kappa=1$ while for the hinge function $\gamma=1$ and $\kappa=1+1/\alpha$. Therefore, in all three cases, for parametric decisions we obtain
\begin{equation*}
	\E[\mathcal{R}(\sign(\hat f_n)) - \mathcal{R}^*] \lesssim \left(\frac{p}{n}\right)^\frac{1+\alpha}{2+\alpha} + \left[\inf_{f\in\mathscr{F}_n}{\mathcal{R}}_{\phi}(f) - \mathcal{R}_\phi^*\right]^\frac{\gamma(\alpha+1)}{\gamma\alpha+1},
\end{equation*}
uniformly over a set of distributions restricted in Theorem~\ref{thm:parametric_predictions}. For a fixed $p$ the convergence rate of the first term can be anywhere between $O(n^{-1/2})$ and $O(n^{-1})$ depending on margin parameter $\alpha$. Since this rate is independent of the convexification, one can use logistic regression reweighted for the asymmetries of the loss function when the signal-to-noise ratio is low and the approximation error is dominated by the first term. Besides simplicity, this choice is also attractive because: 1) the objective function is differentiable; 2) it recovers logistic MLE in the symmetric case; 3) it has a slightly better constant in Theorem~\ref{thm:convexified_risk_bound_margin} than the exponential function. However, from the nonparametric point of view, the choice of the convexifying function is more subtle and the hinge convexification can lead to a better rate of the approximation error as we shall see in the following section.

In the high-dimensional case, when $p$ can be large relative to $n$, we can consider the weighted empirical risk minimization problem with the LASSO penalty
\begin{equation*}
	\min_{\theta\in\Theta}\frac{1}{n}\sum_{i=1}^n\omega(Y_i,X_i)\phi(-Y_if_\theta(X_i)) + \lambda_n|\theta|_1,
\end{equation*}
where $|.|_1$ is the $\ell_1$ norm and $\lambda_n\downarrow 0$ is a tuning parameter. We can deduce from the Online Appendix, Theorem~\ref{thm:lasso} that for the hinge, logistic, and exponential convexifications with probability at least $1-\delta$
\begin{equation*}
	\mathcal{R}(\sign(f_{\hat\theta})) - \mathcal{R}^* \lesssim \left(\frac{s\log(2p)}{n} + \frac{s\log(1/\delta)}{n} \right)^{\frac{\alpha+1}{\alpha+2}} + \left[\mathcal{R}_\phi(f_{\theta^*})-\mathcal{R}_\phi^*
	\right]^\frac{\gamma(\alpha+1)}{\gamma\alpha+1}.
\end{equation*}
where $s$ is the number of non-zero coefficients in a certain oracle vector $\theta^*\in\R^p$. According to this bound: 1) the dimension $p$ may increase exponentially with the sample size if $s$ is small; 2) we can use the logistic convexification in the low signal-to-noise settings, where the approximation error is relatively small.

\subsection{Deep learning \label{sec:deep}}
In this final subsection, we discuss how to construct the asymmetric deep learning architecture and discuss the corresponding risk bounds; see \cite{farrell2018deep} and references therein for more details. The deep learning amounts to fitting a neural network with several hidden layers, also known as a deep neural network. Fitting a deep neural network requires choosing an activation function $\sigma:\R\to\R$ and a network architecture. We focus on the ReLU activation function, $\sigma(z) = \max\{z,0\}$, which is the most popular choice for deep networks.\footnote{Other activation functions used in the deep learning include: leaky ReLU, $\sigma(z)=\max\{\alpha z,0\},\alpha>0$; exponential linear unit (ELU), $\sigma(z) = \alpha(e^z - 1)\one_{z<0} + z\one_{z\geq 0}$; and scaled ELU.}

The architecture consists of $d$ neurons corresponding to covariates $X=(X_1,\dots,X_d)\in\R^d$, one output neuron corresponding to the soft prediction $f\in[-1,1]$, and a number of hidden neurons. The final decision is obtained with $\sign(f)\in\{-1,1\}$. Hidden neurons are grouped in $L$ layers, known as the \textit{depth} of the network. A hidden neuron $j\geq 1$ in a layer $l\geq 1$ operates as $z\mapsto \sigma(z^\top a_j^{(l)} + b_j^{(l)})$, where $z$ is the output of neurons from the layer $l-1$ and $a_j^{(l)},b_j^{(l)}$ are free parameters. The last layer and the output neuron produce together $z\mapsto \sigma(z + b^{(L)}c(x) + 1) - \sigma(z + b^{(L)}c(x) - 1) - 1\in\R$, 
where $b^{(L)}$ is the parameter to be estimated and $c(x)$ is a known decision cut-off function. The network architecture $(L,\mathbf{w})$ is described by the number of hidden layers $L$ and a \textit{width} vector $\mathbf{w}=(w_1,\dots,w_{L})$, where $w_{l}$ denotes the number of hidden neurons at a layer $l=1,2,\dots,L$. For completeness, put also $w_0=d$ and $w_{L+1}=2$. Our final deep learning architecture $(L,\mathbf{w})$ is
\begin{equation*}
	\begin{aligned}
		\mathscr{F}_n^{\rm DNN} & = \left\{x\mapsto \sigma(\theta(x)+c(x)d+1) - \sigma(\theta(x)+c(x)d-1)-1:\; |d|\leq n,\; \theta\in\Theta_n^{\rm DNN}\right\},
	\end{aligned}
\end{equation*}
where $\Theta_n^{\rm DNN} = \left\{\theta(x) = A_{L-1}\sigma_{\mathbf{b}_{L-1}}\circ \dots \circ A_1\sigma_{\mathbf{b}_1}\circ A_0x:\; \|\theta\|_\infty\leq F  \right\}$,
each $A_l$ is $w_{l+1}\times w_{l}$ matrix of network weights and for two vectors $y=(y_1,\dots,y_r)$ and $\mathbf{b}=(b_1,\dots,b_r)$ (a bias vector), we put $\sigma_{\mathbf{b}}\circ y = (\sigma(y_1 + b_1), \dots,\sigma(y_r + b_r))^\top$. Our deep learning architecture can be arranged on a graph presented in Figure~\ref{fig:network}. Note that the asymmetries of the loss function are incorporated in the neural network architecture via the orange neuron which is fed directly to the last layer consisting of 2 ReLU neurons. Interestingly, this additional connection can be obtained using the so-called ``residual connection" which is also often used in Transformer models, such as ChatGPT and AlphaFold, cf. \cite{targ2016resnet}.

\begin{figure}[htp]
	\centering
	\resizebox{0.7\textwidth}{!}{
		\begin{tikzpicture}[x=1.5cm, y=1.5cm]
			\tikzset{%
				neuron/.style={
					circle,
					minimum size=0.9cm},
			}
			\foreach \m [count=\y] in {1,2,3,4}
			\node [neuron, fill=teal!50] (input-\m) at (0,3-\y) {$X_\m$};
			
			\foreach \m [count=\y] in {1,2,3,4}
			\node [neuron, fill=blue!30] (hidden1-\m) at (2,3-\y) {$\sigma$};
			
			\foreach \m [count=\y] in {1,2,3}
			\node [neuron, fill=blue!30] (hidden2-\m) at (4,2.5-\y) {$\sigma$};
			
			\foreach \m [count=\y] in {1,2,3,4,5}
			\node [neuron, fill=blue!30] (hidden3-\m) at (6,3.5-\y) {$\sigma$};
			
			\foreach \m [count=\y] in {1,2}
			\node [neuron, fill=violet!40] (relu-\m) at (8,2-\y) {$\sigma$};
			
			\foreach \m [count=\y] in {1}
			\node [neuron, fill=red!60] (output-\m) at (10,1.5-\y) {$\hat f_n$};		
			
			\node [neuron, fill=orange!50] (c) at (2.2,-2.4) {$c$};
			
			\foreach \i in {1,2,3,4}
			\draw [gray, ->] (input-\i) -- (c);
			
			\foreach \i in {1,2,3,4}
			\foreach \j in {1,...,4}
			\draw [gray, ->] (input-\i) -- (hidden1-\j);
			
			\foreach \i in {1,2,3,4}
			\foreach \j in {1,2,3}
			\draw [gray, ->] (hidden1-\i) -- (hidden2-\j);
			
			\foreach \i in {1,2,3}
			\foreach \j in {1,...,5}
			\draw [gray, ->] (hidden2-\i) -- (hidden3-\j);
			
			\foreach \i in {1,...,5}
			\foreach \j in {1,2}
			\draw [gray, ->] (hidden3-\i) -- (relu-\j);
			
			\foreach \i in {1,2}
			\foreach \j in {1}
			\draw [gray, ->] (relu-\i) -- (output-\j);
			
			\foreach \i in {1,2}
			\draw [gray, ->] (c) -- (relu-\i);
			
			\foreach \l [count=\x from 0] in {Input, Layer 1, Layer 2, Layer 3, 2 ReLU, Output}
			\node [align=center, above] at (\x*2,3) {\l };
	\end{tikzpicture}}
	\caption{Directed graph of our deep learning architecture with $d=4$ covariates, $L=3$ hidden layers of width $\mathbf{w}=(4,3,5)$ neurons, and 2 outer ReLU neurons. The orange neuron takes covariates $X\in\R^d$ as an input and produces $c(X)\in\R$, which is fed directly into 2 ReLU neurons.}
	\label{fig:network}
\end{figure}

The following assumption restricts the smoothness of the conditional probability and imposes some assumptions on how the network architecture should scale with the sample size. For simplicity, we define the width of the network as the maximum width across all layers and denote it as $W_n = \max_{0\leq l\leq L}w_l$.

\begin{assumption}\label{as:deep_learning}
	(i) $\eta\in W^{\beta,\infty}_R[0,1]^d$ for some $R>0$ and $\beta\in\N$; (ii) the neural network architecture is such that the depth is $L_n \leq C_LK_n\log K_n$ and the width is $W_n \leq C_WJ_n\log J_n$, for some $C_L,C_W>0$ and some $J_n\leq n^a$ and $K_n\leq n^b$ with $J_nK_n\leq \left(n/\log^6n\right)^{d/(2\beta(2+\alpha) + 2d)}$ for some $a,b\geq 0$.
\end{assumption}
\noindent It is worth mentioning that we allow for neural networks, where the product of the width and the depth increase at the rate specified in Assumption~\ref{as:deep_learning} (ii). Let $\mathscr{F}_n^{\rm DNN}$ be a set of neural networks with the architecture satisfying Assumption~\ref{as:deep_learning}, where weights and biases $\{A_0,A_l,b_l,l=1,\dots,L \}$ are allowed to take arbitrary real values.

The soft deep learning decision $\hat f_n$ is a solution to the empirical risk minimization problem with the hinge convexification
\begin{equation*}
	\inf_{f \in \mathscr{F}_n^{\rm DNN}} \frac{1}{n}\sum_{i=1}^n\omega(Y_i,X_i)(1-Y_if(X_i))_+.
\end{equation*}
The following result holds for the binary decision estimated with the deep learning.

\begin{theorem}\label{thm:oracle inequality_dnn}
	Suppose that $(Y_i,X_i)_{i=1}^n$ is an i.i.d. sample from a distribution satisfying Assumptions~\ref{as:losses}, \ref{as:phi}, \ref{as:tsybakov}, and \ref{as:deep_learning} with fixed constants, and denoted $\mathcal{P}(\alpha,\beta)$. Then
	\begin{equation*}
		\sup_{P\in\mathcal{P}(\alpha,\beta)}\E_P\left[\mathcal{R}(\sign(\hat f_n)) - \mathcal{R}^*\right] \lesssim \left(\frac{\log^6n}{n}\right)^{\frac{(1+\alpha)\beta}{(2+\alpha)\beta + d}}.
	\end{equation*}
\end{theorem}
Note that in the special case of symmetric binary classification this result recovers the convergence rate recently obtained in \cite{kim2021fast}, Theorem 3.3, who in turn assume that the weights are bounded and allow for very slowly diverging depth of order $L_n = O(\log n)$. In particular, this rate matches the minimax lower bound up to the $\log n$ factor; see \cite{audibert2007fast}.

\section{Monte Carlo Simulations \label{secapp:MCresults}}

Anticipating the empirical application to the economic prediction of recidivism, we report on a simulation study pertaining to pretrial detention decisions. As explained in the next section, we present a design pertaining to judges who need to decide whether to release an offender, facing the possibility that the defendant might commit other crimes versus keeping in jail a defendant who would obey the law and the terms of the release. A first subsection describes the design and the second the simulation results.

\subsection{Simulation design \label{subsec:MCdesign}}

There are two groups, $G=0$ and 1, with one group assumed to be a protected segment of the population. We set $Y=-1$  if the person does not commit a crime upon pretrial release, and $Y=1$ otherwise. This is a much-studied topic and we approach it from a social planner point of view using a simplified stylized example with a loss function taken from Example~\ref{ex:fairness}:
\begin{center}
	\begin{tabular}{lcccccc}
		& \multicolumn{2}{c}{G = 0} & & & \multicolumn{2}{c}{G = 1} \\
		& & & & & & \\
		& $f(0,z) = 1$  & $f(0,z) = -1$  & & & $f(1,z) = 1$  &  $f(1,z) = -1$ \\
		& & & & & & \\
		$Y = 1$	&   $0$ &  $\psi_0$ & & &   $0$ &   $\psi_1$ \\
		& & & & & & \\
		$Y=-1$	&   $\varphi_0$ &  $0$ & & &    $\varphi_1$ & $0$ \\
	\end{tabular}	
\end{center}

\bigskip

\noindent where $z$ is a vector of observable characteristics and $\psi_g,\varphi_g>0$ for $g\in\{0,1\}$. From the above, we do not suffer any losses (or gains) when a defendant being released becomes a productive member of society or when a defendant who would commit another crime is kept in jail.

\medskip

Keeping someone in jail who does not intend to commit another crime comes with costs $\varphi_g$ depending on the group membership $g\in \{0,1\}$. If $\varphi_1>\varphi_0$, then the cost of keeping an individual in jail who does not intend to commit another crime is higher if that individual is in group $g=1$. Similarly, releasing a recidivist comes with costs $\psi_g,g\in\{0,1\}$, so that if $\psi_1\ne \psi_0$, the costs of releasing a recidivist are different for the protected group and everyone else.\footnote{It is worth mentioning that in credit risk applications, the bank might care more about the false negative mistakes (failing to predict defaults), while the social planner might be more concerned with the false positive mistakes (failing to predict that the loan will be repaid) and equal credit opportunities regardless of the group membership. Since our general framework can be applied to different binary prediction problems in economics, in this simulation study, we will look at how both mistakes change with $\varphi_g$ and $\psi_g$ for $g\in\{0,1\}$.}

\medskip

\noindent According to equation (\ref{eq:c-function}), the threshold between the two binary decisions is:
\begin{equation*}
	c(g,z) = \frac{\varphi_g}{\varphi_g + \psi_g},\qquad g\in\{0,1\}
\end{equation*}
and according to Proposition~\ref{prop:bayes} the optimal decision rule is
$f^*(g,z)$ = $\sign(\eta(g,z) - c(g,z))$ with $\eta(g,z)=\Pr(Y=1|G=g,Z=z)$.
Note also that $a(g,z)=\psi_g - \varphi_g$ and $b(g,z)=-(\psi_g + \varphi_g)$.
The design of the data generating process is
\begin{equation*}
	Y = 2\one_{G + Z^\top\gamma + \tau\left(\frac{1}{d}\sum_{j=1}^dZ_j^2 + 2Z_1\sum_{j=2}^dZ_j\right)  \geq \sigma\varepsilon } - 1,
\end{equation*}
where $\varepsilon\sim N(0,1)$, $G\sim \mathrm{Bernoulli}(\rho)$, and $Z_1,\dots,Z_d\sim_{i.i.d.} N(0,1)$. Therefore, the protected segment is a fraction $\rho$ of the population (determined by the Bernoulli distribution parameter) and the conditional probability is
\begin{equation}\label{eq:simrisk}
	\begin{aligned}
		\eta(g,z) & = \Pr(Y=1|G=g,Z=z) \\
		& = \Phi_\sigma\left(G + Z^\top\gamma + \tau\left(\frac{1}{d}\sum_{j=1}^dZ_j^2 + 2Z_1\sum_{j=2}^dZ_j\right)\right),
	\end{aligned}
\end{equation}
where $\Phi_\sigma$ is the CDF of $N(0,\sigma^2)$. Note that the DGP may feature a very simple example of nonlinearities with quadratic terms and interactions involving $Z_1.$ In addition, the DGP covers the linear model with $\tau = 0.$ Lastly, we also set $\gamma = (1,0.9,0.8,0,0,\dots,0)^\top$ and the dimension $d$ of covariates $Z$ is set to 15.

\medskip

Let $(Y_i,G_i,Z_i)_{i=1}^n$ be i.i.d.\ draws of $(Y,G,Z)$.  To evaluate the performance of our approach, we split the sample into the training or estimation sample $(Y_i,G_i,Z_i)_{i=1}^{n_e}$ and the test sample $(Y_i,G_i,Z_i)_{i=n_e+1}^n$. For parametric predictions, we estimate the decision rule solving the weighted logistic regression over the class of linear functions $\{(g,z)\mapsto \theta_0 + \theta_1g + z^\top\gamma:\;\theta_0,\theta_1\in\R,\;\gamma\in\R^{d} \}$. Note that according to our theory if $\tau=0$, then the weighted linear logistic regression provides valid binary predictions even if the choice probabilities are not logistic, cf.\ equation (\ref{eq:simrisk}). However, since in practice we typically do not know the parametric class that can capture all the relevant nonlinearities (i.e., that $\tau\ne 0$), we would often estimate the linear prediction rule
\begin{equation*}
	\min_{(c,\gamma)\in\R^{1+d}}\frac{1}{n_e}\sum_{i=1}^{n_e}(Y_i(\psi_{G_i}-\varphi_{G_i}) + (\psi_{G_i}+\varphi_{G_i}))\log\left(1+e^{-Y_i(\theta_0 + \theta_1G + Z_i^\top\gamma)}\right).
\end{equation*}
Then the estimated prediction rule is $(g,z)\mapsto \sign(\hat \theta_0 + \hat\theta_1g + z^\top\hat\gamma)$, where $(\hat \theta_0,\hat\theta_1,\hat\gamma)$ are estimated above, and the binary predictions evaluated on the test sample are
\begin{equation*}
	\sign(\hat \theta_0 + \hat\theta_1G_i + Z_i^\top\hat\gamma),\qquad i=n_e+1,\dots,n.
\end{equation*}
To obtain binary predictions with neural networks, we solve
\begin{equation*}
	\min_{f\in\mathscr{F}_n^{\rm NN}}\frac{1}{n_e}\sum_{i=1}^{n_e}(Y_i(\psi_{G_i}-\varphi_{G_i}) + (\psi_{G_i}+\varphi_{G_i}))(1-Y_if(G_i,Z_i))_+,
\end{equation*}
where $\mathscr{F}_n^{\rm NN}$ is a relevant neural network class, see Figure~\ref{fig:network} and \cite{targ2016resnet}. Then the estimated prediction rule is $z\mapsto \sign(\hat f(z))$, and the binary predictions evaluated on the test data are
\begin{equation*}
	\sign(\hat f(Z_i)),\qquad i=n_e+1,\dots,n.
\end{equation*}

We set the sample sizes $n$ = 1,000 and 10,000 with 30\% set aside as a test sample in the simulations. To benchmark our asymmetric binary choice approach, we focus first on the unweighted approach with logistic regression, gradient-boosted trees, shallow and deep learning, and support vector machines. For each method, we compute the group-specific false positive (FP) and false negative (FN) by mistakes estimating:
\begin{equation*}
	\begin{aligned}
		{\rm FP}_g & = \Pr(\sign(\hat f(X))=1|Y=-1,G=g), \\
		{\rm FN}_g & = \Pr(\sign(\hat f(X))=-1|Y=1,G=g),
	\end{aligned}
\end{equation*}
for $g\in\{0,1\}$ on the test sample. We also compute the total misclassification error by estimating
\begin{equation*}
	{\rm Error} = \Pr(\sign(\hat f(X))\ne Y)
\end{equation*}
on the test sample. We use the TensorFlow, scikit-learn, and XGBoost packages in Python to implement the machine learning methods and consider simple linear Logit, Logit with quadratic polynomial features and LASSO, XGBoost, support vector machines, as well as shallow and deep learning. 

\smallskip

The tuning parameters are computed using $5$-fold cross-validation, including the LASSO penalty, the SVM penalty, and the number of trees. All other parameters are kept at their default values in the XGBoost package. We also use the radial basis function and the default value of the scaling parameter in the scikit-learn package. The neural networks have a width of 15 neurons and all other parameters are set at their default values (batch size, number of epochs, etc.); i.e. we do not use additional regularization ($\ell_1$/$\ell_2$ or dropout). For shallow learning, we use the sigmoid activation function. The deep ReLU neural network has 5 hidden layers. 

\smallskip

The simulation results appear in Table~\ref{tab:MCresults_mistakes}. We find that in terms of the total misclassification error, the ML methods outperform logistic regression in the nonlinear DGP case (except for shallow learning). The best performance is achieved by LASSO with nonlinearities with boosting and support vector machines coming in second. Importantly, we observe a disproportionate number of false positive and false negative mistakes across the two groups in many cases.

\begin{table}[H]
	\caption{\label{tab:MCresults_mistakes} Monte Carlo Simulation Results: symmetric classification}
	\tablexplain
	{\footnotesize
		The Monte Carlo simulation design is presented in Section \ref{subsec:MCdesign}, which represents a stylized social planner facing a disproportionate number of false positive and false negative mistakes across the two groups with the standard ML classification approach. The population consists of two groups, $G$ = 0 and 1. Constituents of a group $G$ = 1 are a fraction $\rho.$ The dimension $d$ of covariates $Z$ is set to $15$ and the scale parameter is $\sigma=0.1$. FP and FN are false positive and false negative mistakes computed as a share in the corresponding group, Total = misclassification rate. Logit = logistic regression, GB = Gradient Boosted trees, SL = shallow learning, DL = deep learning, SVM = support vector machines. The results for $n=1,000$ are based on $5,000$ MC experiments. For $n=10,000$, we reduce the number of MC experiments to $1,000$ due to computational constraints.}
		{\scriptsize
	\begin{ctabular}{lccccccccccccc}
		&	&\multicolumn{6}{c}{Nonlinear DGP: $\tau = 1$}  & \multicolumn{6}{c}{Linear DGP: $\tau=0$}  \\
		&	&\multicolumn{3}{c}{$\rho=0.2$}  & \multicolumn{3}{c}{$\rho = 0.5$}  	& \multicolumn{3}{c}{$\rho=0.2$}	& \multicolumn{3}{c}{$\rho=0.5$}	\\
		&	G& FP & FN	&  Error 	& FP & FN & Error	& FP & FN & Error & FP & FN & Error \\
		&	 &  	&  &  &  &  & &  \\
		& \multicolumn{13}{c}{Sample size $n$ = 10,000} \\
		&	 &  	&  &  &  &  & &  \\
		Logit & 0 & 0.51 & 0.22 & 0.34 & 0.52 & 0.22 & 0.33 & 0.11 & 0.11 & 0.11 & 0.11 & 0.11 & 0.11 \\
		& 1 & 0.71 & 0.11 &  & 0.72 & 0.10 &  & 0.16 & 0.08 &  & 0.16 & 0.08 &  \\
		LASSO & 0 & 0.12 & 0.05 & 0.08 & 0.16 & 0.03 & 0.08 & 0.13 & 0.10 & 0.12 & 0.19 & 0.06 & 0.12 \\
		& 1 & 0.04 & 0.11 &  & 0.06 & 0.08 &  & 0.04 & 0.21 &  & 0.07 & 0.15 &  \\
		XGB & 0 & 0.16 & 0.10 & 0.12 & 0.17 & 0.10 & 0.12 & 0.13 & 0.12 & 0.13 & 0.14 & 0.12 & 0.12 \\
		& 1 & 0.20 & 0.08 &  & 0.20 & 0.07 &  & 0.16 & 0.10 &  & 0.16 & 0.10 &  \\
		SVM & 0 & 0.14 & 0.08 & 0.10 & 0.14 & 0.08 & 0.10 & 0.13 & 0.10 & 0.12 & 0.14 & 0.09 & 0.11 \\
		& 1 & 0.17 & 0.06 &  & 0.17 & 0.05 &  & 0.12 & 0.11 &  & 0.15 & 0.08 &  \\
		SL & 0 & 0.24 & 0.07 & 0.14 & 0.25 & 0.07 & 0.13 & 0.15 & 0.11 & 0.13 & 0.15 & 0.11 & 0.12 \\
		& 1 & 0.25 & 0.05 &  & 0.26 & 0.05 &  & 0.19 & 0.08 &  & 0.19 & 0.07 &  \\
		DL & 0 & 0.18 & 0.07 & 0.12 & 0.20 & 0.07 & 0.12 & 0.15 & 0.12 & 0.13 & 0.15 & 0.12 & 0.13 \\
		& 1 & 0.19 & 0.06 &  & 0.21 & 0.05 &  & 0.19 & 0.09 &  & 0.19 & 0.09 &  \\
		&	 &  	&  &  &  &  & &  \\
		& \multicolumn{13}{c}{Sample size $n$ = 1,000} \\
		&	 & 	& 	&  &  &  &   \\
		Logit & 0 & 0.51 & 0.24 & 0.35 & 0.52 & 0.24 & 0.34 & 0.12 & 0.12 & 0.12 & 0.12 & 0.12 & 0.12 \\
		& 1 & 0.69 & 0.13 &  & 0.70 & 0.13 &  & 0.17 & 0.08 &  & 0.16 & 0.08 &  \\
		LASSO & 0 & 0.17 & 0.07 & 0.11 & 0.22 & 0.05 & 0.12 & 0.14 & 0.10 & 0.13 & 0.19 & 0.07 & 0.13 \\
		& 1 & 0.09 & 0.13 &  & 0.13 & 0.10 &  & 0.05 & 0.21 &  & 0.08 & 0.16 &  \\
		XGB & 0 & 0.30 & 0.16 & 0.22 & 0.32 & 0.15 & 0.22 & 0.16 & 0.13 & 0.14 & 0.18 & 0.12 & 0.14 \\
		& 1 & 0.32 & 0.15 &  & 0.35 & 0.13 &  & 0.14 & 0.15 &  & 0.17 & 0.12 &  \\
		SVM & 0 & 0.26 & 0.15 & 0.20 & 0.27 & 0.15 & 0.19 & 0.16 & 0.11 & 0.14 & 0.19 & 0.09 & 0.13 \\
		& 1 & 0.31 & 0.12 &  & 0.32 & 0.11 &  & 0.12 & 0.14 &  & 0.18 & 0.10 &  \\
		SL & 0 & 0.73 & 0.10 & 0.36 & 0.79 & 0.07 & 0.35 & 0.17 & 0.12 & 0.14 & 0.21 & 0.09 & 0.14 \\
		& 1 & 0.77 & 0.07 &  & 0.82 & 0.05 &  & 0.19 & 0.10 &  & 0.23 & 0.08 &  \\
		DL & 0 & 0.42 & 0.18 & 0.27 & 0.45 & 0.16 & 0.27 & 0.19 & 0.15 & 0.17 & 0.21 & 0.13 & 0.16 \\
		& 1 & 0.45 & 0.15 &  & 0.48 & 0.13 &  & 0.20 & 0.13 &  & 0.22 & 0.12 &  \\
	\end{ctabular}}
\end{table}	

\smallskip

Next, we investigate whether group-specific misclassification rates can be equalized across the two groups with weighted Logit and Boosting. For simplicity, we focus on the setting with $\tau=0$, $\rho=0.2$, and $n=1,000$, as in this case, we observe a disproportionate number of FP and FN across the two groups. Note that in this case Logit has larger FP rates for group $G=1$ while Boosting for group $G=0$. Figure~\ref{fig:fp_fn} shows that the asymmetric weighted logistic regression can equalize the FP probabilities across the two groups at $\varphi_1\approx 1.6$ and FN probabilities across the two groups for $\psi_0\approx 1.6$. Note that equalizing the FP probabilities comes with the increase in FN probabilities in the group $G=1$ and equalizing the FN probabilities comes with the increase in the FP probabilities in the group $G=0$. Therefore, the decision-maker or the social planner has to think carefully about all these trade-offs when calibrating the asymmetric loss function. The results for the gradient boosting are similar, except for the fact that larger weight factors are required to equalize FP/FN probabilities across groups.

\begin{figure}[h]
	\centering
	\begin{subfigure}{0.45\textwidth} 
		\includegraphics[width=\textwidth]{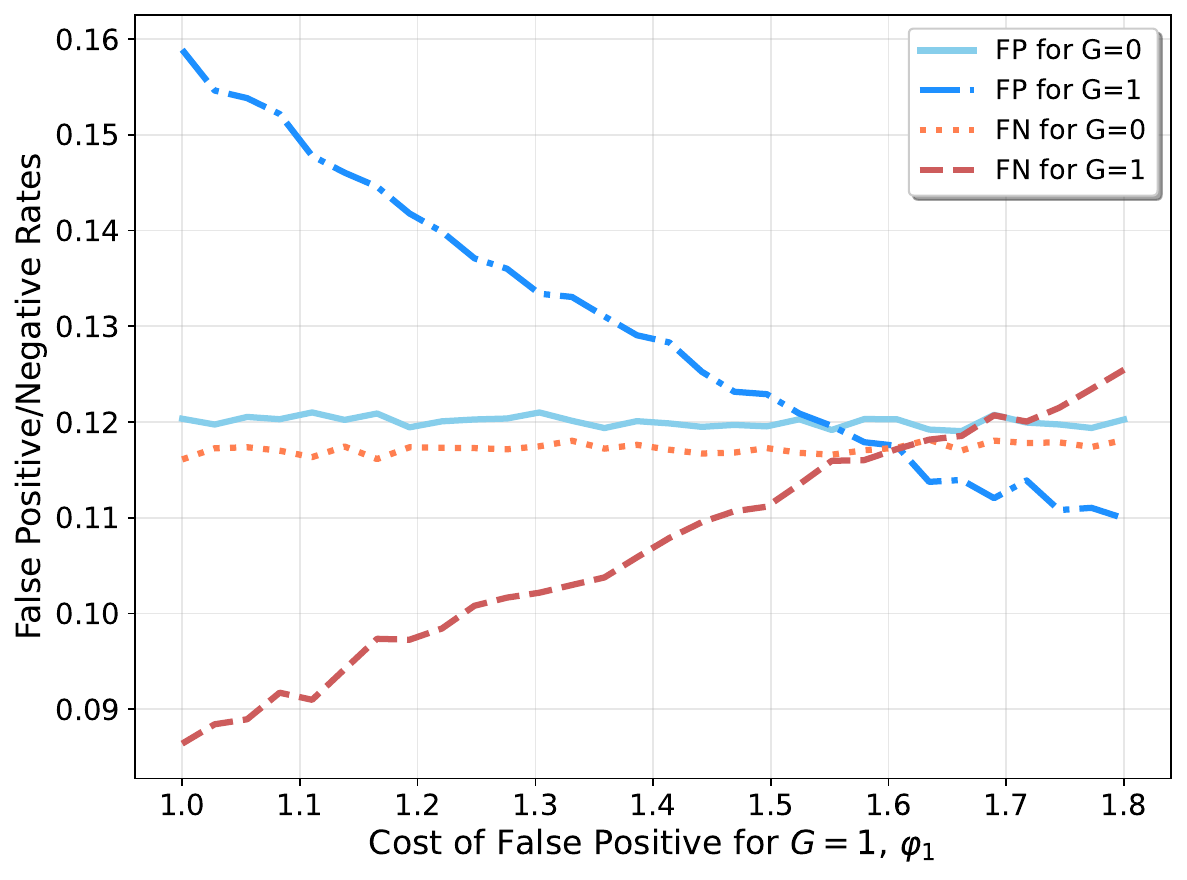}
		\caption{Logit: False Positives (FP)} 
	\end{subfigure}
	\begin{subfigure}{0.45\textwidth} 
		\includegraphics[width=\textwidth]{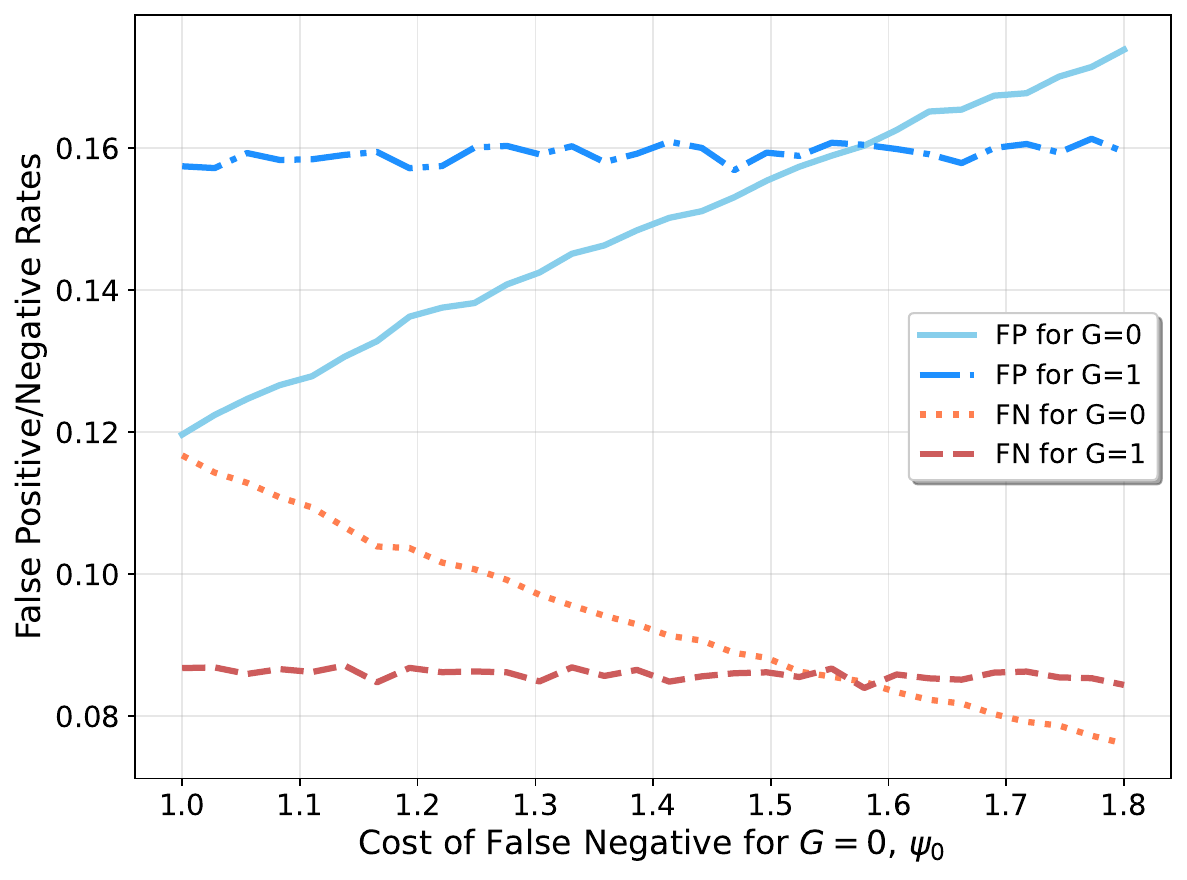}
		\caption{Logit: False Negatives (FN)} 
	\end{subfigure}
	\begin{subfigure}{0.45\textwidth} 
		\includegraphics[width=\textwidth]{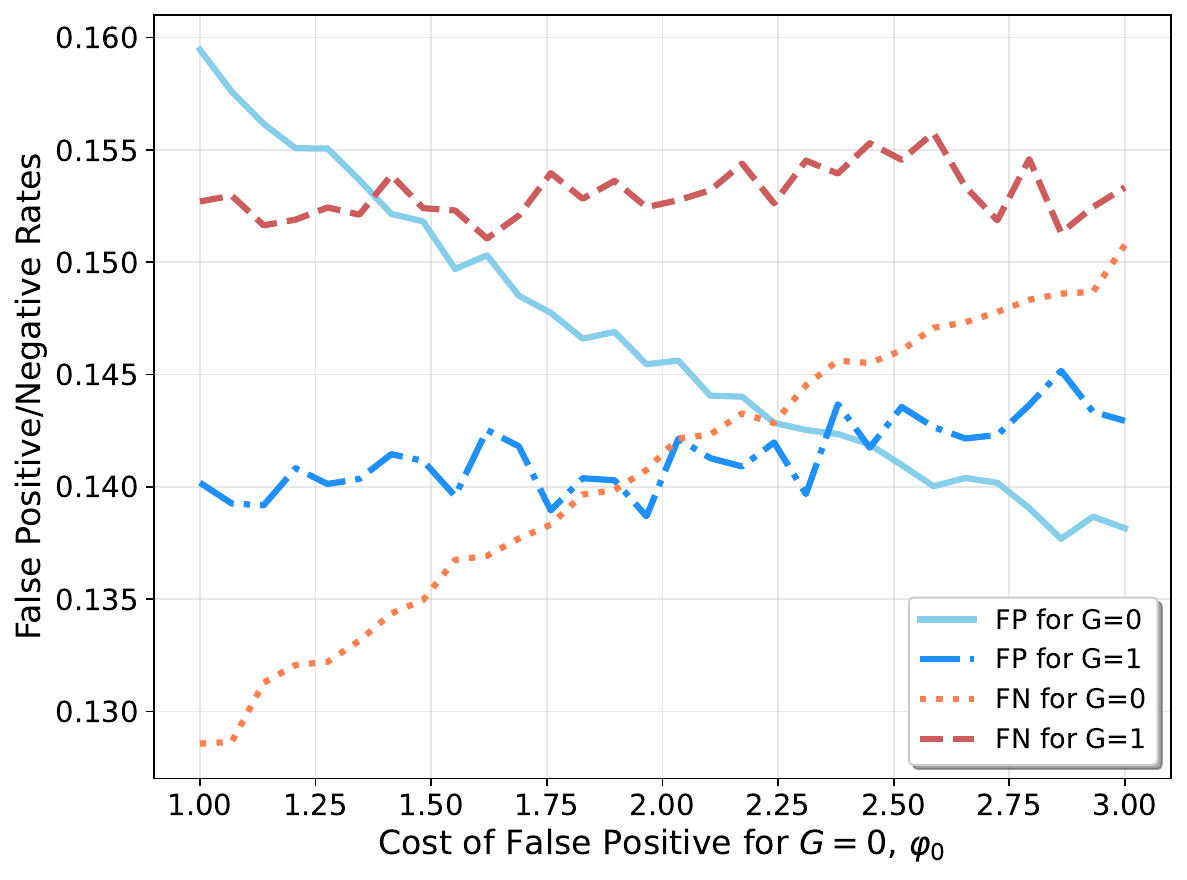}
		\caption{Boosting: False Positives (FP)} 
	\end{subfigure}
	\begin{subfigure}{0.45\textwidth} 
		\includegraphics[width=\textwidth]{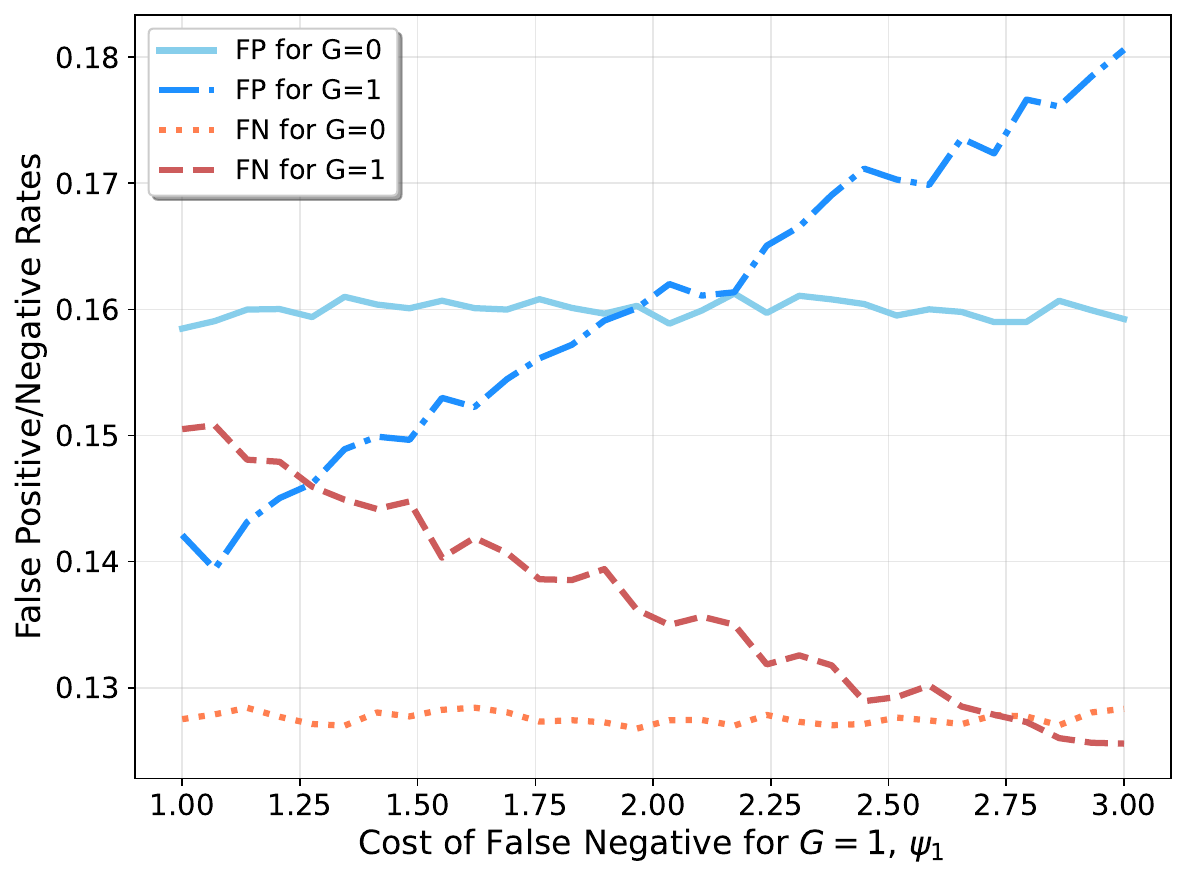}
		\caption{Boosting: False Negatives (FN)} 
	\end{subfigure}
	\caption{Asymmetric Binary Choice. Each subplot corresponds to changes in the FP cost (left) or FN cost (right) for Logit (top) and boosting (bottom). The figure shows that introducing asymmetries in the loss function can equalize the False Positive and the False Negative rates across groups. Setting: $\rho=0.2$, $\sigma=0.1$, $\tau=0$, $n=1,000$. Results based on $5,000$ Monte Carlo experiments.} 
	\label{fig:fp_fn}
\end{figure}

Next, we compare the performance of the standard logistic regression approach, which ignores the asymmetric loss function, with the asymmetric logistic regression in terms of the average loss of a social planner. The social planner loss for the former will be denoted by $\ell_{logit}$ while our new estimator yields $\ell_{w-logit}.$

\smallskip

The simulation results are reported in Table \ref{tab:MCresults}.  We report several measures to appraise the findings. First, we report P($\ell_{logit}$ $>$ $\ell_{w-logit}$) which is the percentage that the standard logistic regressions generate larger social planner costs compared to our weighted regression. Hence, this measure reports how often our estimator outperforms the standard procedure, where the probabilities are computed from the Monte Carlo simulated samples. Next, we report summary statistics for the ratio $\ell_{logit}/\ell_{w-logit}.$ When the ratio is above 1.00 then the weighted logistic approach is better. We report the minimum, maximum, mean, and three quartiles of the simulation distribution. All simulations involve 5000 replications.

\smallskip

We start from a baseline case for the parameter setting, namely: $\psi_0$ = 3, $\psi_1$ = 1, $\varphi_0$ = 1.7, and finally $\varphi_1$ = 1. For the baseline case and $n$ = 1,000,  P($\ell_{logit}$ $>$ $\ell_{w-logit}$) is 0.52, meaning that in more than 50\% of cases, our procedure is superior (lower social costs) to the standard logistic regression. A more detailed summary statistics is reported in the Online Appendix, Table~\ref{tab:MCresults_summary}. We also examine various deviations from the baseline case. For the larger sample size $n$ = 10,000 the gains are larger, which is now 0.62, meaning that the probability that the symmetric logistic regression leads to larger losses increases as we get more data.

\begin{table}[H]
	\caption{\label{tab:MCresults} Monte Carlo Simulation Results: symmetric vs. asymmetric Logit}
	\tablexplain
	{\footnotesize
		The Monte Carlo simulation design is presented in Section \ref{subsec:MCdesign}, which represents a stylized social planner with a loss function from Example~\ref{ex:fairness} featuring asymmetries for false positives and false negatives. The population consists of two groups, $G$ = 0 and 1, with the former assumed to be a protected segment of the population. Constituents of a group $G$ = 1 are a fraction $\rho.$ The baseline case has the loss function with the following setting: $\psi_1=\varphi_1$ = 1, $\varphi_0=1.7$, and $\psi_0=3$. We also set $\rho = 0.2$, $\sigma=0.1$, and $\tau=0$. We compare the performance of a standard logistic regression approach, which ignores the asymmetric loss function, with our convexified weighted logistic model appearing in equation (\ref{eq:simrisk}). The social planner loss for the former will be denoted by $\ell_{logit}$ while our new estimator yields $\ell_{w-logit}.$ We report P($\ell_{logit}$ $>$ $\ell_{w-logit}$) which is the percentage that standard logistic regressions generate larger social planner costs compared to our weighted regression. Hence, this measure reports how often our estimator outperforms the standard procedure, where the probabilities are computed from the Monte Carlo simulated samples. Columns with $\varphi_0,\varphi_1,\psi_0$, or $\psi_1$ as headers represent deviations from the baseline case.}
	\begin{ctabular}{lccccccccc}
		&	Baseline & $\rho$ 	&  $\tau$  	& $\varphi_0$ & $\varphi_1$	& $\varphi_1$ & $\psi_0$	& $\psi_1$ 	& $\psi_1$\\
		&	case & 0.5	&  1 	& 2 & 2 &3  & 4 & 2	& 3 \\
		&	 & 	& 	&  &  &  &  & &  \\
		& \multicolumn{9}{c}{P($\ell_{logit}$ $>$ $\ell_{w-logit}$) } \\
		&	 & 	& 	&  &  &  &   \\
		Sample, $n$  & \\
		1,000 & 0.52 & 0.45 & 0.84 & 0.52 & 0.64 & 0.69 & 0.58 & 0.65 & 0.70 \\
		5,000 & 0.57 & 0.51 & 0.93 & 0.52 & 0.64 & 0.70 & 0.62 & 0.64 & 0.69 \\
		10,000 & 0.62 & 0.53 & 0.97 & 0.54 & 0.69 & 0.75 & 0.71 & 0.68 & 0.74 \\
	\end{ctabular}
\end{table}	

Next, we report two columns in Table \ref{tab:MCresults} where we change the fraction of the protected population $\rho$ from 0.2 to respectively 0.5, introducing nonlinearities with $\tau=1$, and changing the cost structure. We can see that in most of the cases, our approach is superior with a larger margin compared to the baseline case. Introducing more asymmetries in the loss function implies better performance of asymmetric logistic regression, as expected.

\medskip

Lastly, we compare the performance of the weighted ERM decisions to simple plug-in rules motivated by Proposition~\ref{prop:bayes}. The plug-in rules predict $\hat Y=1$ when $\hat\eta(x)\geq c(x)$, where $\hat\eta(x)$ is an estimator of $\eta(x)=\Pr(Y=1|X=x)$. Note that in the special case of the symmetric binary classification with the logistic convexification, the plug-in decisions are actually equivalent to the ERM decisions in population since
\begin{equation*}
	\eta(x) = \frac{1}{1+e^{-f(x)}}>0.5 \iff f(x)>0.
\end{equation*}
However, this equivalence does not hold in our asymmetric setting.\footnote{It is known that the plug-in rules can be inferior or superior relative to ERM in the symmetric case depending on the distribution of the data; see \cite{audibert2007fast} for more detailed comparison.} In addition, some of the popular and successful ML techniques, e.g., support vector machines and neural networks with hinge convexification, are ERM-based and do not estimate $\eta$. For simplicity, we focus on the parametric case and compare the asymmetric Logit to the asymmetric plug-in rules based on the symmetric Logit. These results are presented in Table~\ref{tab:MCresults_plug_in}. In most of the cases, the plug-in rules are inferior to our approach, and the difference is more pronounced in small samples. Lastly, we provide more detailed summary statistics for the ratio $\ell_{plugin}/\ell_{w-logit}$ in the Online Appendix, Table~\ref{tab:MCresults_plug_in_summary}.

\medskip

\begin{table}
	\caption{\label{tab:MCresults_plug_in} Monte Carlo Simulation Results: plug-in vs. asymmetric Logit}
	\tablexplain
	{\footnotesize
		The Monte Carlo simulation design is presented in Section \ref{subsec:MCdesign}, which represents a stylized social planner with a loss function from Example~\ref{ex:fairness} featuring asymmetries for false positives and false negatives. The population consists of two groups, $G$ = 0 and 1, with the former assumed to be a protected segment of the population. Constituents of a group $G$ = 1 are a fraction $\rho.$ The baseline case has the loss function with the following setting: $\psi_1=\varphi_1$ = 1, $\varphi_0=1.7$, and $\psi_0=3$. We also set $\rho = 0.2$, $\sigma=0.1$, and $\tau=0$. We compare the performance of a standard logistic regression approach, which ignores the asymmetric loss function, with our convexified weighted logistic model appearing in equation (\ref{eq:simrisk}). The social planner loss for the former will be denoted by $\ell_{plugin}$ while our new estimator yields $\ell_{w-logit}.$ We report P($\ell_{plugin}$ $>$ $\ell_{w-logit}$) which is the percentage that standard logistic regressions generate larger social planner costs compared to our weighted regression. Hence, this measure reports how often our estimator outperforms the standard procedure, where the probabilities are computed from the Monte Carlo simulated samples. Columns with $\varphi_0,\varphi_1,\psi_0$, or $\psi_1$ as headers represent deviations from the baseline case.}
	\begin{ctabular}{lccccccccc}
		&	Baseline & $\rho$ 	&  $\tau$  	& $\varphi_0$ & $\varphi_1$	& $\varphi_1$ & $\psi_0$	& $\psi_1$ 	& $\psi_1$\\
		&	case & 0.5	&  1 	& 2 & 2 &3  & 4 & 2	& 3 \\
		&	 & 	& 	&  &  &  &  & &  \\
		& \multicolumn{9}{c}{P($\ell_{plugin}$ $>$ $\ell_{w-logit}$) } \\
		&	 & 	& 	&  &  &  &   \\
		Sample, $n$  & \\
		1,000 & 0.76 & 0.64 & 0.00 & 0.55 & 0.80 & 0.84 & 0.90 & 0.81 & 0.83 \\
		5,000 & 0.76 & 0.56 & 0.00 & 0.42 & 0.83 & 0.89 & 0.96 & 0.83 & 0.88 \\
		10,000 & 0.52 & 0.30 & 0.00 & 0.18 & 0.66 & 0.76 & 0.92 & 0.66 & 0.75 \\
	\end{ctabular}
\end{table}

\section{Pretrial Detention Decisions: Algorithmic Fairness and Recidivism Revisited \label{sec:empirical}}
The purpose of this section is to illustrate our novel econometric methods with an application to the problem of algorithmic fairness in pretrial detention. We will show how ML taking into account social planner preferences can be incorporated directly into the digital decision process.\footnote{Among economists, the idea to apply ML to pretrial detention decisions has recently been explored by \cite{kleinberg2018human}. } In particular, we will illustrate that valuing outcomes of protected groups allows us to reduce various forms of algorithmic discrimination.

\smallskip

Across the nation, judges are increasingly using algorithms to assess the risk of whether a defendant, if released, would fail to appear in court or be rearrested for a new crime, a term called recidivism. Journalists at ProPublica have investigated a commercial algorithm, called Correctional Offender Management Profiling for Alternative Sanctions (COMPAS), developed by Northpointe, Inc. The algorithm assigns defendants recidivism risk scores based on more than 100 factors, including age, sex, and criminal history. 

\smallskip

COMPAS does not explicitly use race as an input, hence, there are no disparate treatment issues.\footnote{The legal scholars distinguish between the disparate treatment and the disparate impact and there is often a tension between these two. The \textit{disparate treatment} pertains to using race explicitly in various decisions and is often illegal, including bail, hiring, admission, loans, etc. However, even when the decision is carefully made to avoid the disparate treatment, it may still have unintended \textit{disparate impact} consequences with substantial legal implications. } Nevertheless, the aforementioned ProPublica article revealed that African American defendants are substantially more likely to be classified as high risk. Therefore, the ProPublica article pertains to the disparate impact issues, and remedying it would require using race explicitly. This bring us to the legal issues related to algorithmic affirmative action; see \cite{bent2019algorithmic} for further discussion. As a simple illustration of our methodology, we will focus on equalizing the false positive rates, computed as
\begin{equation}\label{eq:fp_g}
	{\rm FP}_g = \Pr(\sign(\hat f(X))=1|Y=-1,G=g)
\end{equation}
or positive predictive values 
\begin{equation}\label{eq:ppv_g}
	{\rm PPV}_g = \Pr(Y=1|\sign(\hat f(X))=1,G=g)
\end{equation}
for African American versus other individuals, i.e. $g\in\{\mathrm{AF,O}\}$. Equalizing the false positive rates in equation~\eqref{eq:fp_g} achieving the error rate balance fairness while equalizing the positive predictive values in equation~\eqref{eq:ppv_g} achieves the predictive parity fairness; see \cite{mitchell2021algorithmic} for further discussion. Therefore, we will show that our methodology can mitigate the disparate impact problem. Note that it is not the purpose of our study to compare ML outcomes with human decisions (as in \cite{kleinberg2018human}) or to compare ML outcomes with COMPAS.

\smallskip

We use data from Broward County, Florida, originally collected by ProPublica. Following their analysis, we only consider defendants who were assigned COMPAS risk scores within 30 days of their arrest and define the recidivism as a new felony or misdemeanor charge within two years after the initial screening. In addition, we restrict our analysis to only those defendants who spent at least two years (after their COMPAS evaluation) outside a correctional facility without being arrested for a violent crime or were arrested for a violent crime within these two years. Following standard practice, we use this two-year violent recidivism metric and set $Y_i=1$  for those who re-offended within this window, and $Y_i = -1$ for those who did not. In the Online Appendix Section \ref{appsubsec:cleaning}, we provide a detailed description of data cleaning and report the summary statistics for the cleaned dataset.\footnote{As we explain in Online Section \ref{appsubsec:cleaning}, while working with the COMPAS data, we discovered several issues regarding how the dataset was constructed by ProPublica journalists. Ideally, for each defendant we should only have one entry included in the data set with that record containing all the information including current charges and past crimes, to predict recidivism. However, we found multiple records for the same arrest in the COMPAS data set, with only differences in charges. As a result we undertook a thorough cleaning process described in the Online Appendix.} The total number of individuals is 8,227 with 495 variables making the problem high-dimensional, hence, calling for ML techniques as it may not be clear what types of crimes are predictive of recidivism. A large number of these variables are dummies for different kinds of crimes from the criminal code. The rest include variables like age, sex, race, crime history, marriage status, etc. We have a demographic mix dominated by African Americans and Caucasians, respectively 4,109 and 2,797 in numbers. We stratify race along two groups, African American versus all other, which given our sample is roughly African American versus Caucasian. The binary outcome $is\_recid$ indicates that 31.3\% of cases resulted in a recidivism. The minimum age is 16 with a median of 29.

\begin{figure}[h]
	\centering
	\includegraphics[width=0.6\textwidth]{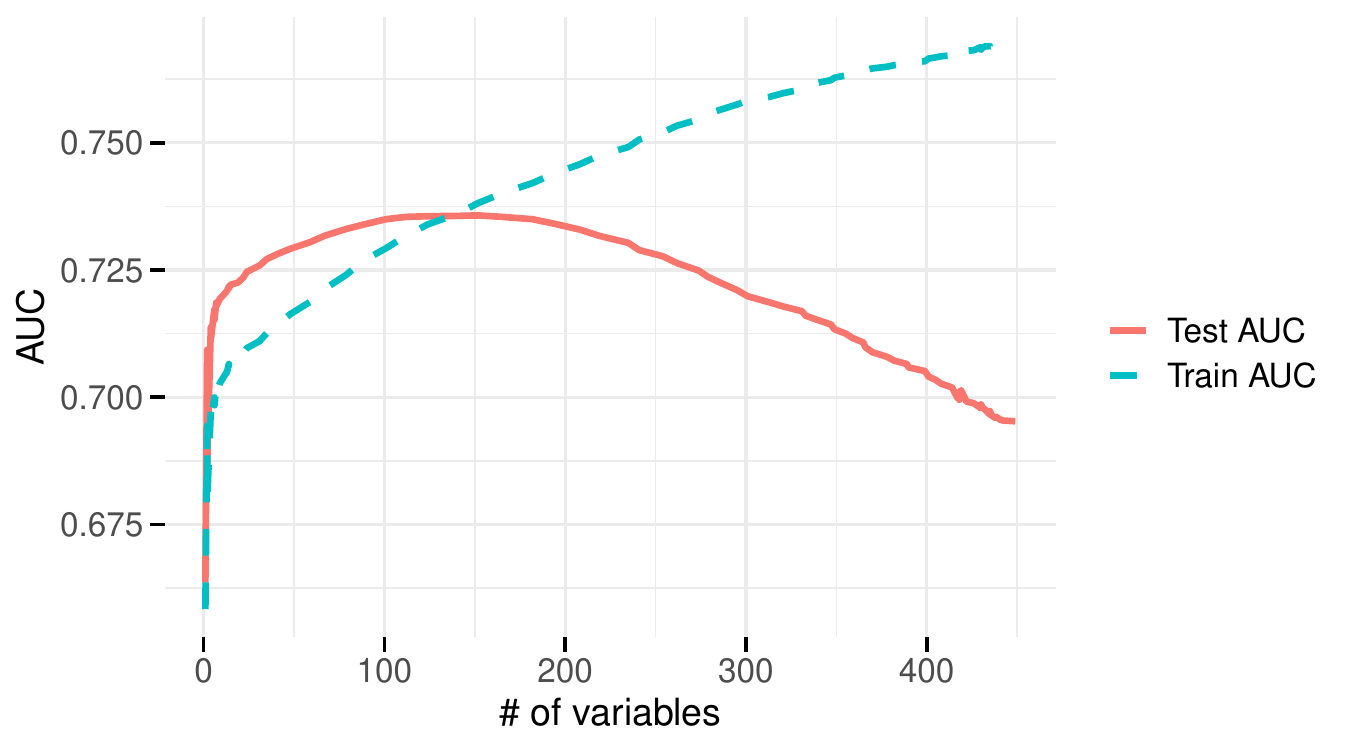}
	\caption{Training and Test AUC of LASSO-Logit path. The figure shows that the highest test AUC is achieved for the model with 152 covariates.} 
	\label{fig:auc_application}
\end{figure}

To understand whether the high-dimensional machine learning methods are needed for this dataset, we fit a LASSO-Logit regression and plot the area under the curve (AUC) for the 20\% of the data left out as the test sample. The results appear in Figure~\ref{fig:auc_application}. We find that the highest test AUC, namely $0.74$, is achieved for the model with 152 covariates.\footnote{It is worth mentioning that simple linear Logit using all 495 covariates is not feasible because the matrix of covariates is close to singular. See also Online Appendix Table~\ref{tab:post_lasso} for the post-LASSO Logit estimates.}

\begin{table}[!htbp] \centering 
	\caption{Symmetric Classifiers: false positive/negative rates and AUC for test sample.} 
	\label{tab:symmetric} 
	\begin{tabular}{@{\extracolsep{5pt}} lccc} 
		\\[-1.8ex]\hline 
		\hline \\[-1.8ex] 
		& FP & FN & AUC \\ 
		& \multicolumn{3}{c}{LASSO-Logit} \\
		All & $0.049$ & $0.772$ & $0.734$ \\ 
		Other & $0.026$ & $0.840$ & $0.733$ \\ 
		African American & $0.075$ & $0.720$ & $0.719$ \\ 
		& & & \\
		& \multicolumn{3}{c}{Boosting} \\
		All & $0.095$ & $0.695$ & $0.725$ \\ 
		Other & $0.055$ & $0.763$ & $0.736$ \\ 
		African American & $0.143$ & $0.644$ & $0.697$ \\ 
		\hline \\[-1.8ex] 
	\end{tabular} 
\end{table} 

Next, we fit the standard symmetric LASSO-Logit and Boosting classifiers using 494 covariates (excluding the race dummy). We tune the penalty of LASSO-Logit and the learning rate of Boosting using 5-fold cross-validation on the training sample. The test errors of both classifiers are reported in Table~\ref{tab:symmetric}. The standard symmetric classifiers have disproportionate biases even though the race variable has not been used to fit these models. African American defendants have almost 3 times higher false positive rates compared to the rest of the population. We also note that false negative rates are enormously large compared to the false positive rates, and calling for additional balancing. The performance of boosting is slightly inferior to LASSO-Logit, suggesting the limited role of additional nonlinearities. This is probably not suprising given that our covariates already include interactions between various crime types and degrees as well as a quadratic age term.

\smallskip

To balance the false positive and false negative rates rates, we fit the asymmetric classifiers with a loss function
\begin{equation*}
	\ell(f,y) = \psi\one_{f=-1,y=1} + \one_{f=1,y=-1},
\end{equation*}
where $f\in\{-1,1\}$ is recidivism prediction and $y\in\{-1,1\}$ recidivism outcome. The cost of false positive mistakes is normalized to $1$ and the costs of false negative mistakes denoted by $\psi$. 

\smallskip

The performance of asymmetric classifiers for various values of $\psi\in[0,10]$ appears in Figure~\ref{fig:fp_fn_application}. We find that changing the cost of false negative mistakes allows us to achieve balanced false positive and false negative rates both for LASSO-Logit and Boosting. The false positive and false negative rates are equalized approximately at $\psi=2.25$ for both classifiers.
\begin{figure}[h]
	\centering
	\begin{subfigure}{0.49\textwidth} 
		\includegraphics[width=\textwidth]{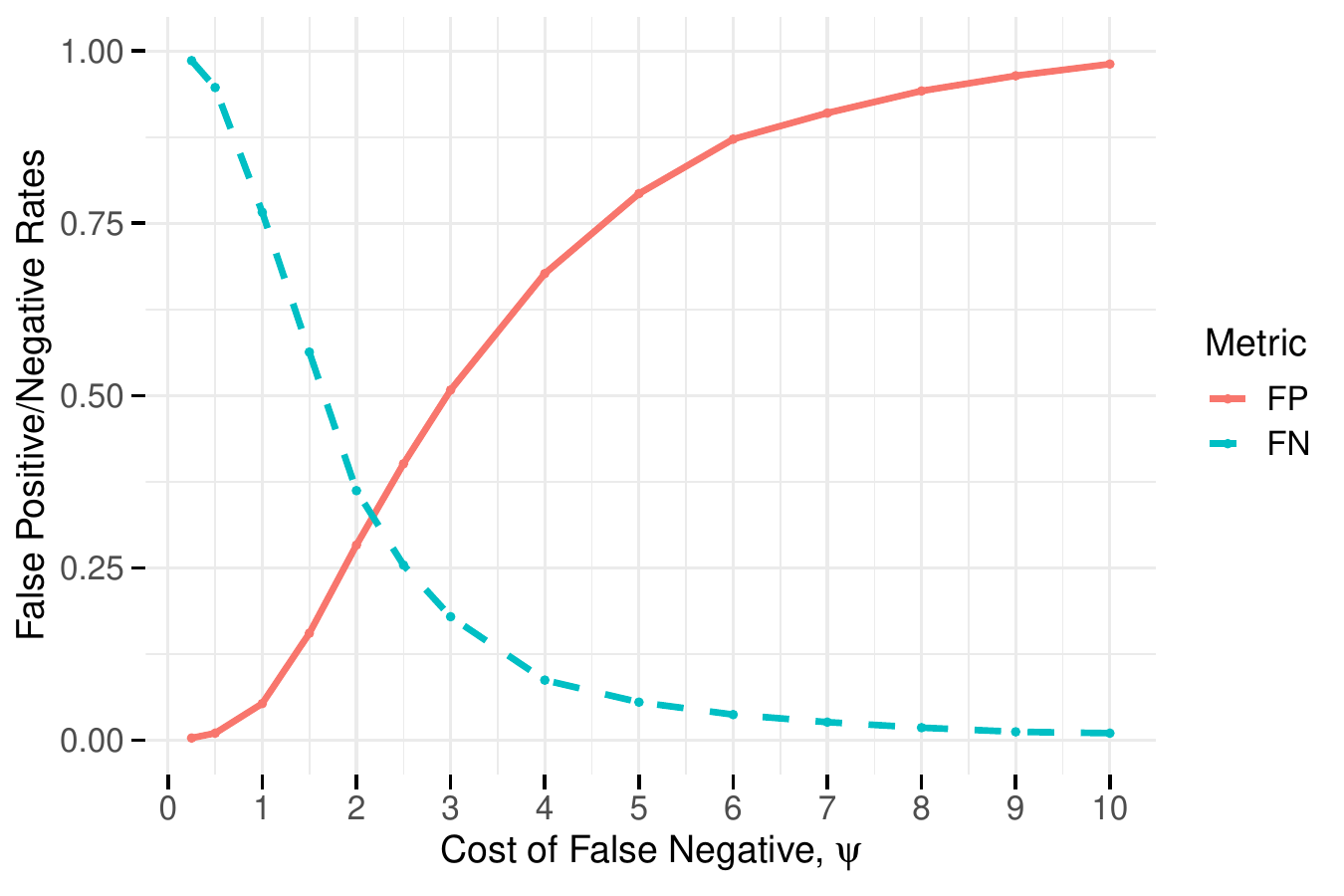}
		\caption{LASSO-Logit: FP and FN rates} 
	\end{subfigure}
	\begin{subfigure}{0.49\textwidth} 
		\includegraphics[width=\textwidth]{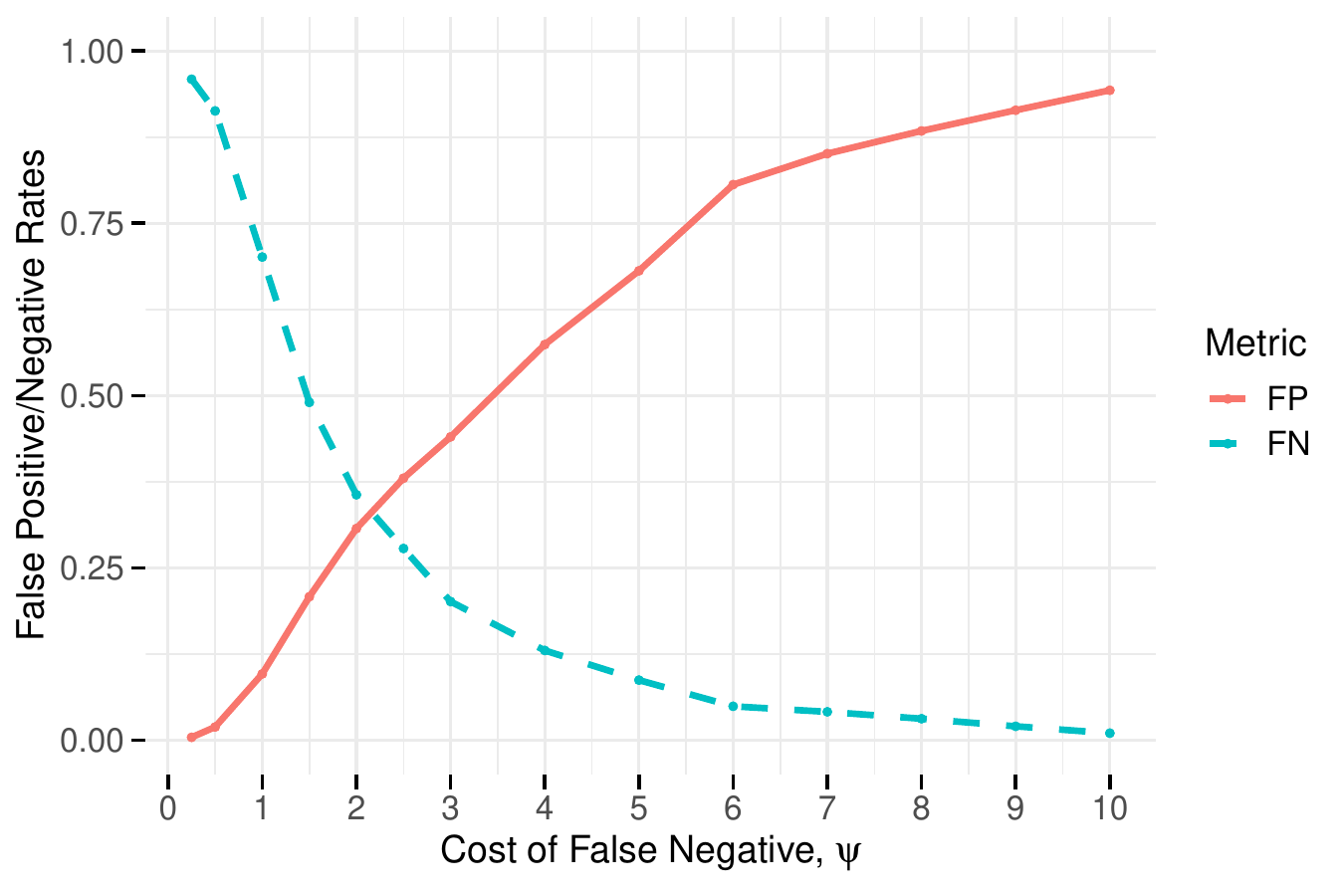}
		\caption{Boosting: FP and FN rates} 
	\end{subfigure}
	\caption{Asymmetric Binary Choice: Penalizing False Negative Mistakes. The figure shows that increasing the cost of false negative mistakes in the loss function is enough to balance FP and FN rates.} 
	\label{fig:fp_fn_application}
\end{figure}
However, changing $\psi$ is not enough to balance the false positive rates across the two groups; see  Online Appendix Figure~\ref{fig:fp_fn_application_race} and Table~\ref{tab:balanced} for more details. 

\smallskip

To address the algorithmic fairness issue, we rely on a two-group setup also used in the previous section where the protected group ($G$ = 1) are African American offenders.  We focus on the loss function from Example~\ref{ex:fairness}:
\begin{equation*}
	\ell(f,y,g) = \psi\one_{f=-1,y=1} + \varphi_g\one_{f=1,y=-1},
\end{equation*}
where $f\in\{-1,1\}$ is a recidivism prediction, $y\in\{-1,1\}$ is a recidivism outcome, and $g\in\{0,1\}$ is the group membership. Therefore, $\varphi_1$ is a false positive cost for African Americans which we will vary while $\psi=2.25$ is a false negative cost calibrated to balance the false positive and false negative rates.

\begin{figure}[H]
	\centering
	\begin{subfigure}{0.45\textwidth} 
		\includegraphics[width=\textwidth]{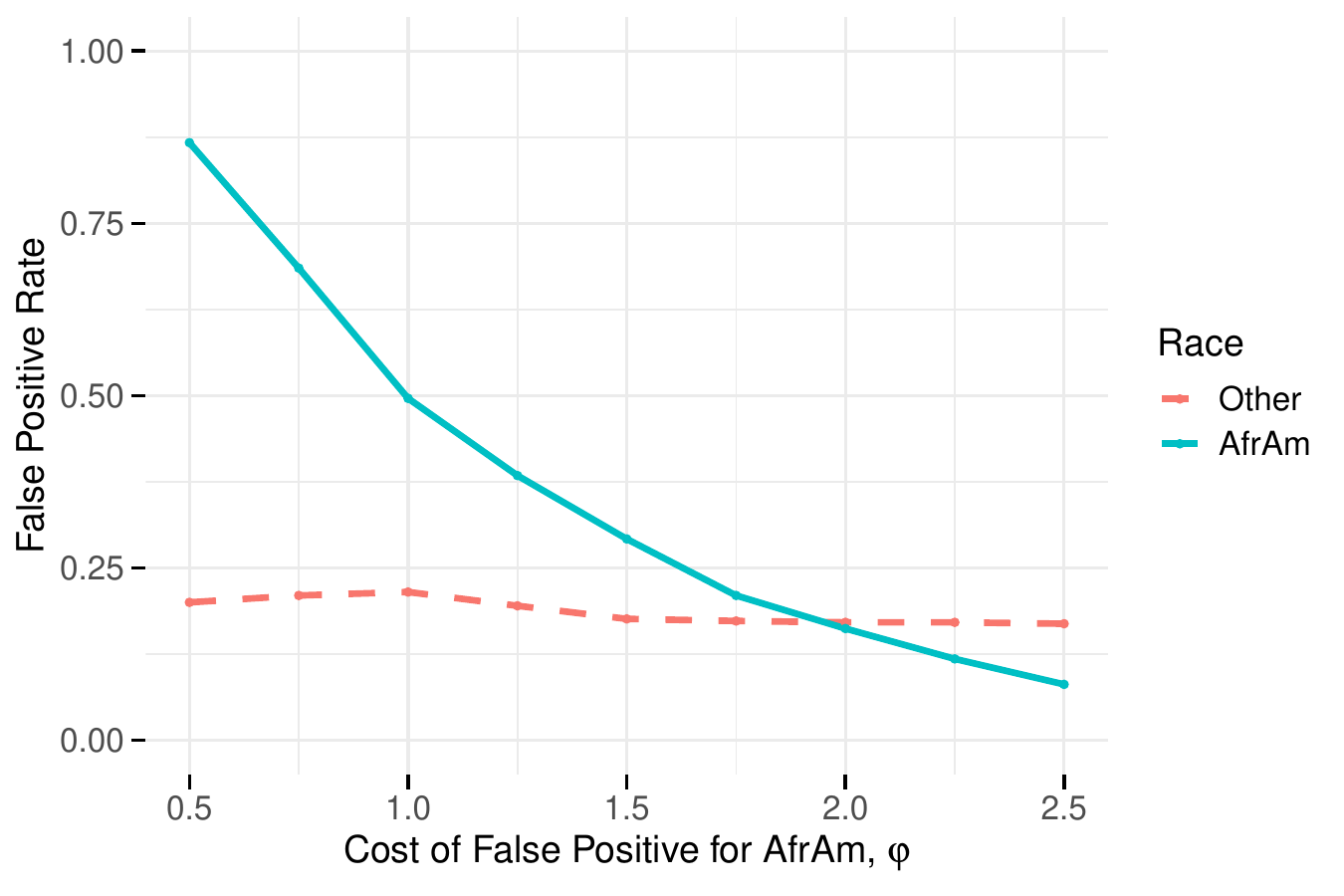}
		\caption{LASSO-Logit: FP rates} 
	\end{subfigure}
	\begin{subfigure}{0.45\textwidth} 
		\includegraphics[width=\textwidth]{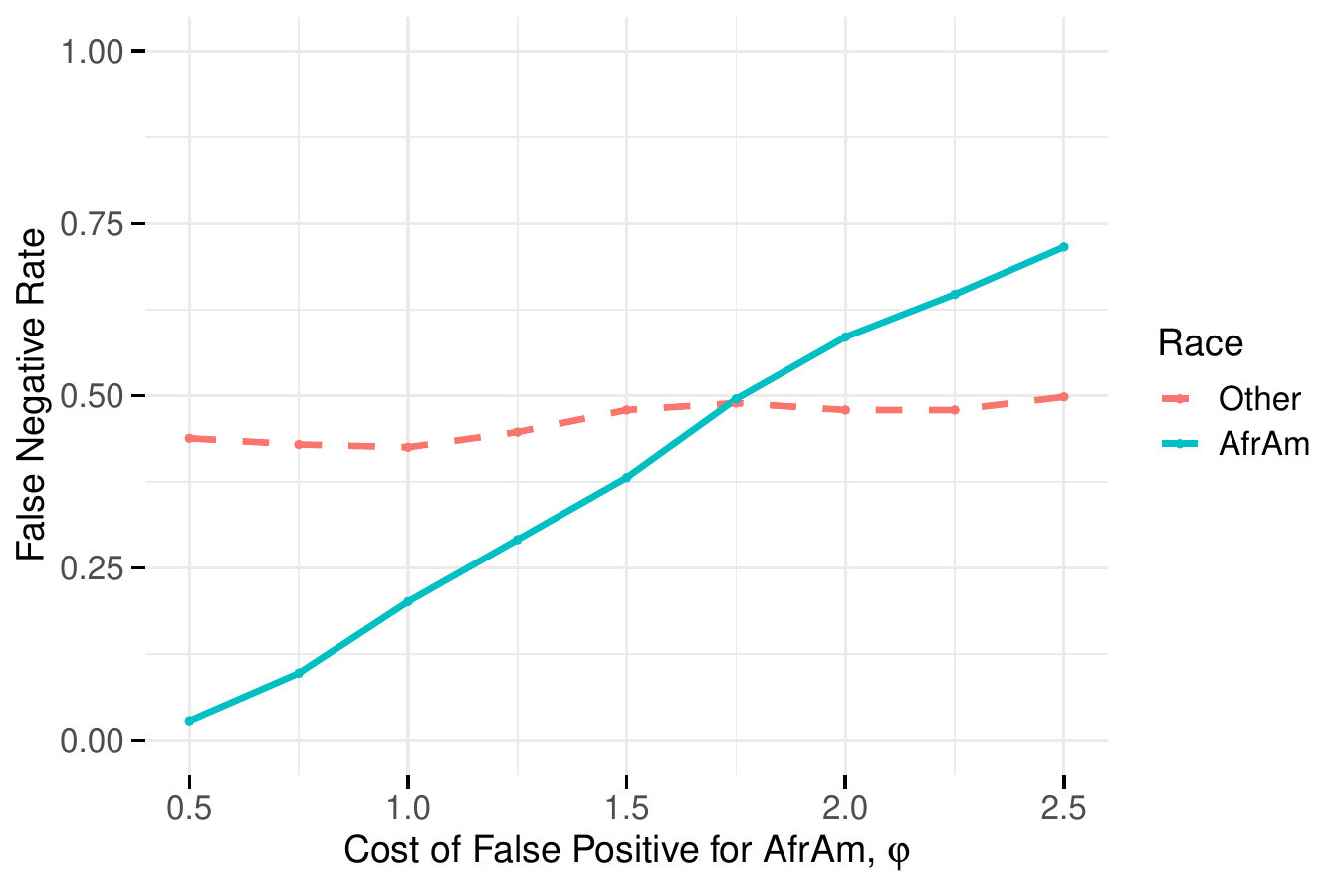}
		\caption{LASSO-Logit: FN rates} 
	\end{subfigure}
	\begin{subfigure}{0.45\textwidth} 
		\includegraphics[width=\textwidth]{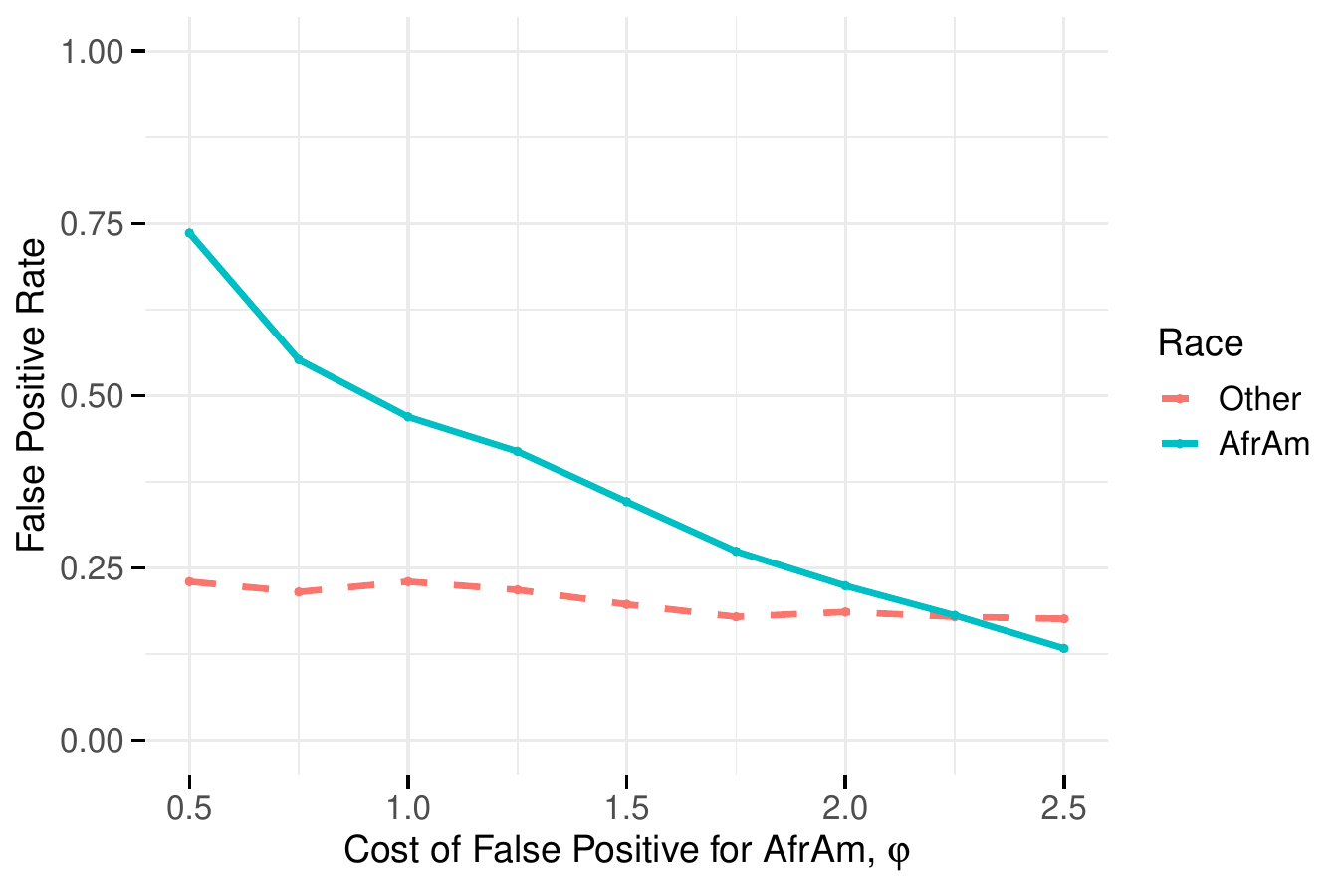}
		\caption{Boosting: FP rates} 
	\end{subfigure}
	\begin{subfigure}{0.45\textwidth} 
		\includegraphics[width=\textwidth]{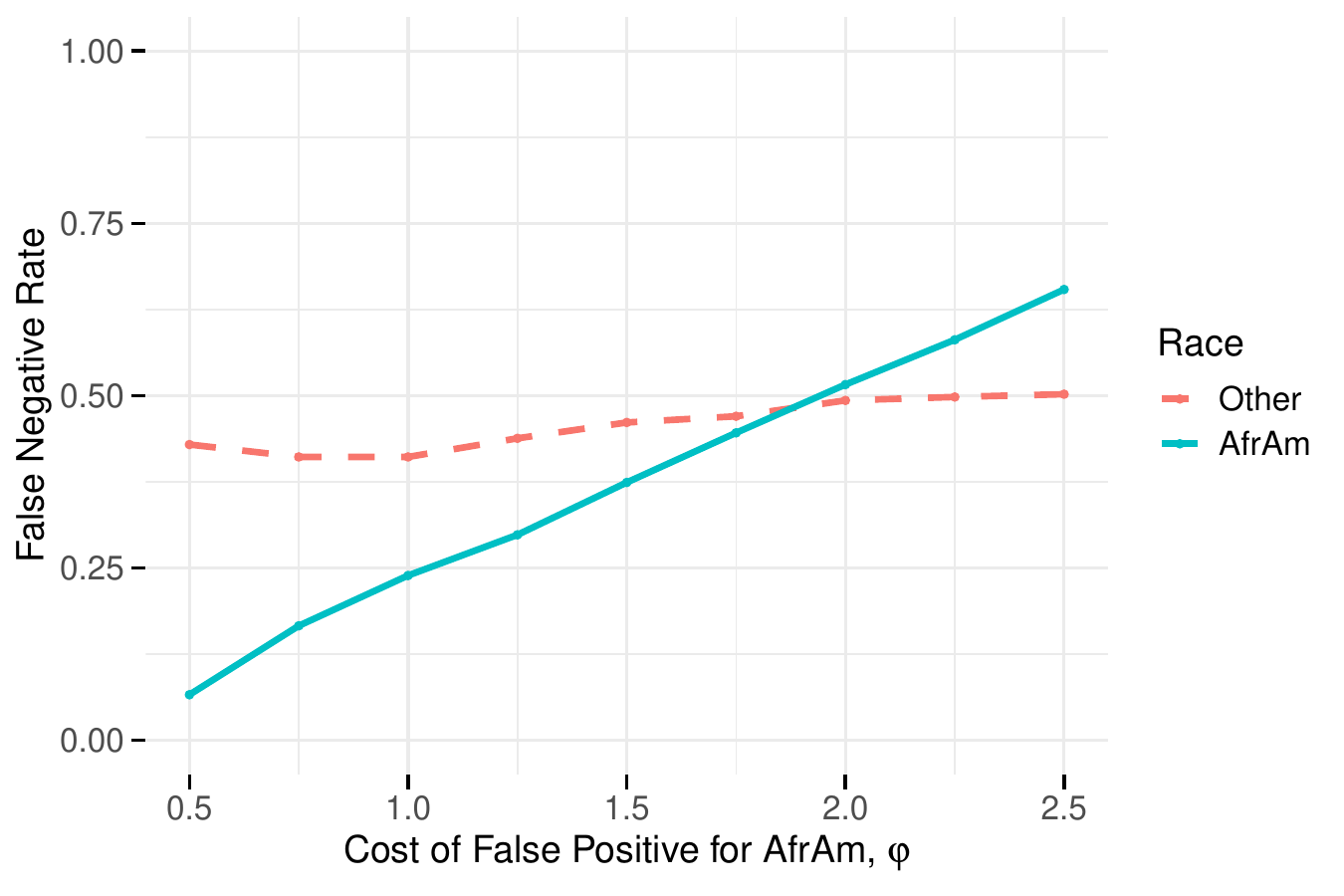}
		\caption{Boosting: FN rates} 
	\end{subfigure}
	\caption{Fairness: Balanced Error Rates. The figure shows that increasing the cost of false positive (FP) mistakes for African Americans allows us to balance FP rates across the two groups. We use the cost of false negatives $\psi=2.25$ and vary the cost of false positives for African American defendants $\varphi_1\in[0.5,2.5]$.} 
	\label{fig:fairness_fp}
\end{figure}

First, we look at balancing the false positive rates across the two groups by varying $\varphi_1$. The results appear in Figure~\ref{fig:fairness_fp}. We find that the false positive rates for African Americans and Others can be equalized if the decision-maker values the outcomes of African American defendants more highly, i.e., the costs of false positive mistakes for African Americans, $\varphi_1$, are sufficiently high. The two rates are equalized when $\varphi_1\approx 1.9$ (LASSO-Logit) and $\varphi_1\approx2.25$ (Boosting). Therefore, varying $\varphi_1$ allows us to mitigate the disparate impact problem highlighted by ProPublica, achieving balanced false positive rates for the two groups. We also note that this type of fairness comes with a cost of increased false negative rates for African Americans. In other words, to reduce the number of wrongly detained African American defendants, we have to accept a larger number of wrongly released African American defendants.

\begin{figure}[H]
	\centering
	\begin{subfigure}{0.45\textwidth} 
		\includegraphics[width=\textwidth]{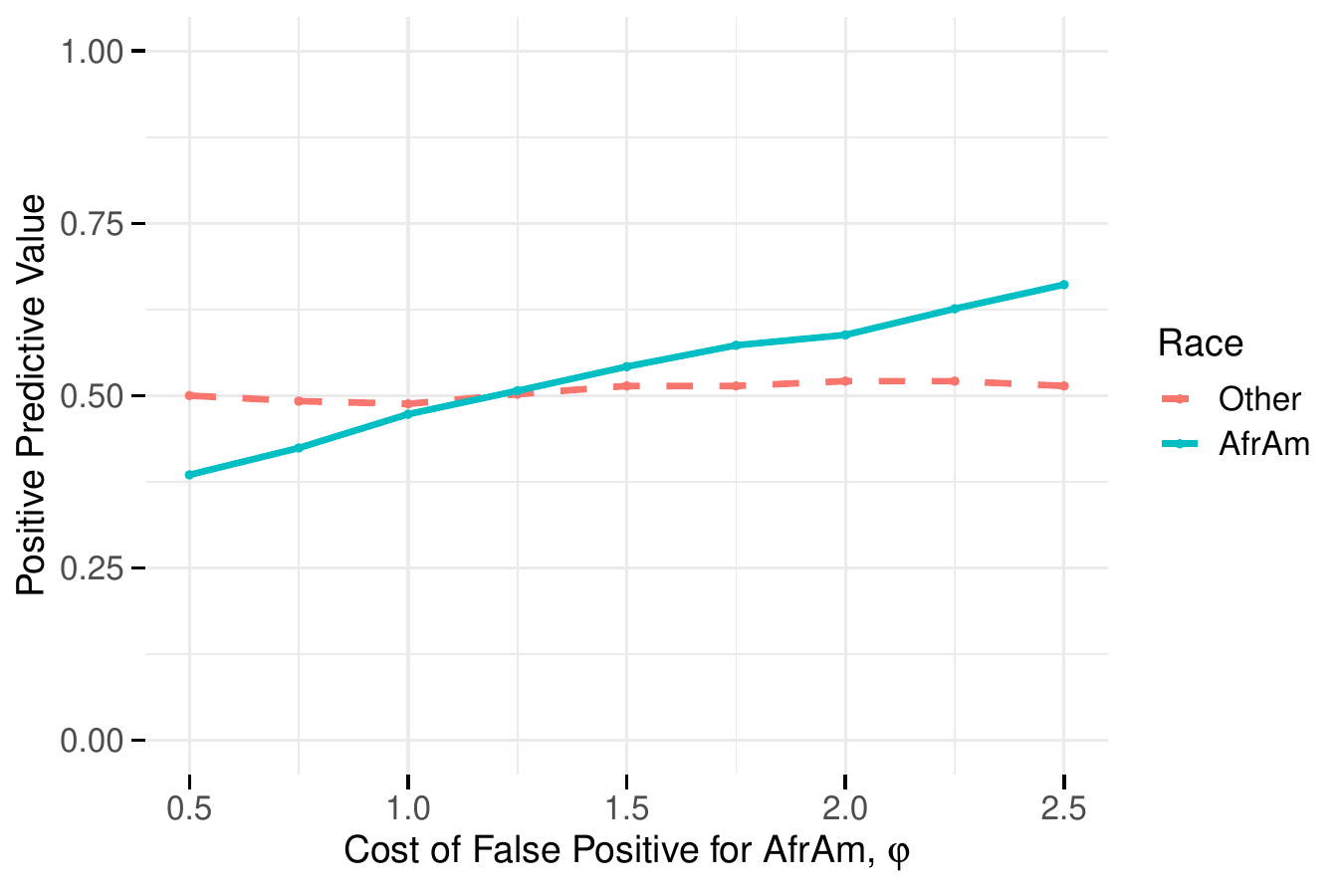}
		\caption{LASSO-Logit: PPV} 
	\end{subfigure}
	\begin{subfigure}{0.45\textwidth} 
		\includegraphics[width=\textwidth]{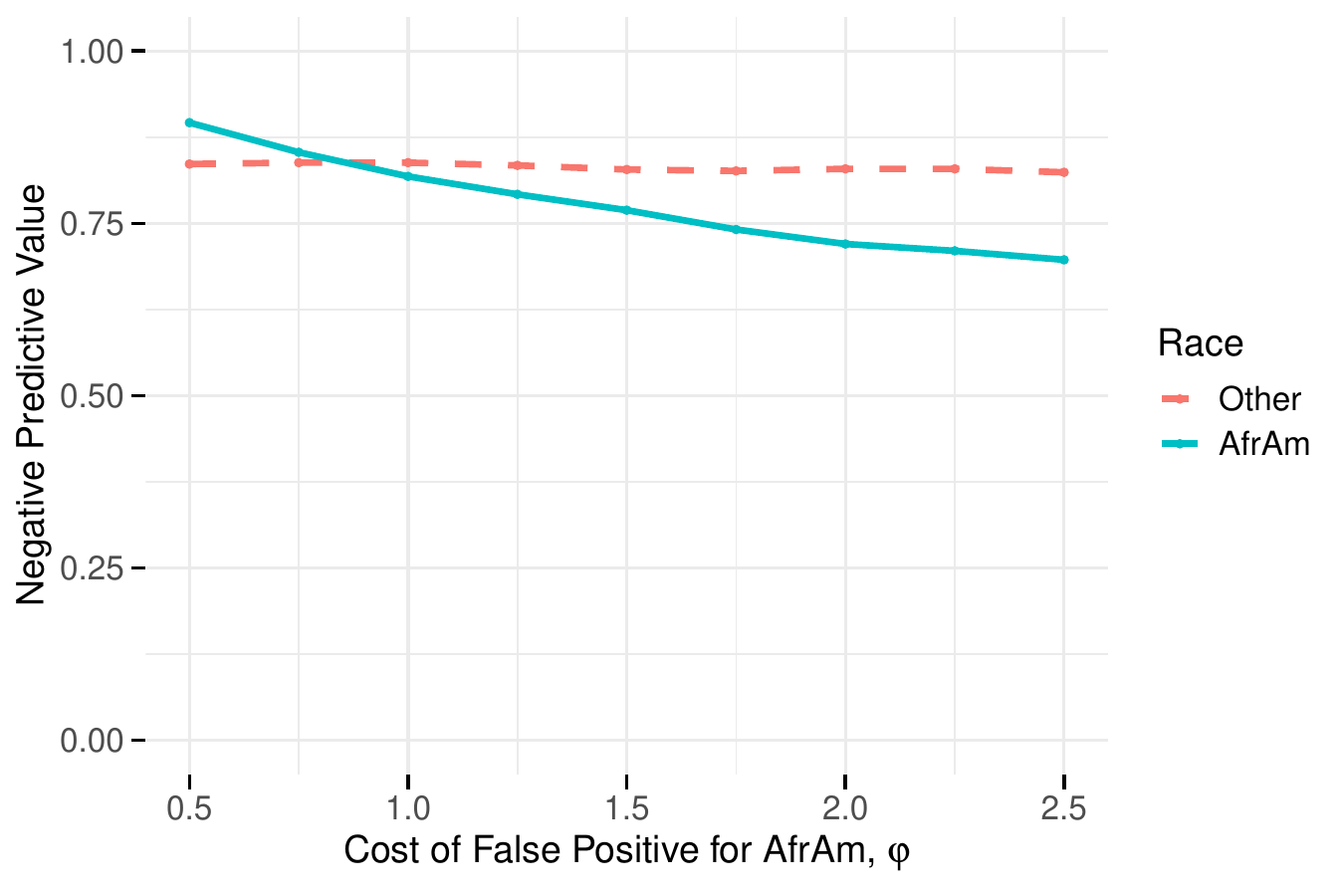}
		\caption{LASSO-Logit: NPV} 
	\end{subfigure}
	\begin{subfigure}{0.45\textwidth} 
		\includegraphics[width=\textwidth]{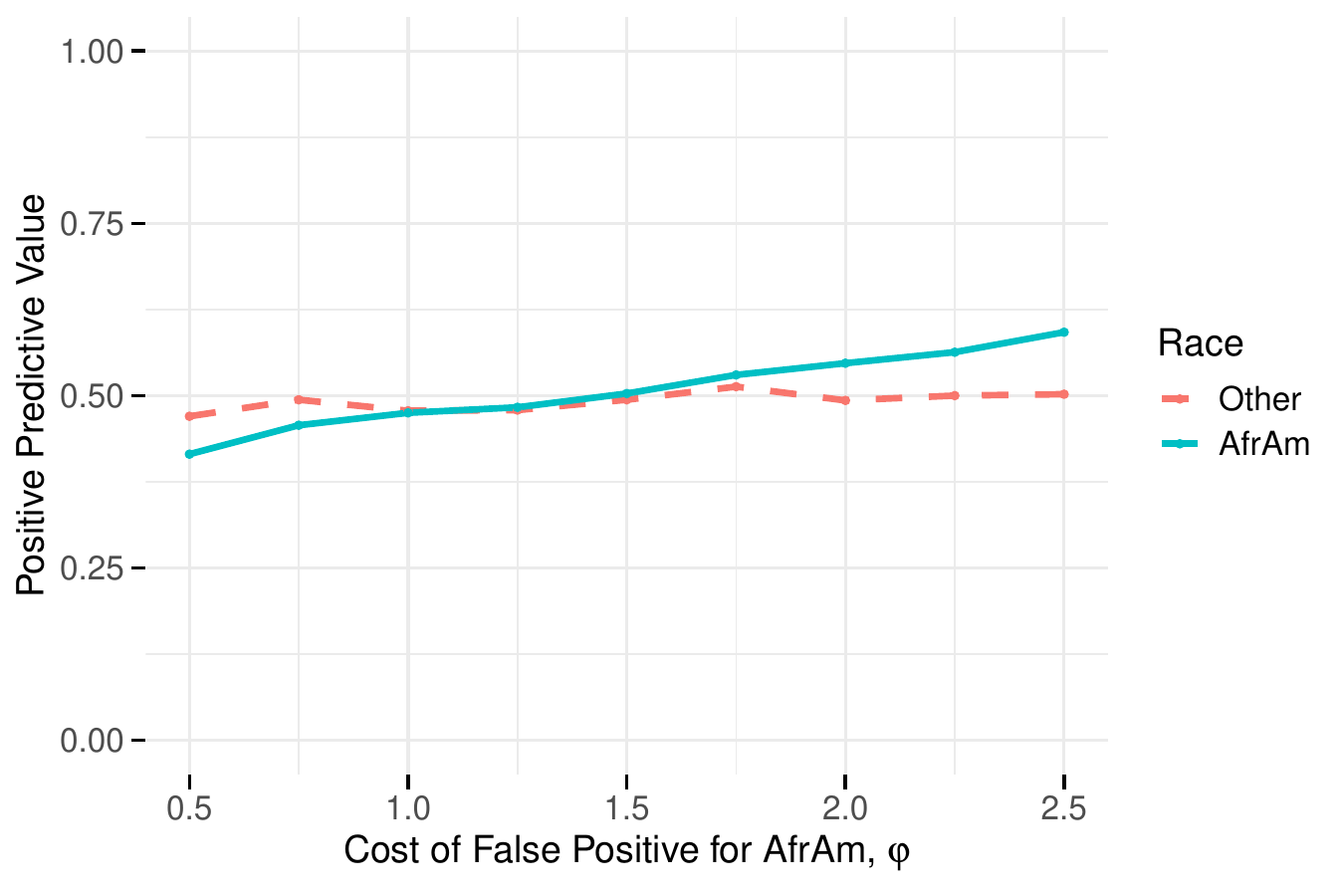}
		\caption{Boosting: PPV} 
	\end{subfigure}
	\begin{subfigure}{0.45\textwidth} 
		\includegraphics[width=\textwidth]{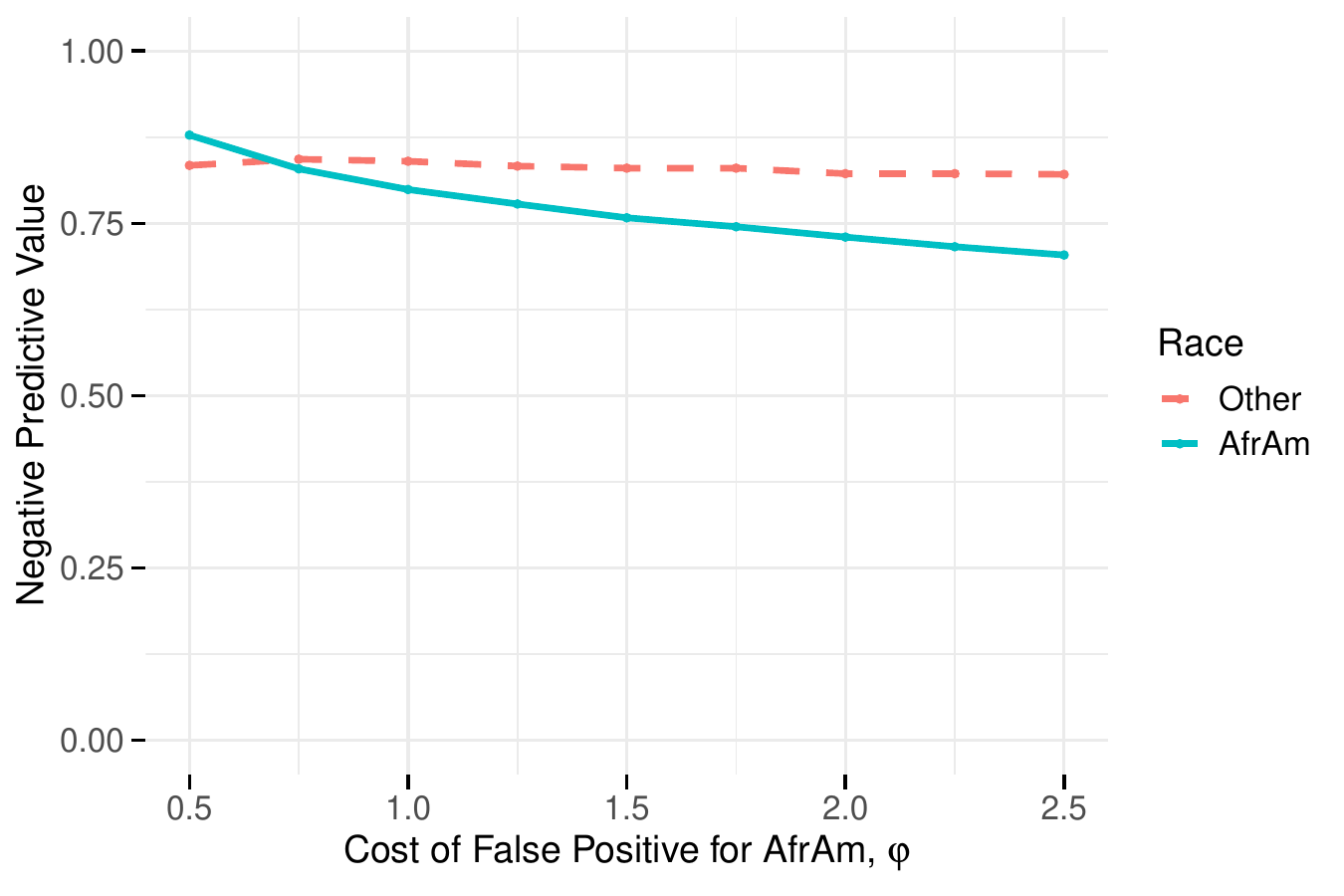}
		\caption{Boosting: NPV} 
	\end{subfigure}
	\caption{Fairness: Balanced Positive Predictive Values (PPV). The figure shows that increasing the cost of false positive mistakes for African Americans allows us to balance PPV across the two groups. We use the cost of false negatives $\psi=2.25$ and vary the cost of false positives for African American defendants $\varphi_1\in[0.5,2.5]$.} 
	\label{fig:fairness_ppv}
\end{figure}

Next, we look at balancing the positive predictive values (PPV) across the two groups by varying $\varphi_1$. The results appear in Figure~\ref{fig:fairness_ppv}. We find that the PPV for African Americans and others can be equalized if the decision-maker imposes higher costs for false positive mistakes for African American defendants. The PPVs for two groups are equalized when $\varphi_1\approx 1.25$ for both methods.\footnote{This fairness measure roughly corresponds to what the COMPAS algorithm was originally designed to achieve; see \cite{dieterich2016compas} for Northpointe's reply to the ProPublica article.} Therefore, varying $\varphi_1$ allows us to mitigate the disparate impact problem, achieving balanced PPV for the two groups. We also note that this type of fairness comes with a cost of increased negative predictive values (NPV) for Others. In other words, to balance the PPVs, we need to increase the likelihood that an African American defendant predicted to reoffend actually does so, which leads to a larger likelihood that Other defendants predicted not to reoffend actually do so.\footnote{See \cite{liang2021algorithm} and \cite{liu2024inference} for recent work on quantifying the fairness-accuracy frontier.}

\smallskip

In conclusion, algorithmic fairness can be achieved by placing a higher value on the outcomes of protected groups within our loss-based framework. This approach is transparent and can be directly incorporated into the digital decision process. However, the exact value of the loss function parameters that achieves balancing depends on the classifier and the fairness notion used; hence, it should be carefully calibrated in practice. It is impossible to achieve all fairness notions simultaneously, cf. \cite{chouldechova2017fair}.

\section{Conclusions}
This paper provides a new perspective on data-driven binary decisions with losses/utilities/welfare driven by economic factors, and contributes more broadly to the growing literature at the intersection of econometrics and machine learning. We show that the economic costs and benefits lead to a very simple loss-based reweighting of logistic regression or state-of-the art ML techniques. Our approach constitutes a significant advantage relative to others previously considered in the literature and leads to theoretically justified binary decisions for high-dimensional datasets frequently encountered in  practice. We adopt a distribution-free approach and show that the loss-based reweighted logistic regression may lead to valid binary decisions even when the choice probabilities are not logistic. We also show that the carefully crafted asymmetric deep learning architectures are optimal from the minimax point of view.

Our work opens several directions for future research. First, it calls for a more disciplined application of ML classification methods to economic decisions based on the costs/benefit considerations. For example, the false positive and false negative mistakes in pretrial detention decisions can be associated with different monetary costs for the society while the correct decisions may also have additional social benefits. In our empirical application we opted for a stylized cost differentiation scheme to make the empirical exercise transparent, but in practice one can use costs that are covariate-driven and based on a comprehensive cost-benefit analysis for the U.S.\ pretrial detention decision, such as in \cite{baughman2017costs}, to build a preference-based approach for this particular application.  Second, various economic applications often involve multi-class decisions and to that end some of the results of \cite{farrell2018deep} could potentially be extended to the asymmetric case provided that a suitable convexification result is established for generic covariate-driven losses.

\newpage
\setcounter{page}{1}
\setcounter{section}{1}
\setcounter{equation}{0}
\setcounter{table}{0}
\setcounter{figure}{0}
\renewcommand\thepage{Appendix - \arabic{page}}
\renewcommand{\theequation}{A.\arabic{equation}}
\renewcommand\thetable{A.\arabic{table}}
\renewcommand\thefigure{A.\arabic{figure}}
\renewcommand\thesection{A.\arabic{section}}
\renewcommand\thetheorem{A.\arabic{theorem}}

\begin{center}
	{\LARGE\textbf{APPENDIX}}	
\end{center}
\bigskip

In this Appendix we provide the proofs of all the main theorems, lemmas, and propositions. We will use $a\lesssim b$ if $a\leq Cb$ for a constant $C<\infty$ that does not depend on a particular distribution in the class of distributions restricted by our assumptions.

\begin{proof}[Proof of Proposition~\ref{prop:risk_representation}]
	Note that for every $f,y\in\{-1,1\}$ and $x\in\mathcal{X}$
	\begin{equation*}
		\begin{aligned}
			\ell(f,y,x) & = \ell_{1,1}(x)\frac{(1+f)(1+y)}{4} + \ell_{1,-1}(x)\frac{(1+f)(1-y)}{4} \\
			& \qquad + \ell_{-1,1}(x)\frac{(1-f)(1+y)}{4} + \ell_{-1,-1}(x)\frac{(1-f)(1-y)}{4} \\
			& = -0.25(a(x) + yb(x))f + d_0(y,x)
		\end{aligned}
	\end{equation*}
	with $d_0(y,x) \triangleq 0.25(\ell_{1,1}(x) + \ell_{-1,1}(x))(1+y) + 0.25(\ell_{1,-1}(x) + \ell_{-1,-1}(x))(1-y)$. Next, for every $(y,x)\in\{-1,1\}\times\mathcal{X}$ and every $f\in\{-1,1\}$
	\begin{equation*}
		\begin{aligned}
			(a(x) + yb(x))f & = (ya(x) + b(x))(1-2\one_{-yf\geq 0}) \\
			& = -2(ya(x) + b(x))\one_{-yf\geq 0} + ya(x) + b(x).
		\end{aligned}
	\end{equation*}
	Therefore, for every measurable $f:\mathcal{X}\to\{-1,1\}$
	\begin{equation}\label{eq:risk_app}
		\mathcal{R}(f) = 0.5\E[(Ya(X)+b(X))\one_{-Yf(X)\geq 0}] + \E[d(Y,X)].
	\end{equation}
	with $d(y,x)\triangleq d_0(y,x) - 0.25(ya(x) + b(x))$. The result follows because the second term does not depend on $f$.
\end{proof}

\begin{proof}[Proof of Proposition~\ref{prop:bayes}]
	By the law of iterated expectations, $f^*(x)$ minimizes
	\begin{equation*}
		\E\left[\omega(Y,X)\one_{-Yf\geq 0}|X=x\right] = \eta(x)\omega(1,x)\one_{f\leq 0} + (1-\eta(x))\omega(-1,x)\one_{f\geq 0}
	\end{equation*}	
	over $f\in\{-1,1\}$. The solution to this problem is
	\begin{equation*}
		f^*(x) = \begin{cases}
			1 &  \text{if}\quad \eta(x)\omega(1,x) \geq (1-\eta(x))\omega(-1,x), \\
			-1 & \text{if}\quad \eta(x)\omega(1,x) < (1-\eta(x))\omega(-1,x).
		\end{cases}
	\end{equation*}
	Under Assumption~\ref{as:losses} (i), $\omega(Y,X) > 0$ a.s., whence we obtain the first statement	with
	\begin{equation*}
		c(x) = \frac{\omega(-1,x)}{\omega(1,x) + \omega(-1,x)}.
	\end{equation*}
	
	For the second statement, note that $f_\phi^*(x)$ solves
	$\inf_{f\in\R}\E[\omega(Y,X)\phi(-Yf)|X=x].$
	By the law of iterated expectations, for every $f\in\R$, the objective function is
	\begin{equation*}
		\E[\omega(Y,X)\phi(-Yf)|X=x] = \eta(x)\omega(1,x)\phi(-f) +  (1-\eta(x))\omega(-1,x)\phi(f).
	\end{equation*}
	Since $\phi$ is differentiable, the optimum $f^*_\phi$ solves
	$-\eta(x)\omega(1,x)\phi'(-f^*_\phi(x))$ + $(1-\eta(x))\omega(-1,x)\phi'(f^*_\phi(x))$ = 0.	Under Assumption~\ref{as:losses} (ii), $\eta\not\in\{0,1\}$, and therefore:
	\begin{equation*}
		\frac{\eta(x)\omega(1,x)}{(1-\eta(x))\omega(-1,x)} = \frac{\phi'(f^*_\phi(x))}{\phi'(-f^*_\phi(x))}.
	\end{equation*}
	Since $\phi$ is a convex function, its derivative $\phi'$ is non-decreasing. Therefore, the optimum satisfies
	\begin{equation*}
		\begin{aligned}
			f_\phi^*(x) \geq 0 & \iff \frac{\eta(x)\omega(1,x)}{(1-\eta(x))\omega(-1,x)} \geq 1 \\
			& \iff \eta(x) \geq \frac{\omega(-1,x)}{\omega(1,x) + \omega(-1,x)} = c(x).
		\end{aligned}
	\end{equation*}
\end{proof}

\begin{proof}[Proof of Theorem~\ref{thm:convexified_risk_bound_margin}]
	Since $\gamma\in(0,1]$, the function $g:x\mapsto x^{1/\gamma}$ is convex on $\R_+$. By \cite{bartlett2006convexity}, Lemma 1, $y^{1/\gamma}\leq x^{1/\gamma}y/x,\forall x\geq y\geq 0,x>0$. Setting $x=|\eta(X)-c(X)|$ and $y=\epsilon$ for some $\epsilon>0$, we obtain
	\begin{equation}\label{eq:inequality_margin}
		|\eta(X)-c(X)|\one_{|\eta(X)-c(X)|\geq\epsilon} \leq \epsilon^{1-1/\gamma}|\eta(X)-c(X)|^{1/\gamma}.
	\end{equation}
	Under Assumption~\ref{as:losses} (i), by Online Appendix Lemma~\ref{lemma:risk_eta_c} for every $f:\mathcal{X}\to\R$ and $\epsilon>0$,
	\begin{equation*}
		\begin{aligned}
			& \mathcal{R}(\sign(f)) - \mathcal{R}^* = \\
			& = \E_{\sign(f)f^*<0}\left[b(X)|\eta(X)-c(X)|\right] \\
			& = \E_{\sign(f)f^*<0}\left[b(X)|\eta(X)-c(X)|\one_{|\eta(X)-c(X)|<\epsilon}\right] \\
			& \qquad +\E_{\sign(f)f^*<0}\left[b(X)|\eta(X)-c(X)|\one_{|\eta(X)-c(X)|\geq\epsilon}\right] \\
			& \leq 4M\epsilon  P_X(\{x:\; \sign(f(x))f^*(x)< 0 \}) + \epsilon^\frac{\gamma - 1}{\gamma}\E_{\sign(f)f^*<0}\left[b(X)|\eta(X)-c(X)|^{1/\gamma}\right] \\
			& \leq 4M c_b^{-\frac{\alpha}{1+\alpha}}(2C_m)^\frac{1}{1+\alpha}\epsilon\left[\mathcal{R}(\sign(f)) - \mathcal{R}^* \right]^\frac{\alpha}{1+\alpha} + \epsilon^\frac{\gamma - 1}{\gamma}C^{1/\gamma}\left[\mathcal{R}_\phi(f) - \mathcal{R}_\phi^*\right],
		\end{aligned}
	\end{equation*}
	where the first inequality follows from $|b(X)|\leq 4M$ a.s.\ under Assumptions~\ref{as:losses} (iii) and inequality~(\ref{eq:inequality_margin}); and the second by Lemma~\ref{lemma:margin} under Assumptions~\ref{as:losses} (i) and \ref{as:tsybakov}, and
	\begin{equation*}
		\begin{aligned}
			& \E_{\sign(f)f^*<0}\left[b(X)|\eta(X)-c(X)|^{1/\gamma}\right] \\
			& \leq C^{1/\gamma}\E_{\sign(f)f^*<0}\left[b(X)\left(\eta(X) + c(X) - 2\eta(X)c(X) - \inf_{y\in\R}Q_{c(X)}(\eta(X),y)\right)\right]  \\
			& \leq C^{1/\gamma}\left[\mathcal{R}_\phi(f) - \mathcal{R}_\phi^*\right],
		\end{aligned}
	\end{equation*}
	which in turn follows by Assumption~\ref{as:phi} (iii) and the Online Appendix equation (\ref{eq:convex_inequality}). Therefore,
	\begin{equation*}
		\mathcal{R}(\sign(f)) - \mathcal{R}^* \leq 4M c_b^{-\frac{\alpha}{1+\alpha}}(2C_m)^\frac{1}{1+\alpha}\epsilon\left[\mathcal{R}(\sign(f)) - \mathcal{R}^* \right]^\frac{\alpha}{1+\alpha} + \epsilon^\frac{\gamma - 1}{\gamma}C^{1/\gamma}\left[\mathcal{R}_\phi(f) - \mathcal{R}_\phi^*\right].
	\end{equation*} 
	Setting $\epsilon = 0.5(4M)^{-1}c_b^\frac{\alpha}{1+\alpha}(2C_m)^{-\frac{1}{1+\alpha}}\left[\mathcal{R}(\sign(f)) - \mathcal{R}^*\right]^{\frac{1}{1+\alpha}}$ and rearranging, we obtain
	\begin{equation*}
		\mathcal{R}(\sign(f)) - \mathcal{R}^* \leq C_\phi\left[\mathcal{R}_\phi(f) - \mathcal{R}_\phi^*\right]^\frac{\gamma(1+\alpha)}{\gamma\alpha + 1}
	\end{equation*}
	with $C_\phi = (2^\gamma C)^\frac{1+\alpha}{\gamma\alpha+1}\left[c_b^\alpha/(2^{3\alpha+4}C_mM^{\alpha+1})\right]^\frac{\gamma-1}{\gamma\alpha+1}$.
\end{proof}

\begin{proof}[Proof of Theorem \ref{thm:oracle inequality_dnn}]
	Under Assumption~\ref{as:deep_learning}, by \cite{lu2020deep}, Theorem 1.1, there exists a deep neural network $\eta_n\in\Theta^{\rm DNN}_n$, with width $W_n \leq 17\beta^{d+1}3^dd(J_n+2)\log_2(8J_n)$ and depth $L_n \leq 18\beta^2(K_n+2)\log_2 (4K_n) + 2d$ such that $\|\eta_n - \eta\|_\infty \leq 85(\beta+1)^d8^\beta R\left(\frac{\log^6 n}{n}\right)^\frac{\beta}{\beta(2+\alpha)+d}\triangleq  \varepsilon_n,$ where we use that under Assumption~\ref{as:deep_learning} (ii) $J_nK_n\leq (n/\log^6n)^\frac{d}{2\beta(2+\alpha)+2d}$. Put $f_n(x) \triangleq \sigma\left(\frac{\eta_n(x)-c(x)}{\varepsilon_n} + 1 \right)
	-\sigma\left(\frac{\eta_n(x)-c(x)}{\varepsilon_n} - 1\right) - 1.$ Then
	\begin{equation*}
		f_n(x) = \begin{cases}
			1 & \text{if  } \eta_n(x)-c(x)> \varepsilon_n, \\
			\frac{\eta_n(x)-c(x)}{\varepsilon_n}  & \text{if  } |{\eta}_n(x) - c(x)|\leq \varepsilon_n, \\
			-1 & \text{if  } \eta_n(x)-c(x)<- \varepsilon_n.
		\end{cases}
	\end{equation*}
	Then on the event $\{x\in\mathcal{X}:|\eta(x) - c(x)|>2\varepsilon_n\}$, we have $f_n(x)=f_\phi^*(x)$. To see this note that $f_\phi^*(x)=\sign(\eta(x)-c(x))$ and that if $\eta(x)> c(x)$, we have $\eta_n(x)-c(x) = (\eta(x)-c(x))-({\eta}(x)-\eta_n(x))\geq \varepsilon_n$ while if $\eta(x)<c(x)$, we have $\eta_n(x)-c(x)<-\varepsilon_n$. Therefore, by Lemma~\ref{lemma:equation_hinge}
	\begin{equation*}
		\begin{aligned}
			\inf_{f\in\mathscr{F}_n^{\rm DNN}}{\mathcal{R}}_{\phi}(f) - \mathcal{R}_\phi^* & \leq {\mathcal{R}}_{\phi}(f_n) - \mathcal{R}_\phi^* = \int_\mathcal{X}b|\eta - c||f_n - f^*_\phi|\dx P_X \\
			& \leq 4M\int_{|\eta-c|\leq 2\varepsilon_n}|\eta - c||f_n - f^*_\phi|\dx P_X  \leq16M \varepsilon_n P_X(|\eta - c|\leq 2\varepsilon_n) \\
			&  \leq 2^{4+\alpha}MC_m\varepsilon_n^{1+\alpha},
		\end{aligned}
	\end{equation*}
	where the third line follows since $|b(X)|\leq 4M$ under Assumption~\ref{as:losses} (iii); the fourth since $|f_n-f_\phi^*|\leq 2$ in the region of integration; and the last under Assumption~\ref{as:tsybakov}.
	
	Let $p_n$ be the total number of free parameters in $\Theta_n^{\rm DNN}$ and let $U_n$ be the total number of nodes in $\Theta_n^{\rm DNN}$. Note that $U_n\leq L_nW_n$ and that $p_n = \sum_{l=1}^{L-1}(w_{l+1}\times w_l + w_l) + w_1\times w_0 \leq L_nW_n(W_n + 1) + W_n^2 \leq 3L_nW_n^2$. Then by \cite{bartlett2019nearly}, Theorem 7, there exist universal constant $C_0>0$ such that $V \leq C_0 p_nL_n\log(U_n) \leq 3C_0(L_nW_n)^2\log(W_nL_n)$. Therefore, by Online Appendix Lemmas~\ref{lemma:lrc_deep_learning} and \ref{lemma:lrc_fp}
	\begin{equation*}
		\begin{aligned}
			\psi_{n,\kappa}^{\sharp}(\epsilon) & \leq C(3C_0)^\frac{1+\alpha}{2+\alpha}\left(\frac{(W_nL_n)^2\log (W_nL_n)}{n}\log\left(\frac{n}{3C_0(W_nL_n)^2\log (W_nL_n)}\right)\right)^\frac{1+\alpha}{2+\alpha}.
		\end{aligned}
	\end{equation*}
	Since under Assumption~\ref{as:deep_learning} (ii), $W_nL_n\leq abC_LC_W(n/\log^6 n)^\frac{d}{2\beta(2+\alpha)+2d}\log^2n$, by Theorem~\ref{thm:oracle inequality} for every $t>0$ with probability at least $1-c_qe^{-t}$
	\begin{align*}\small
		\mathcal{R}(\sign(\hat f_n)) - \mathcal{R}^* & \leq K\left(\frac{\log^6 n}{n}\right)^\frac{\beta(1+\alpha)}{\beta(2+\alpha)+d} + \left(\frac{t}{n}\right)^\frac{1+\alpha}{2+\alpha} + \frac{t}{n} \\
		&\qquad + 2^{4+\alpha}MC_m\left(85(\beta+1)^d8^\beta R\left(\frac{\log^6 n}{n}\right)^\frac{\beta}{\beta(2+\alpha)+d}\right)^{1+\alpha}
	\end{align*}
	where $K>0$ depends only on $\alpha,\beta,d,C_0,C_1,C>0$. Integrating the tail bound, we obtain the second claim.
\end{proof}

\begin{lemma}\label{lemma:risk_eta_c}
	Suppose that Assumption~\ref{as:losses} (i) is satisfied. Then for every measurable $f:\mathcal{X}\to\R$, the excess risk is
	\begin{equation*}
		\mathcal{R}(\mathrm{sign}(f)) - \mathcal{R}^* = \E_{\sign(f)f^*<0}\left[b(X)\left|\eta(X)-c(X)\right|\right],
	\end{equation*}
	where for an event $A$ and a random variable $\xi$, we put $\E_A\xi = \E\one_A\xi$.
\end{lemma}
\begin{proof}
	By the law of iterated expectations
	\begin{align*}
		\E\left[(Ya(X)+b(X))\one_{-Y\mathrm{sign}(f)\geq 0}\right] &= \E\left[\eta(X)(a(X)+b(X))\one_{\mathrm{sign}(f)\leq 0}\right]\\
		&\qquad+\E[(1-\eta(X))(b(X)-a(X))\one_{\mathrm{sign}(f)\geq 0}]\\
		&=\E[( a(X)-b(X) + 2\eta(X)b(X))\one_{\mathrm{sign}(f)\leq 0}]\\
		&\qquad + \E[(1-\eta(X))(b(X)-a(X))].
	\end{align*}
	Combining this observation with equation (\ref{eq:risk_app}) yields
	\begin{equation*}\label{eq:risk_eta_c}
		\begin{aligned}
			2(\mathcal{R}(\mathrm{sign}(f)) - \mathcal{R}^*) & = \E\left[(Ya(X)+b(X))(\one_{-Y\mathrm{sign}(f)\geq 0} - \one_{-Yf^*(X)\geq 0})\right] \\
			& = \E\left[( a(X)-b(X) + 2\eta(X)b(X))(\one_{\mathrm{sign}(f)\leq 0} - \one_{f^*(X)\leq 0})\right] \\
			& = \E\left[2b(X)(\eta(X)-c(X))(\one_{\mathrm{sign}(f)\leq 0} - \one_{f^*(X)\leq 0})\right] \\
			& =  \E_{\mathrm{sign}(f)f^*<0}\left[2b(X)\left|\eta(X)-c(X)\right|\right],
		\end{aligned}
	\end{equation*}
	where the last line follows since under Assumption~\ref{as:losses} (i), by Proposition~\ref{prop:bayes}
	\begin{equation*}
		f^*(x)\geq 0 \iff \eta(x) \geq \frac{b(x)-a(x)}{2b(x)} = c(x).
	\end{equation*}
\end{proof}

\begin{lemma}\label{lemma:convexified_risk_bound}
	Suppose that Assumptions~\ref{as:losses} and \ref{as:phi} are satisfied. Then for every measurable function $f:\mathcal{X}\to\R$
	\begin{equation*}
		\mathcal{R}(\sign(f)) - \mathcal{R}^* \leq 2^{\frac{2}{\gamma}-2}CM^{\frac{1}{\gamma}-1} \left[\mathcal{R}_\phi(f) - \mathcal{R}_\phi^*\right]^\gamma.
	\end{equation*}
\end{lemma}
\begin{proof}
	Under Assumption~\ref{as:losses} (i), by Lemma~\ref{lemma:risk_eta_c}
	\begin{equation*}
		\begin{aligned}
			& \mathcal{R}(\sign(f)) - \mathcal{R}^*= \\
			& = \E_{\mathrm{sign}(f)f^*<0}b(X)\left|\eta(X)-c(X)\right| \\
			& \leq C\E_{\mathrm{sign}(f)f^*<0}\left[b(X)\left(\eta(X)+c(X)-2\eta(X)c(X) -\inf_{y\in\R}Q_{c}(\eta,y)\right)^\gamma\right] \\
			& \leq C\left(\E_{\mathrm{sign}(f)f^*<0}\left[b(X)^{1/\gamma}\left( \eta(X)+c(X)-2\eta(X)c(X) - \inf_{y\in\R}Q_{c}(\eta,y)\right)\right]\right)^\gamma\\
			& \leq C(4M)^\frac{1-\gamma}{\gamma} \left(\E_{\mathrm{sign}(f)f^*<0}\left[b(X)\left(\eta(X)+c(X)-2\eta(X)c(X) - \inf_{y\in\R}Q_{c}(\eta,y)\right)\right]\right)^\gamma,
		\end{aligned}
	\end{equation*}
	where the second line follows under Assumption~\ref{as:phi} (iii) since $\eta,c\in(0,1)$ under Assumptions~\ref{as:losses} (i)-(ii); the third by Jensen's inequality since $\gamma\in(0,1]$ under Assumption~\ref{as:phi} (iii); and the last since $|b(X)|\leq 4M$ a.s. under Assumption~\ref{as:losses} (ii). Next, from equation~(\ref{eq:excess_convex}) in the paper, we have
	\begin{equation*}
		\mathcal{R}_\phi(f) - \mathcal{R}_\phi^* = 0.5\E\left[(Ya(X) + b(X))\left(\phi(-Yf(X)) - \phi(-Yf^*_\phi(X))\right)\right].
	\end{equation*}
	Therefore, if we show that
	\begin{equation}\label{eq:convex_inequality}
		\begin{aligned}
			& b(x)\one_{\sign(f(x))f^*(x)< 0}\left(\eta(x)+c(x)-2\eta(x)c(x) - \inf_{y\in\R}Q_{c(x)}(\eta(x),y)\right) \\
			& \qquad\leq  0.5\E[(Ya(X)+b(X))(\phi(-Yf(X)) - \phi(-Yf^*_\phi(X)))|X=x] \\
		\end{aligned}
	\end{equation}
	the result will follow from integrating over $x$. To that end, if $\sign(f(x))f^*(x)<0$, then the inequality in equation (\ref{eq:convex_inequality}) follows trivially since by definition $f^*_\phi$ minimizes $f\mapsto \E[(Ya(X)+b(X))\phi(-Yf(X))]$. Suppose now that $\sign(f(x))f^*(x)\geq0$. Then by the law of iterated expectations and the definition of $Q$ in the Assumption~\ref{as:phi} (iii)
	\begin{equation*}
		\E[(Ya(X)+b(X))\phi(-Yf^*_\phi(X))|X=x] = 2b(x)\inf_{y\in\R}Q_{c(x)}(\eta(x),y).
	\end{equation*}
	Therefore, the inequality in equation (\ref{eq:convex_inequality}) follows if we can show that
	$$2b(x)(\eta(x)+c(x)-2\eta(x)c(x))\leq \E[(Ya(X)+b(X))\phi(-Yf(X))|X=x].$$ 
	But this follows from
	\begin{equation*}
		\begin{aligned}
			&\E[(Ya(X)+b(X))\phi(-Yf(X))|X=x] \\
			&=\eta(x)(a(x)+b(x))\phi(-f(x))+(1-\eta(x))(b(x)-a(x))\phi(f(x)) \\
			&=2b(x)[\eta(x)(1-c(x))\phi(-f(x))+(1-\eta(x))c(x)\phi(f(x))] \\
			& \geq 2b(x)(\eta(x)+c(x)-2\eta(x)c(x))\phi\left(\frac{f(x)(c(x)-\eta(x))}{\eta(x)+c(x)-2\eta(x)c(x)} \right) \\ 
			& \geq 2b(x)(\eta(x)+c(x)-2\eta(x)c(x))\phi(0)\\
			& = 2b(x)(\eta(x)+c(x)-2\eta(x)c(x)),
		\end{aligned}
	\end{equation*}
	where the first inequality follows by the convexity of $\phi$ under Assumption~\ref{as:phi} (i) and since $\eta+c-2\eta c>0$ under Assumption~\ref{as:losses} (i)-(ii); the second inequality since $\phi$ is nondecreasing under Assumption~\ref{as:phi} (i) and $0\leq \sign(f(x))f^*(x) = \sign(f(x))\mathrm{sign}(c(x)-\eta(x))$ by Proposition~\ref{prop:bayes}; and the last line since $\phi(0)=1$ under Assumption~\ref{as:phi} (i).
\end{proof}

\begin{lemma}\label{lemma:logistic}
	For the logistic convexifying function $\phi(z)=\log_2(1 + e^z)$,
	\begin{equation*}
		f_\phi^*(x) = \log\left(\frac{\eta(x)(1-c(x))}{(1-\eta(x))c(x)}\right).
	\end{equation*}
	Assumption~\ref{as:phi} is satisfied with $\gamma=1/2$ and $C=\sqrt{2\log 2}$. 
\end{lemma}
\begin{proof}
	Note that the minimum of $y\mapsto Q_{c}(x,y)$ is achieved at $y^*$ = $\log\left(\frac{x(1-c)}{(1-x)c}\right).$ Put $\eta_y\triangleq\frac{y(1-c)}{y(1-c)+(1-y)c}$ and $\eta \triangleq \eta_x$, and note that
	\begin{align*}
		\inf_{y\in\R} Q_{c}(x,y) & = x(1-c)\log_2(1+e^{-y^*}) + (1-x)c\log_2(1+e^{y^*})\\
		& =-x(1-c)\log_2\frac{x(1-c)}{x(1-c)+(1-x)c}-(1-x)c\log_2 \frac{c(1-x)}{x(1-c)+(1-x)c}\\
		& =(x(1-c)+(1-x)c)\left[-\eta\log_2\eta-(1-\eta)\log_2(1-\eta)\right].
	\end{align*}
	For every $y\in[0,1]$, put
	{\footnotesize\begin{equation*}
			\begin{aligned}
				L(y) & \triangleq (x(1-c)+(1-x)c)\Big[\eta\log_2\left({\eta}\right) + (1-\eta)\log_2\left({1-\eta}\right)-\eta\log_2\left({\eta_y}\right) - (1-\eta)\log_2\left({1-\eta_y}\right)\Big] \\
				& = (x(1-c)+(1-x)c)\frac1{\log 2}\left[\eta\log\left(\frac{\eta}{\eta_y}\right) + (1-\eta)\log\left(\frac{1-\eta}{1-\eta_y}\right)\right]
			\end{aligned}
	\end{equation*}}
	and note that the equality in the first line implies that
	$$
	L(c)=(x+c-2xc) - \inf_{y\in\R} Q_{c}(x,y).
	$$
	By Taylor's theorem there exists $\eta'$ between $\eta$ and $\eta_y$ such that
	\begin{align*}
		\eta\log\left(\frac{\eta}{\eta_y}\right) + (1-\eta)\log\left(\frac{1-\eta}{1-\eta_y}\right) &= \frac{1}{2\eta'(1-\eta')}(\eta-\eta_y)^2\\
		&\geq 2(\eta-\eta_y)^2,
	\end{align*}
	since $\eta'\in[0,1]$. This shows that for every $y\in[0,1]$
	\begin{align*}
		L(y)&\geq \frac{2}{\log 2}(x(1-c)+(1-x)c)(\eta-\eta_y)^2.
	\end{align*}
	In particular, for $y=c$,
	\begin{align*}
		L(c)&\geq \frac{2}{\log 2}(x(1-c)+(1-x)c)(\eta-\eta_c)^2\\
		&= \frac{1}{2\log 2}(x(1-c)+(1-x)c)\left(\frac{x-c}{x(1-c) + (1-x)c}\right)^2 \\
		&\geq \frac{1}{2\log 2}(x-c)^2.
	\end{align*}
	Therefore, Assumption is verified with $\gamma=1/2$ and $C=\sqrt{2\log 2}$.
\end{proof}

\begin{lemma}\label{lemma:hinge}
	For the hinge convexifying function $\phi(z)=(1+z)_+$,
	\begin{equation*}
		f_\phi^*(x) = \begin{cases}
			1 & \text{if }\; \eta(x)\geq c(x), \\
			-1 & \text{if }\; \eta(x)<c(x).
		\end{cases}
	\end{equation*}
	Assumption~\ref{as:phi} is satisfied with $\gamma=1$ and $C=1$.
\end{lemma}
\begin{proof}
	Note that the minimum of $y\mapsto Q_c(x,y)$ is achieved at
	\begin{equation*}
		\begin{cases}
			1 & \text{if }\; x>c, \\
			-1 & \text{if }\; x<c.
		\end{cases}
	\end{equation*}
	Then
	\begin{equation*}
		\begin{aligned}
			\inf_{y\in\R}Q_c(x,y)	& = \inf_{y\in\R} x(1-c)(1-y)_+ + c(1-x)(1+y)_+ \\
			& = \min\{2x(1-c),2c(1-x)\}.
		\end{aligned}
	\end{equation*}
	If $x\leq c$, then
	\begin{equation*}
		\begin{aligned}
			(x+c-2xc) - \inf_{y\in\R}Q_c(x,y) & = (x+c-2xc) - 2x(1-c) \\
			& = (c-x).
		\end{aligned}
	\end{equation*}
	If $x>c$, then
	\begin{equation*}
		\begin{aligned}
			(x+c-2xc) - \inf_{y\in\R}Q_c(x,y) & = (x+c-2xc) - 2c(1-x) \\
			& = (x-c).
		\end{aligned}
	\end{equation*}
	Therefore,
	\begin{equation*}
		(x+c-2xc) - \inf_{y\in\R}Q_c(x,y) = |x-c|
	\end{equation*}
	and Assumption~\ref{as:phi} is satisfied with $\gamma=1$ and $C=1$.
\end{proof}

\begin{lemma}\label{lemma:exponential}
	For the exponential convexifying function $\phi(z)=e^z$,
	\begin{equation*}
		f_\phi^*(x) = \frac{1}{2}\log\left(\frac{\eta(x)(1-c(x))}{(1-\eta(x))c(x)}\right).
	\end{equation*}
	Assumption~\ref{as:phi} is satisfied with $\gamma=1/2$ and $C=2$.
\end{lemma}
\begin{proof}
	Note that the minimum of $y\mapsto Q_c(x,y)$ is achieved at
	\begin{equation*}
		y^* = \frac{1}{2}\log\left(\frac{x(1-c)}{(1-x)c}\right).
	\end{equation*}	
	Then
	\begin{equation*}
		\begin{aligned}
			\inf_{y\in\R}Q_c(x,y)	& = \inf_{y\in\R} x(1-c)e^{-y} + c(1-x)e^y \\
			& = 2\sqrt{xc(1-x)(1-c)}.
		\end{aligned}
	\end{equation*}
	Then
	\begin{equation*}
		\begin{aligned}
			(x+c-2xc) - \inf_{y\in\R}Q_c(x,y) & = (x+c-2xc)-2\sqrt{xc(1-x)(1-c)} \\
			& =\left(\sqrt{x(1-c)}-\sqrt{c(1-x)}\right)^2 \\
			& =\frac{(x-c)^2}{(\sqrt{x(1-c)}+\sqrt{c(1-x)})^2}\geq \frac{(x-c)^2}{4}.
		\end{aligned}
	\end{equation*}	
	where last line is due to the fact $x,c\in(0,1)$. Therefore, Assumption~\ref{as:phi} is satisfied with $\gamma=1/2$ and $C=2$.
\end{proof}

\begin{lemma}\label{lemma:margin}
	Suppose that Assumptions \ref{as:losses} (i) and \ref{as:tsybakov} are satisfied. Then for every measurable $f:\mathcal{X}\to\R$
	\begin{equation*}
		\mathcal{R}(\mathrm{sign}(f)) - \mathcal{R}^* \geq c_b(2C_m)^{-1/\alpha} P_X^{\frac{1+\alpha}{\alpha}}\left(\{x:\;\sign(f(x))f^*(x)< 0  \}\right).
	\end{equation*}
\end{lemma}
\begin{proof}
	By  Lemma~\ref{lemma:risk_eta_c}
	{\footnotesize
		\begin{equation*}
			\begin{aligned}
				\mathcal{R}(\mathrm{sign}(f)) - \mathcal{R}^* & = \E_{\sign(f)f^*<0}\left[b(X)|\eta(X)-c(X)|\right] \\
				& \geq 2c_b \int_\mathcal{X}\one_{\sign(f(x))f^*(x)<0}|\eta(x)-c(x)|\dx P_X(x) \\
				& \geq 2c_b uP_X\left(\{x:\; \sign(f(x))f^*(x)< 0\}\cap\{x:\; |\eta(x)-c(x)|>u \}\right) \\
				& \geq 2c_b uP_X\left(\{x:\; \sign(f(x))f^*(x)< 0\}\right) - 2c_b uP_X\left(\{x:\;|\eta(x)-c(x)|\leq u \}\right) \\
				& \geq 2c_b uP_X\left(\{x:\; \sign(f(x))f^*(x)< 0\}\right) - 2c_b C_mu^{1+\alpha},
			\end{aligned}
	\end{equation*}}
	where the first inequality follows under Assumption~\ref{as:losses} (i); the second by Markov's inequality for every $u>0$; the third by $\Pr(A\cap B)\geq \Pr(A) - \Pr(B^c)$; and the fourth under Assumption~\ref{as:tsybakov}. The result follows from substituting $u$ solving
	\begin{equation*}
		P_X\left(\{x:\; \sign(f(x))f^*(x)< 0\}\right) = 2C_mu^\alpha
	\end{equation*}
	in the last equation.
\end{proof}

\begin{lemma}\label{lemma:curvature_exp_log}
	Suppose that Assumptions~\ref{as:losses} (i)-(ii) are satisfied and that there exists a constant $F<\infty$ such that $|f|\leq F$ for all $f\in\mathscr{F}_n$. Then the exponential convexifying function satisfies Assumption~\ref{as:loss_modulus} with $\kappa = 1$ and $c_\phi=c_b\epsilon$, while for the logistic function we have $\kappa=1$ and $c_\phi = 2c_b\epsilon\phi''(F\vee \log((1-\epsilon)c_b/2\epsilon M))$.
\end{lemma}
\begin{proof}	
	First note that for every $f\in\mathscr{F}_n$ by the law of iterated expectations
	{\footnotesize
		\begin{equation*}
			\begin{aligned}
				\mathcal{R}_\phi(f) & = 0.5\E\left[(Ya(X) + b(X))\phi(-Yf(X))\right] + \E[d(Y,X)] \\
				& = 0.5\E\left[\eta(X)(a(X)+b(X))\phi(-f(X)) + (1-\eta(X) )(b(X)-a(X))\phi(f(X))\right] + \E[d(Y,X)] \\
				& = \E\left[b(X)\eta(X)(1-c(X))\phi(-f(X)) + b(X)(1-\eta(X))c(X)\phi(f(X))\right] + \E[d(Y,X)].
			\end{aligned}
	\end{equation*}}
	
	Then, since $f_\phi^*(x)$ minimizes $f\mapsto \eta(x)(1-c(x))\phi(-f) + (1-\eta(x) )c(x)\phi(f)$, by Taylor's theorem there exists $\tau\in[0,1]$ such that for $f_\tau \triangleq \tau f+(1-\tau)f_\phi^*$, we have
	{\footnotesize	
		\begin{equation*}
			\begin{aligned}
				\mathcal{R}_\phi(f) - \mathcal{R}_\phi^* & = \frac{1}{2}\E\left[b(X)\left[\eta(X)(1-c(X))\phi''(-f_\tau(X)) + (1-\eta(X))c(X)\phi''(f_\tau(X))\right]|f(X) - f^*_\phi(X)|^2\right]
			\end{aligned}
	\end{equation*}}
	Then for the exponential convexifying function since $\phi''(z)=e^{z}=\phi(z)$ and $f_\phi^*$ minimizes $f\mapsto \eta(1-c)\phi(-f) + (1-\eta)c\phi(f)$
	{\footnotesize 
		\begin{equation*}
			\begin{aligned}
				\mathcal{R}_\phi(f) - \mathcal{R}_\phi^* & =  \frac{1}{2}\E\left[b(X)\left[\eta(X)(1-c(X))\phi(-f_\tau(X)) + (1-\eta(X))c(X)\phi(f_\tau(X))\right]|f(X) - f^*_\phi(X)|^2\right] \\
				& \geq  \frac{1}{2}\E\left[b(X)\left[\eta(X)(1-c(X))e^{-f_\phi^*(X)} + (1-\eta(X))c(X)e^{f_\phi^*(X)}\right]|f(X) - f^*_\phi(X)|^2\right] \\
				& = \E\left[b(X)\sqrt{\eta(X)(1-c(X))c(X)(1-\eta(X))}|f(X) - f^*_\phi(X)|^2\right] \\
				& = \frac{1}{2}\E\left[\sqrt{\eta(X)(a(X)+b(X))(b(X)-a(X))(1-\eta(X))}|f(X) - f^*_\phi(X)|^2\right] \\
				&\geq c_b\epsilon\|f - f^*_\phi\|^2,
			\end{aligned}
	\end{equation*}}
	where the third line follows since $f_\phi^*=\frac1 2\log\left(\frac{\eta(1-c)}{c(1-\eta)}\right)$, see Lemma~\ref{lemma:exponential}; and the last since $a+b\geq 2c_b$ and $b-a\geq 2c_b$ under Assumption~\ref{as:losses} (i), and $\eta\geq \epsilon$ and  $1-\eta \geq \epsilon$ under Assumption~\ref{as:losses} (ii).
	
	Similarly, for the logistic convexifying function,
	\begin{equation*}
		\begin{aligned}
			\mathcal{R}_\phi(f) - \mathcal{R}_\phi^* & \geq 2c_b\epsilon \E\left[\phi''(f_\tau(X))|f(X) - f^*_\phi(X)|^2\right]\\
			& \geq 2c_b\epsilon \E\left[\phi''(F \vee f_\phi^*\vee -f^*_\phi)|f(X) - f^*_\phi(X)|^2 \right] \\
			& \geq 2c_b\epsilon\phi''\left(F\vee \log\left(\frac{(1-\epsilon)c_b}{2\epsilon M}\right)\right) \|f - f^*_\phi\|^2 \\
		\end{aligned}
	\end{equation*}
	where the first inequality uses $\phi''(z)=\frac{e^z}{(1+e^z)^2\log 2}$, so that $\phi''(-z)=\phi''(z)$; the second since $-F\wedge f^*_\phi \leq f_\tau\leq F\vee f_\phi^*$ and $\phi''$ is decreasing to zero; and the last since by Lemma~\ref{lemma:logistic}
	\begin{equation*}
		\log\left(\frac{2M\epsilon }{c_b(1-\epsilon)}\right) \leq f_\phi^* = \log\left(\frac{\eta(a + b)}{(1-\eta)(b - a)}\right) \leq \log\left(\frac{(1-\epsilon)c_b}{\epsilon 2M}\right),
	\end{equation*}
	which follows under Assumption~\ref{as:losses}.
\end{proof}
\begin{lemma}\label{lemma:equation_hinge}
	Suppose that  $|f|\leq 1$. Then the hinge convexifying function satisfies 
	\begin{equation*}
		\begin{aligned}
			\mathcal{R}_\phi(f) - \mathcal{R}_\phi^* & 
			& = \int_\mathcal{X}b\left|f - f_\phi^*\right||\eta-c|\dx P_X.
		\end{aligned}
	\end{equation*}	
\end{lemma}
\begin{proof}
	Since, $|f|\leq 1$ and $f_\phi^*(x)=\sign(\eta(x)-c(x))$, see Lemma~\ref{lemma:hinge}, we have 
	\begin{equation*}
		\begin{aligned}
			\mathcal{R}_\phi(f) - \mathcal{R}_\phi^* & = 0.5\E[(a(X)Y +b(X))Y(f_\phi^*(X) - f(X))] \\
			& = 0.5\E[(a(X)-b(X)+2\eta(X)b(X))(f^*_\phi(X) - f(X))] \\
			& = \E\left[b(X)(\eta(X)-c(X))(f^*_\phi(X) - f(X))\right] \\
			& = \int_\mathcal{X}b\left|f - f_\phi^*\right||\eta-c|\dx P_X.
		\end{aligned}
	\end{equation*}	
\end{proof}
\begin{lemma}\label{lemma:curvature_hinge}
	Suppose that Assumptions~\ref{as:losses} and \ref{as:tsybakov} are satisfied and that $|f|\leq 1$ for all $f\in\mathscr{F}_n$. Then the hinge convexifying function satisfies Assumption~\ref{as:loss_modulus} with $\kappa = 1+1/\alpha$ and $c_\phi=2^{-3/\alpha-1}c_bC_m^{-1/\alpha}$.
\end{lemma}
\begin{proof}
	For every $u>0$,
	\begin{equation*}
		\begin{aligned}
			\|f - f^*_\phi\|^2	& \leq 2\int_\mathcal{X}|f-f_\phi^*|\dx P_X \\
			& = 2\int_{|\eta-c|>u}\left|f - f_\phi^*\right|\dx P_X + 2\int_{|\eta-c|\leq u}\left|f - f_\phi^*\right|\dx P_X \\
			& \leq \frac{1}{c_bu}\int_\mathcal{X}b\left|f - f_\phi^*\right||\eta-c|\dx P_X +  4P_X(|\eta - c|\leq u) \\
			& \leq \frac{1}{c_bu}[\mathcal{R}_\phi(f) - \mathcal{R}_\phi^*] +4C_m u^\alpha,
		\end{aligned}
	\end{equation*}
	where the third line follows since $b\geq 2c_b$ under Assumption~\ref{as:losses}; and the last line by Lemma~\ref{lemma:equation_hinge} and Assumption~\ref{as:tsybakov}. To balance the two terms above, we shall take $u = \left([\mathcal{R}_\phi(f) - \mathcal{R}_\phi^*]/(4c_bC_m)\right)^{1/(1+\alpha)}$, in which case 
	\begin{equation*}
		\|f - f_\phi^*\|^2 \leq \frac{2}{c_b}(4c_bC_m)^{1/(1+\alpha)}[\mathcal{R}_\phi(f) - \mathcal{R}_\phi^*]^{\alpha/(1+\alpha)}.
	\end{equation*}
	yields the result with $c_\phi=2^{-3/\alpha-1}c_bC_m^{-1/\alpha}$ and $\kappa=1+\frac 1\alpha$.
\end{proof}

\begin{proof}[Proof of Theorem~\ref{thm:oracle inequality}]	
	By Theorem~\ref{thm:convexified_risk_bound_margin}
	\begin{equation}\label{eq:convexified_bound}
		\begin{aligned}
			\mathcal{R}(\sign(\hat f_n)) - \mathcal{R}^* & \leq C_\phi \left[\mathcal{R}_\phi(\hat f_n) - \mathcal{R}_\phi^*\right]^\frac{\gamma(\alpha+1)}{\gamma\alpha+1}  \triangleq C_\phi\left[\mathcal{R}_\phi(\hat f_n) - \mathcal{R}_\phi(f_n^*) + \triangle_n \right]^\frac{\gamma(\alpha+1)}{\gamma\alpha+1}
		\end{aligned}
	\end{equation}
	with $\triangle_n \triangleq \mathcal{R}_\phi(f_n^*) - \mathcal{R}_\phi^*$. We will bound the stochastic term by \cite{koltchinskii2011oracle}, Theorem 4.3. To that end, put $\mathscr{F}(\delta) = \left\{\ell\circ f:\; \mathcal{R}_\phi(f) - \mathcal{R}_\phi(f_n^*)\leq \delta, \;f\in\mathscr{F}_n\right\}$ for some $\delta>0$ and $(\ell\circ f)(y,x) = \omega(y,x)\phi(-yf(x))$. Then for every $f\in\mathscr{F}_n$ such that $\ell\circ f\in\mathscr{F}(\delta)$
	\begin{equation}\label{eq:f-f_n^*}
		\begin{aligned}
			\|f - f_n^*\| & \leq \|f - f_\phi^*\| + \|f_n^* - f_\phi^*\| \leq c_\phi^{-\frac{1}{2\kappa}}\left[\mathcal{R}_\phi(f) - \mathcal{R}_\phi^*\right]^{\frac{1}{2\kappa}} + c_\phi^{-\frac{1}{2\kappa}}\triangle_n^{\frac{1}{2\kappa}} \\
			& \leq { 2^{1-\frac{1}{2\kappa}}}c_\phi^{-\frac{1}{2\kappa}}\left[\mathcal{R}_\phi(f) - \mathcal{R}_\phi^* + \triangle_n\right]^{\frac{1}{2\kappa}}  \leq 2^{1-\frac{1}{2\kappa}}[c_\phi^{-1}(\delta + 2\triangle_n)]^{\frac{1}{2\kappa}},
		\end{aligned}
	\end{equation}
	where the second inequality follows under Assumption~\ref{as:loss_modulus} and the third by Jensen's inequality since $x\mapsto x^{1/2\kappa}$ is concave on $\R_+$ for $\kappa\geq 1$. Therefore,
	\begin{equation}\label{eq:class_inclusion}
		\mathscr{F}(\delta) \subset \left\{\ell\circ f:\; \|f - {f_n^*}\| \leq 2[c_\phi^{-1}(\delta/2 + \triangle_n)]^\frac{1}{2\kappa},\; f\in\mathscr{F}_n \right\}.
	\end{equation}
	Under Assumptions~\ref{as:losses} (iii) and \ref{as:phi} (ii) for all $f_1,f_2\in\mathscr{F}_n$ and all $(y,x)\in\{-1,1\}\times\mathcal{X}$, we have $|(\ell\circ f_1)(y,x) - (\ell\circ f_2)(y,x)| \leq 4LM|f_1(x) - f_2(x)|$. In conjunction with inequalities in equations (\ref{eq:f-f_n^*}) and (\ref{eq:class_inclusion}) this shows that the $L_2$-diameter of $\mathscr{F}(\delta)$ satisfies $D(\delta) \triangleq \sup_{g_1,g_2\in\mathscr{F}(\delta)}\|g_1-g_2\| \leq 8LM[c_\phi^{-1}(\delta/2 + \triangle_n)]^{\frac{1}{2\kappa}}$, therefore,
	\begin{equation*}
		(D^2)^{\flat}(\sigma) \triangleq \sup_{\delta\geq\sigma}\frac{D^2(\delta)}{\delta} \leq (8LM)^2c_\phi^{-1/\kappa}\sup_{\delta \geq \sigma }\delta^{\frac{1}{\kappa}-1}[0.5+\triangle_n/\delta]^{\frac{1}{\kappa}} \leq (8LM)^2c_\phi^{-1/\kappa} \sigma^{\frac{1}{\kappa}-1}[0.5+\tau]^{\frac{1}{\kappa}},
	\end{equation*}
	where $\tau \triangleq \triangle_n/\sigma$. Likewise, it follows from the equation (\ref{eq:class_inclusion}) that
	\begin{equation*}
		\begin{aligned}
			\phi_n(\delta) & \triangleq \E\left[\sup_{g_1,g_2\in\mathscr{F}(\delta)}|(P_n - P)(g_1-g_2)|\right]  \leq 2\E\left[\sup_{g\in\mathscr{F}(\delta)}|(P_n - P)(g-\ell\circ f_n^*)|\right] \\
			& \leq 2\E\left[\sup_{f\in\mathscr{F}_n:\|f -{f_n^*}\|\leq 2[c_\phi^{-1}(\delta/2 + \triangle_n)]^\frac{1}{2\kappa}}|(P_n - P)(\ell\circ f-\ell\circ f_n^*)|\right] \\
			& \leq  4\E\left[\sup_{f\in\mathscr{F}_n:\|f - {f_n^*}\| \leq 2[c_\phi^{-1}(\delta/2 + \triangle_n)]^\frac{1}{2\kappa}}|R_n(\ell\circ f-\ell\circ f_n^*)|\right] \\
			& \leq 8\E\left[\sup_{f\in\mathscr{F}_n:\|f - {f_n^*}\| \leq 2[c_\phi^{-1}(\delta/2 + \triangle_n)]^\frac{1}{2\kappa}}|R_n(f-f_n^*)|\right] \\
			& = 8\psi_n\left({4[c_\phi^{-1}}(\delta/2 + \triangle_n)]^\frac{1}{\kappa};\mathscr{F}_n\right),
		\end{aligned}
	\end{equation*}
	where we use the symmetrization and contraction inequalities, see \cite{koltchinskii2011oracle}, Theorems 2.1 and 2.3 since under Assumption~\ref{as:phi} (i), $\phi(0)=1$. This gives
	\begin{align*}
		\phi_n^\flat(\sigma) &=\sup_{\delta\geq\sigma}\frac{\phi_n(\delta)}{\delta}  \leq \sup_{\delta\geq\sigma}\frac{8\psi_n\left(4\delta^\frac{1}{\kappa}\left[c_\phi^{-1}(0.5+\tau)\right]^\frac{1}{\kappa};\mathscr{F}_n\right)}{\delta} \\
		&\leq  32\sigma^{\frac{1}{\kappa}-1}\left[c_\phi^{-1}(0.5+\tau)\right]^\frac{1}{\kappa}\sup_{\delta\geq\sigma}\frac{\psi_n\left(4\delta^\frac{1}{\kappa}\left[c_\phi^{-1}(0.5+\tau)\right]^\frac{1}{\kappa};\mathscr{F}_n\right)}{4\delta^\frac{1}{\kappa}\left[c_\phi^{-1}(0.5+\tau)\right]^\frac{1}{\kappa}}\\
		& =  32\sigma^{\frac{1}{\kappa}-1}\left[c_\phi^{-1}(0.5+\tau)\right]^\frac{1}{\kappa}\psi_n^\flat\left(4\sigma^\frac{1}{\kappa}\left[c_\phi^{-1}(0.5+\tau)\right]^\frac{1}{\kappa}\right).
	\end{align*}
	Next, by \cite{koltchinskii2011oracle}, Theorem 4.3, there exists $q>1$ such that for every $t>0$ 
	\begin{equation*}
		\Pr\left(\mathcal{R}_\phi(\hat f_n) - \mathcal{R}_\phi(f_n^*) \leq \inf\left\{\sigma:\; V_n^t(\sigma) \leq 1 \right\}\right) \geq 1 - c_qe^{-t}
	\end{equation*}
	with $c_q = \frac{q}{q-1}\vee e$ and
	{\footnotesize
		\begin{equation*}
			\begin{aligned}
				V_n^t(\sigma) & \triangleq 2q\left[\phi_n^\flat(\sigma) + \sqrt{\frac{(D^2)^\flat(\sigma)t}{n\sigma}} + \frac{t}{n\sigma}\right] \\
				& \leq 64q\sigma^{\frac{1-\kappa}{\kappa}}\left[c_\phi^{-1}(0.5+\tau)\right]^\frac{1}{\kappa}\psi_n^\flat\left(4\sigma^\frac{1}{\kappa}\left[c_\phi^{-1}(0.5+\tau)\right]^\frac{1}{\kappa}\right) + 16LMqc_\phi^\frac{-1}{2\kappa} \sqrt{\frac{[0.5+\tau]^\frac{1}{\kappa}t}{n\sigma^{2-\frac{1}{\kappa}}}} + \frac{2qt}{n\sigma},
			\end{aligned}
	\end{equation*}}
	which follows from our computations above. Note that if $\sigma\geq \triangle_n$, then $\tau = \triangle_n/\sigma\leq 1$, and
	\begin{equation*}
		V_n^t(\sigma) \leq 16q\sigma^{\frac{1-\kappa}{\kappa}}\left(\frac {4^{\kappa}3 }{2c_\phi}\right)^\frac{1}{\kappa}\psi_n^\flat\left(\left[\frac{4^\kappa 3\sigma}{2c_\phi}\right]^{1/\kappa}\right) + 16qLM\left(\frac 3{2c_\phi}\right)^\frac{1}{2\kappa} \sqrt{\frac{t}{n\sigma^{2-\frac{1}{\kappa}}}} + \frac{2qt}{n\sigma}.
	\end{equation*}
	Since all functions in this upper bound are decreasing in $\sigma$, we have $V_n^t(\sigma)\leq 1$ as soon as
	\begin{equation*}
		\sigma\geq\frac{2c_\phi}{4^\kappa 3} \psi_{n,\kappa}^{\sharp}\left( \frac{c_\phi}{72q4^\kappa}\right)\vee \left(\frac{3(48qLM)^{2\kappa}t^\kappa}{2c_\phi n^\kappa}\right)^\frac{1}{2\kappa-1}\vee \frac{6qt}{n}.
	\end{equation*}
	Therefore, since $\inf\{\sigma \leq \triangle_n:\; V_n^t(\sigma)\leq 1\}\leq \triangle_n$, we obtain with $\epsilon= c_\phi/(72q4^\kappa)$
	\begin{equation*}
		\inf\{\sigma:\; V_n^t(\sigma)\leq 1 \}\leq \frac{2c_\phi}{4^\kappa3}\psi_{n,\kappa}^{\sharp}(\epsilon)\vee \left(\frac{3(48qLM)^{2\kappa}t^\kappa}{2c_\phi n^\kappa}\right)^\frac{1}{2\kappa-1}\vee \frac{6qt}{n} + \triangle_n.
	\end{equation*}
	This shows that in conjunction with the inequality in equation~(\ref{eq:convexified_bound}), for every $t>0$ with probability at least $1-c_qe^{-t}$
	\begin{equation*}
		\mathcal{R}(\sign(\hat f_n)) - \mathcal{R}^* \leq C_\phi\left[\frac{2c_\phi}{4^\kappa3}\psi_{n,\kappa}^{\sharp}(\epsilon)\vee \left(\frac{3(48qLM)^{2\kappa}t^\kappa}{2c_\phi n^\kappa}\right)^\frac{1}{2\kappa-1}\vee \frac{6qt}{n} + 2\triangle_n \right]^\frac{\gamma(\alpha+1)}{\gamma\alpha+1}.
	\end{equation*}
\end{proof}

\begin{proof}[Proof of Theorem~\ref{thm:parametric_predictions}]
	By \cite{koltchinskii2011oracle}, Proposition 3.2, for the linear class $\mathscr{F}_n=\{f_\theta(x)=\sum_{j=1}^p\theta_j\varphi_j(x): \theta\in\R^p\}$, the local Rademacher complexity is bounded as $\psi_n(\delta;\mathscr{F}_n) \leq \sqrt{\delta p/n}$. This gives $\psi_n^\flat(\sigma) \leq \sqrt{p/(\sigma n)}$, and so $\psi^{\sharp}_{n,\kappa}(\epsilon) \leq \left(p/(n\epsilon^2)\right)^\frac{\kappa}{2\kappa - 1}$. Therefore, by Theorem~\ref{thm:oracle inequality} for every $t>0$,
	\begin{equation*}
		\Pr\left(\left[\mathcal{R}(\sign(\hat f_n)) - \mathcal{R}^*\right]^\frac{\gamma\alpha+1}{\gamma(\alpha+1)} > K\left[\left(\frac{p}{n}\right)^\frac{\kappa}{2\kappa - 1} + tn^{-\frac{\kappa}{2\kappa - 1}} + \triangle_n \right]\right)\leq c_qe^{-t}.
	\end{equation*}	
	where we use $t^{\kappa/(2\kappa - 1)} \leq p\vee t$ since $\kappa\geq 1$, $\epsilon,c_q,C_\phi>0,q>1\triangle_n$ as defined in the proof of Theorem~\ref{thm:oracle inequality}, and put $K\triangleq C_\phi^\frac{\gamma\alpha+1}{\gamma(\alpha+1)}\left(\frac{2c_\phi}{4^\kappa3}\epsilon^{-\frac{2\kappa}{2\kappa-1}}\vee\left[\frac{3(48qLM)^{2\kappa}}{2c_\phi}\right]^{1/(2\kappa-1)}\vee 6q\vee 2\right)$. Integrating the tail bound
	\begin{equation*}
		\begin{aligned}
			\E\left[\left(\mathcal{R}(\sign(\hat f_n)) - \mathcal{R}^*\right)^\frac{\gamma\alpha+1}{\gamma(\alpha+1)}\right] & = \int_0^\infty\Pr\left(\left[\mathcal{R}(\sign(\hat f_n)) - \mathcal{R}^*\right]^\frac{\gamma\alpha+1}{\gamma(\alpha+1)} > u\right)\dx u \\
			& \leq Kn^{-\frac{\kappa}{2\kappa - 1}}\left\{n^\frac{\kappa}{2\kappa - 1}\left[\left(\frac{p}{n}\right)^\frac{\kappa}{2\kappa - 1} + \triangle_n\right] + \int_0^\infty c_qe^{-t}\dx t\right\} \\
			& \leq K\left[\left(\frac{p}{n}\right)^\frac{\kappa}{2\kappa - 1} + \triangle_n\right] + Kc_qn^{-\frac{\kappa}{2\kappa - 1}},
		\end{aligned}
	\end{equation*}
	where the second line follows by the change of variables $u = K\left[\left(\frac{p}{n}\right)^\frac{\kappa}{2\kappa - 1} + tn^{-\frac{\kappa}{2\kappa - 1}} + \triangle_n \right]$ and bounding the probability by $1$ for $t<0$. Since $\gamma\in(0,1]$, by Jensen's inequality, this gives
	\begin{equation*}
		\begin{aligned}
			\E\left[\mathcal{R}(\sign(\hat f_n)) - \mathcal{R}^*\right] & \leq  \left\{K\left[\left(\frac{p}{n}\right)^\frac{\kappa}{2\kappa - 1} + \triangle_n\right] + Kc_qn^{-\frac{\kappa}{2\kappa - 1}} \right\}^\frac{\gamma(\alpha+1)}{\gamma\alpha+1}.
		\end{aligned}
	\end{equation*}
\end{proof}

\bibliographystyle{abbrvnat}
\bibliography{references}

\begin{thebibliography}{90}
\providecommand{\natexlab}[1]{#1}
\providecommand{\url}[1]{\texttt{#1}}
\expandafter\ifx\csname urlstyle\endcsname\relax
  \providecommand{\doi}[1]{doi: #1}\else
  \providecommand{\doi}{doi: \begingroup \urlstyle{rm}\Url}\fi

\bibitem[Adjaho and Christensen(2023)]{adjaho2022externally}
C.~Adjaho and T.~Christensen.
\newblock Externally valid policy choice.
\newblock \emph{arXiv preprint arXiv:2205.05561}, 2023.

\bibitem[Agarwal et~al.(2018)Agarwal, Beygelzimer, Dudik, Langford, and Wallach]{agarwal2018reductions}
A.~Agarwal, A.~Beygelzimer, M.~Dudik, J.~Langford, and H.~Wallach.
\newblock A reductions approach to fair classification.
\newblock In J.~Dy and A.~Krause, editors, \emph{Proceedings of the 35th International Conference on Machine Learning}, volume~80 of \emph{Proceedings of Machine Learning Research}, pages 60--69. PMLR, 10--15 Jul 2018.

\bibitem[Anthony and Bartlett(2009)]{anthony2009neural}
M.~Anthony and P.~L. Bartlett.
\newblock \emph{Neural network learning: Theoretical foundations}.
\newblock Cambridge University Press, 2009.

\bibitem[Athey and Wager(2021)]{athey2021policy}
S.~Athey and S.~Wager.
\newblock Policy learning with observational data.
\newblock \emph{Econometrica}, 89\penalty0 (1):\penalty0 133--161, 2021.

\bibitem[Audibert and Tsybakov(2007)]{audibert2007fast}
J.-Y. Audibert and A.~B. Tsybakov.
\newblock Fast learning rates for plug-in classifiers.
\newblock \emph{Annals of Statistics}, 35\penalty0 (2):\penalty0 608--633, 2007.

\bibitem[Bahnsen et~al.(2014)Bahnsen, Aouada, and Ottersten]{bahnsen2014example}
A.~C. Bahnsen, D.~Aouada, and B.~Ottersten.
\newblock Example-dependent cost-sensitive logistic regression for credit scoring.
\newblock In \emph{2014 13th International conference on machine learning and applications}, pages 263--269. IEEE, 2014.

\bibitem[Bahnsen et~al.(2015)Bahnsen, Aouada, and Ottersten]{bahnsen2015example}
A.~C. Bahnsen, D.~Aouada, and B.~Ottersten.
\newblock Example-dependent cost-sensitive decision trees.
\newblock \emph{Expert Systems with Applications}, 42\penalty0 (19):\penalty0 6609--6619, 2015.

\bibitem[Bao et~al.(2020)Bao, Scott, and Sugiyama]{bao2020calibrated}
H.~Bao, C.~Scott, and M.~Sugiyama.
\newblock Calibrated surrogate losses for adversarially robust classification.
\newblock In \emph{Conference on Learning Theory}, pages 408--451. PMLR, 2020.

\bibitem[Bartlett et~al.(2006)Bartlett, Jordan, and McAuliffe]{bartlett2006convexity}
P.~L. Bartlett, M.~I. Jordan, and J.~D. McAuliffe.
\newblock Convexity, classification, and risk bounds.
\newblock \emph{Journal of the American Statistical Association}, 101\penalty0 (473):\penalty0 138--156, 2006.

\bibitem[Bartlett et~al.(2019)Bartlett, Harvey, Liaw, and Mehrabian]{bartlett2019nearly}
P.~L. Bartlett, N.~Harvey, C.~Liaw, and A.~Mehrabian.
\newblock Nearly-tight vc-dimension and pseudodimension bounds for piecewise linear neural networks.
\newblock \emph{Journal of Machine Learning Research}, 20\penalty0 (63):\penalty0 1--17, 2019.

\bibitem[Baughman(2017)]{baughman2017costs}
S.~B. Baughman.
\newblock Costs of pretrial detention.
\newblock \emph{Boston University Law Review}, 97\penalty0 (1), 2017.

\bibitem[Belloni et~al.(2015)Belloni, Chernozhukov, Chetverikov, and Kato]{Belloni2015some}
A.~Belloni, V.~Chernozhukov, D.~Chetverikov, and K.~Kato.
\newblock Some new asymptotic theory for least squares series: Pointwise and uniform results.
\newblock \emph{Journal of Econometrics}, 186\penalty0 (2):\penalty0 345--366, 2015.

\bibitem[Belloni et~al.(2018)Belloni, Chernozhukov, Chetverikov, Hansen, and Kato]{Belloni2018high}
A.~Belloni, V.~Chernozhukov, D.~Chetverikov, C.~Hansen, and K.~Kato.
\newblock High-dimensional econometrics and regularized gmm.
\newblock \emph{arXiv preprint arXiv:1806.01888}, 2018.

\bibitem[Bent(2019)]{bent2019algorithmic}
J.~R. Bent.
\newblock Is algorithmic affirmative action legal.
\newblock \emph{Geo. LJ}, 108:\penalty0 803, 2019.

\bibitem[Bhattacharya and Dupas(2012)]{bhattacharya2012inferring}
D.~Bhattacharya and P.~Dupas.
\newblock Inferring welfare maximizing treatment assignment under budget constraints.
\newblock \emph{Journal of Econometrics}, 167\penalty0 (1):\penalty0 168--196, 2012.
\newblock ISSN 0304-4076.

\bibitem[Bickel et~al.(2009)Bickel, Ritov, and Tsybakov]{Bickel2009simultaneous}
P.~J. Bickel, Y.~Ritov, and A.~B. Tsybakov.
\newblock Simultaneous analysis of lasso and dantzig selector.
\newblock \emph{The Annals of statistics}, 37\penalty0 (4):\penalty0 1705--1732, 2009.

\bibitem[Boucheron et~al.(2005)Boucheron, Bousquet, and Lugosi]{boucheron2005theory}
S.~Boucheron, O.~Bousquet, and G.~Lugosi.
\newblock Theory of classification: A survey of some recent advances.
\newblock \emph{ESAIM: Probability and Statistics}, 9:\penalty0 323--375, 2005.

\bibitem[Breiman(2000)]{breiman2000some}
L.~Breiman.
\newblock Some infinity theory for predictor ensembles.
\newblock Technical report, 2000.

\bibitem[B{\"u}hlmann and van~de Geer(2011)]{buhlmann2011statistics}
P.~B{\"u}hlmann and S.~van~de Geer.
\newblock \emph{Statistics for high-dimensional data}.
\newblock Springer Series in Statistics. Springer, Heidelberg, 2011.
\newblock ISBN 978-3-642-20191-2.
\newblock Methods, theory and applications.

\bibitem[Chen and Guestrin(2016)]{chen2016xgboost}
T.~Chen and C.~Guestrin.
\newblock Xgboost: A scalable tree boosting system.
\newblock In \emph{Proceedings of the 22nd acm sigkdd international conference on knowledge discovery and data mining}, pages 785--794, 2016.

\bibitem[Chen(2007)]{chen2007sieves}
X.~Chen.
\newblock Large sample sieve estimation of semi-nonparametric models.
\newblock In J.~J. Heckman and E.~E. Leamer, editors, \emph{{Handbook of Econometrics - Volume 6b}}, pages 5549--5632. Elsevier, 2007.

\bibitem[Chen and Ludvigson(2009)]{chen2009land}
X.~Chen and S.~C. Ludvigson.
\newblock Land of addicts? an empirical investigation of habit-based asset pricing models.
\newblock \emph{Journal of Applied Econometrics}, 24\penalty0 (7):\penalty0 1057--1093, 2009.

\bibitem[Chen et~al.(2001)Chen, Racine, and Swanson]{chen2001semiparametric}
X.~Chen, J.~Racine, and N.~R. Swanson.
\newblock Semiparametric arx neural-network models with an application to forecasting inflation.
\newblock \emph{IEEE Transactions on Neural Networks}, 12\penalty0 (4):\penalty0 674--683, 2001.

\bibitem[Chetverikov and S{\o}rensen(2021)]{Chetverikov2021analytic}
D.~Chetverikov and J.~R.-V. S{\o}rensen.
\newblock Analytic and bootstrap-after-cross-validation methods for selecting penalty parameters of high-dimensional m-estimators.
\newblock \emph{arXiv preprint arXiv:2104.04716}, 2021.

\bibitem[Chouldechova(2017)]{chouldechova2017fair}
A.~Chouldechova.
\newblock Fair prediction with disparate impact: A study of bias in recidivism prediction instruments.
\newblock \emph{Big data}, 5\penalty0 (2):\penalty0 153--163, 2017.

\bibitem[Christoffersen and Diebold(1996)]{diebold1996optimal}
P.~F. Christoffersen and F.~X. Diebold.
\newblock Further results on forecasting and model selection under asymmetric loss.
\newblock \emph{Journal of Applied Econometrics}, 11\penalty0 (5):\penalty0 561--571, 1996.

\bibitem[Christoffersen and Diebold(1997)]{diebold1997optimal}
P.~F. Christoffersen and F.~X. Diebold.
\newblock Optimal prediction under asymmetric loss.
\newblock \emph{Econometric Theory}, 13\penalty0 (6):\penalty0 808–817, 1997.

\bibitem[Corbett-Davies et~al.(2017)Corbett-Davies, Pierson, Feller, Goel, and Huq]{corbett2017algorithmic}
S.~Corbett-Davies, E.~Pierson, A.~Feller, S.~Goel, and A.~Huq.
\newblock Algorithmic decision making and the cost of fairness.
\newblock In \emph{Proceedings of the 23rd ACM SIGKDD International Conference on Knowledge Discovery and Data Mining}, pages 797--806, 2017.

\bibitem[Cotter et~al.(2019)Cotter, Jiang, and Sridharan]{cotter2019two}
A.~Cotter, H.~Jiang, and K.~Sridharan.
\newblock Two-player games for efficient non-convex constrained optimization.
\newblock In A.~Garivier and S.~Kale, editors, \emph{Proceedings of the 30th International Conference on Algorithmic Learning Theory}, volume~98 of \emph{Proceedings of Machine Learning Research}, pages 300--332. PMLR, 22--24 Mar 2019.

\bibitem[Cowgill and Tucker(2019)]{cowgill2019economics}
B.~Cowgill and C.~E. Tucker.
\newblock Economics, fairness and algorithmic bias.
\newblock {\it Journal of Economic Perspectives}, forthcoming, 2019.

\bibitem[Datta et~al.(2015)Datta, Tschantz, and Datta]{datta2015automated}
A.~Datta, M.~C. Tschantz, and A.~Datta.
\newblock Automated experiments on ad privacy settings: A tale of opacity, choice, and discrimination.
\newblock \emph{Proceedings on Privacy Enhancing Technologies}, 2015\penalty0 (1):\penalty0 92--112, 2015.

\bibitem[Dell(2024)]{dell2024deep}
M.~Dell.
\newblock Deep learning for economists.
\newblock Technical report, National Bureau of Economic Research, 2024.

\bibitem[Devroye et~al.(1996)Devroye, Gy{\"o}rfi, and Lugosi]{devroye1996probabilistic}
L.~Devroye, L.~Gy{\"o}rfi, and G.~Lugosi.
\newblock \emph{A probabilistic theory of pattern recognition}, volume~31.
\newblock Springer, 1996.

\bibitem[Dieterich et~al.(2016)Dieterich, Mendoza, and Brennan]{dieterich2016compas}
W.~Dieterich, C.~Mendoza, and T.~Brennan.
\newblock Compas risk scales: Demonstrating accuracy equity and predictive parity.
\newblock \emph{Northpointe Inc}, 7\penalty0 (4):\penalty0 1--36, 2016.

\bibitem[Elkan(2001)]{elkan2001foundations}
C.~Elkan.
\newblock The foundations of cost-sensitive learning.
\newblock In \emph{International joint conference on artificial intelligence}, volume~17, pages 973--978. Lawrence Erlbaum Associates Ltd, 2001.

\bibitem[Elliott and Lieli(2013)]{elliott2013predicting}
G.~Elliott and R.~P. Lieli.
\newblock Predicting binary outcomes.
\newblock \emph{Journal of Econometrics}, 174\penalty0 (1):\penalty0 15--26, 2013.

\bibitem[Elliott and Timmermann(2016)]{elliott2016economic}
G.~Elliott and A.~Timmermann.
\newblock \emph{Economic Forecasting}.
\newblock Princeton University Press, 2016.

\bibitem[Farrell et~al.(2021)Farrell, Liang, and Misra]{farrell2018deep}
M.~H. Farrell, T.~Liang, and S.~Misra.
\newblock Deep neural networks for estimation and inference: Application to causal effects and other semiparametric estimands.
\newblock \emph{Econometrica}, 89\penalty0 (1):\penalty0 181--213, 2021.

\bibitem[Florios and Skouras(2008)]{florios2008exact}
K.~Florios and S.~Skouras.
\newblock Exact computation of max weighted score estimators.
\newblock \emph{Journal of Econometrics}, 146\penalty0 (1):\penalty0 86--91, 2008.
\newblock ISSN 0304-4076.

\bibitem[Friedman et~al.(2001)Friedman, Hastie, Tibshirani, et~al.]{friedman2001elements}
J.~Friedman, T.~Hastie, R.~Tibshirani, et~al.
\newblock \emph{The elements of statistical learning}, volume~1.
\newblock Springer series in statistics New York, 2001.

\bibitem[Gallant and White(1988)]{gallant1988there}
A.~R. Gallant and H.~White.
\newblock There exists a neural network that does not make avoidable mistakes.
\newblock In \emph{IEEE 1988 International Conference on Neural Networks}, pages 657--664, 1988.

\bibitem[Gallant and White(1992)]{gallant1992learning}
A.~R. Gallant and H.~White.
\newblock On learning the derivatives of an unknown mapping with multilayer feedforward networks.
\newblock \emph{Neural Networks}, 5\penalty0 (1):\penalty0 129--138, 1992.

\bibitem[Granger(1969)]{granger1969prediction}
C.~W. Granger.
\newblock Prediction with a generalized cost of error function.
\newblock \emph{Journal of the Operational Research Society}, 20\penalty0 (2):\penalty0 199--207, 1969.

\bibitem[Granger(1995)]{granger1995modelling}
C.~W. Granger.
\newblock Modelling nonlinear relationships between extended-memory variables.
\newblock \emph{Econometrica}, 63:\penalty0 265--279, 1995.

\bibitem[Granger and Pesaran(2000)]{granger2000economic}
C.~W. Granger and M.~H. Pesaran.
\newblock Economic and statistical measures of forecast accuracy.
\newblock \emph{Journal of Forecasting}, 19\penalty0 (7):\penalty0 537--560, 2000.

\bibitem[Hirano and Porter(2009)]{hirano2009asymptotics}
K.~Hirano and J.~R. Porter.
\newblock Asymptotics for statistical treatment rules.
\newblock \emph{Econometrica}, 77\penalty0 (5):\penalty0 1683--1701, 2009.

\bibitem[Horowitz(1992)]{horowitz1992smoothed}
J.~L. Horowitz.
\newblock A smoothed maximum score estimator for the binary response model.
\newblock \emph{Econometrica}, 60\penalty0 (3):\penalty0 505--531, 1992.

\bibitem[Hutchinson et~al.(1994)Hutchinson, Lo, and Poggio]{hutchinson1994nonparametric}
J.~M. Hutchinson, A.~W. Lo, and T.~Poggio.
\newblock A nonparametric approach to pricing and hedging derivative securities via learning networks.
\newblock \emph{Journal of Finance}, 49\penalty0 (3):\penalty0 851--889, 1994.

\bibitem[Kallus(2021)]{kallus2021more}
N.~Kallus.
\newblock More efficient policy learning via optimal retargeting.
\newblock \emph{Journal of the American Statistical Association}, 116\penalty0 (534):\penalty0 646--658, 2021.

\bibitem[Kim et~al.(2021)Kim, Ohn, and Kim]{kim2021fast}
Y.~Kim, I.~Ohn, and D.~Kim.
\newblock Fast convergence rates of deep neural networks for classification.
\newblock \emph{Neural Networks}, 138:\penalty0 179--197, 2021.

\bibitem[Kitagawa and Tetenov(2018)]{kitagawa2018should}
T.~Kitagawa and A.~Tetenov.
\newblock {Who should be treated? Empirical welfare maximization methods for treatment choice}.
\newblock \emph{Econometrica}, 86\penalty0 (2):\penalty0 591--616, 2018.

\bibitem[Kitagawa et~al.(2023)Kitagawa, Sakaguchi, and Tetenov]{kitagawa2023constrained}
T.~Kitagawa, S.~Sakaguchi, and A.~Tetenov.
\newblock Constrained classification and policy learning.
\newblock \emph{arXiv preprint arXiv:2106.12886}, 2023.

\bibitem[Kleinberg et~al.(2018{\natexlab{a}})Kleinberg, Lakkaraju, Leskovec, Ludwig, and Mullainathan]{kleinberg2018human}
J.~Kleinberg, H.~Lakkaraju, J.~Leskovec, J.~Ludwig, and S.~Mullainathan.
\newblock Human decisions and machine predictions.
\newblock \emph{Quarterly Journal of Economics}, 133\penalty0 (1):\penalty0 237--293, 2018{\natexlab{a}}.

\bibitem[Kleinberg et~al.(2018{\natexlab{b}})Kleinberg, Ludwig, Mullainathan, and Rambachan]{kleinberg2018algorithmic}
J.~Kleinberg, J.~Ludwig, S.~Mullainathan, and A.~Rambachan.
\newblock Algorithmic fairness.
\newblock In \emph{AEA Papers and Proceedings}, volume 108, pages 22--27, 2018{\natexlab{b}}.

\bibitem[Koenker and Bassett(1978)]{koenker1978regression}
R.~Koenker and G.~Bassett.
\newblock Regression quantiles.
\newblock \emph{Econometrica}, pages 33--50, 1978.

\bibitem[Koltchinskii(2011)]{koltchinskii2011oracle}
V.~Koltchinskii.
\newblock \emph{Oracle Inequalities in Empirical Risk Minimization and Sparse Recovery Problems: Ecole d'Et{\'e} de Probabilit{\'e}s de Saint-Flour XXXVIII-2008}, volume 2033.
\newblock Springer, 2011.

\bibitem[Larson et~al.(2016)Larson, Mattu, Kirchner, and Angwin]{larson2016surya}
J.~Larson, S.~Mattu, L.~Kirchner, and J.~Angwin.
\newblock How we analyzed the compas recidivism algorithm.
\newblock \emph{ProPublica (5 2016)}, 2016.

\bibitem[Lee et~al.(1993)Lee, White, and Granger]{lee1993testing}
T.-H. Lee, H.~White, and C.~W. Granger.
\newblock Testing for neglected nonlinearity in time series models: A comparison of neural network methods and alternative tests.
\newblock \emph{Journal of Econometrics}, 56\penalty0 (3):\penalty0 269--290, 1993.

\bibitem[Liang et~al.(2021)Liang, Lu, Mu, and Okumura]{liang2021algorithm}
A.~Liang, J.~Lu, X.~Mu, and K.~Okumura.
\newblock Algorithm design: A fairness-accuracy frontier.
\newblock \emph{arXiv preprint arXiv:2112.09975}, 2021.

\bibitem[Lieli and White(2010)]{lieli2010construction}
R.~P. Lieli and H.~White.
\newblock The construction of empirical credit scoring rules based on maximization principles.
\newblock \emph{Journal of Econometrics}, 157\penalty0 (1):\penalty0 110--119, 2010.

\bibitem[Liu and Molinari(2024)]{liu2024inference}
Y.~Liu and F.~Molinari.
\newblock Inference for an algorithmic fairness-accuracy frontier.
\newblock \emph{arXiv preprint arXiv:2402.08879}, 2024.

\bibitem[Lu et~al.(2021)Lu, Shen, Yang, and Zhang]{lu2020deep}
J.~Lu, Z.~Shen, H.~Yang, and S.~Zhang.
\newblock Deep network approximation for smooth functions.
\newblock \emph{arXiv preprint arXiv:2001.03040}, 2021.

\bibitem[Manski(2004)]{manski2004statistical}
C.~Manski.
\newblock Statistical treatment rules for heterogeneous populations.
\newblock \emph{Econometrica}, 72\penalty0 (4):\penalty0 1221--1246, 2004.

\bibitem[Manski(1975)]{manski1975maximum}
C.~F. Manski.
\newblock Maximum score estimation of the stochastic utility model of choice.
\newblock \emph{Journal of econometrics}, 3\penalty0 (3):\penalty0 205--228, 1975.

\bibitem[Manski and Thompson(1989)]{manski1989estimation}
C.~F. Manski and T.~S. Thompson.
\newblock Estimation of best predictors of binary response.
\newblock \emph{Journal of Econometrics}, 40\penalty0 (1):\penalty0 97--123, 1989.

\bibitem[Mbakop and Tabord-Meehan(2021)]{mbakop2021model}
E.~Mbakop and M.~Tabord-Meehan.
\newblock Model selection for treatment choice: Penalized welfare maximization.
\newblock \emph{Econometrica}, 89\penalty0 (2):\penalty0 825--848, 2021.

\bibitem[McCulloch and Pitts(1943)]{mcculloch1943logical}
W.~S. McCulloch and W.~Pitts.
\newblock A logical calculus of the ideas immanent in nervous activity.
\newblock \emph{The Bulletin of Mathematical Biophysics}, 5\penalty0 (4):\penalty0 115--133, 1943.

\bibitem[Mehta et~al.(2020)Mehta, Babu, Rao, and Kumar]{Mehta2020deepcatch}
P.~Mehta, C.~S. Babu, S.~K.~V. Rao, and S.~Kumar.
\newblock Deepcatch: Predicting return defaulters in taxation system using example-dependent cost-sensitive deep neural networks.
\newblock In \emph{2020 IEEE International Conference on Big Data (Big Data)}, pages 4412--4419. IEEE, 2020.

\bibitem[Menon and Williamson(2018)]{menon2018cost}
A.~K. Menon and R.~C. Williamson.
\newblock The cost of fairness in binary classification.
\newblock In S.~A. Friedler and C.~Wilson, editors, \emph{Proceedings of the 1st Conference on Fairness, Accountability and Transparency}, volume~81 of \emph{Proceedings of Machine Learning Research}, pages 107--118. PMLR, 23--24 Feb 2018.

\bibitem[Mhaskar(1996)]{mhaskar1996neural}
H.~N. Mhaskar.
\newblock Neural networks for optimal approximation of smooth and analytic functions.
\newblock \emph{Neural computation}, 8\penalty0 (1):\penalty0 164--177, 1996.

\bibitem[Mitchell et~al.(2021)Mitchell, Potash, Barocas, D'Amour, and Lum]{mitchell2021algorithmic}
S.~Mitchell, E.~Potash, S.~Barocas, A.~D'Amour, and K.~Lum.
\newblock Algorithmic fairness: Choices, assumptions, and definitions.
\newblock \emph{Annual review of statistics and its application}, 8\penalty0 (1):\penalty0 141--163, 2021.

\bibitem[Newey and Powell(1987)]{newey1987asymmetric}
W.~K. Newey and J.~L. Powell.
\newblock Asymmetric least squares estimation and testing.
\newblock \emph{Econometrica}, pages 819--847, 1987.

\bibitem[Ponomarev and Semenova(2024)]{ponomarev2024lower}
K.~Ponomarev and V.~Semenova.
\newblock On the lower confidence band for the optimal welfare.
\newblock \emph{arXiv preprint arXiv:2410.07443}, 2024.

\bibitem[Posner and Saran(2025)]{posner2025judge}
E.~A. Posner and S.~Saran.
\newblock Judge ai: Assessing large language models in judicial decision-making.
\newblock \emph{University of Chicago Coase-Sandor Institute for Law \& Economics Research Paper}, \penalty0 (2503), 2025.

\bibitem[Rambachan et~al.(2020)Rambachan, Kleinberg, Ludwig, and Mullainathan]{rambachan2020economic}
A.~Rambachan, J.~Kleinberg, J.~Ludwig, and S.~Mullainathan.
\newblock An economic approach to regulating algorithms.
\newblock Technical report, National Bureau of Economic Research, 2020.

\bibitem[Rosenblatt(1958)]{rosenblatt1958perceptron}
F.~Rosenblatt.
\newblock The perceptron: a probabilistic model for information storage and organization in the brain.
\newblock \emph{Psychological review}, 65\penalty0 (6):\penalty0 386, 1958.

\bibitem[Scott(2011)]{scott2011calibrated}
C.~Scott.
\newblock Surrogate losses and regret bounds for cost-sensitive classification with example-dependent costs.
\newblock \emph{Proceedings of the 28th International Conference on International Conference on Machine Learning}, pages 153--160, 2011.

\bibitem[Scott(2012)]{scott2012calibrated}
C.~Scott.
\newblock Calibrated asymmetric surrogate losses.
\newblock \emph{Electronic Journal of Statistics}, 6:\penalty0 958--992, 2012.

\bibitem[Su(2021{\natexlab{a}})]{su2021model}
J.-H. Su.
\newblock Model selection in utility-maximizing binary prediction.
\newblock \emph{Journal of Econometrics}, 223\penalty0 (1):\penalty0 96--124, 2021{\natexlab{a}}.

\bibitem[Su(2021{\natexlab{b}})]{su2021utility}
J.-H. Su.
\newblock Utility-maximizing binary prediction via the nearest neighbor method and its application to credit scoring.
\newblock \emph{UCSD Working Paper}, 2021{\natexlab{b}}.

\bibitem[Sun et~al.(2007)Sun, Kamel, Wong, and Wang]{sun2007cost}
Y.~Sun, M.~S. Kamel, A.~K. Wong, and Y.~Wang.
\newblock Cost-sensitive boosting for classification of imbalanced data.
\newblock \emph{Pattern Recognition}, 40\penalty0 (12):\penalty0 3358--3378, 2007.

\bibitem[Targ et~al.(2016)Targ, Almeida, and Lyman]{targ2016resnet}
S.~Targ, D.~Almeida, and K.~Lyman.
\newblock Resnet in resnet: Generalizing residual architectures.
\newblock \emph{arXiv preprint arXiv:1603.08029}, 2016.

\bibitem[Ting(1998)]{ting1998inducing}
K.~M. Ting.
\newblock Inducing cost-sensitive trees via instance weighting.
\newblock In \emph{Principles of Data Mining and Knowledge Discovery}, pages 139--147. Springer Berlin Heidelberg, 1998.

\bibitem[van~de Geer(2008)]{geer2008}
S.~van~de Geer.
\newblock High-dimensional generalized linear models and the lasso.
\newblock \emph{Annals of Statistics}, 36\penalty0 (2):\penalty0 614--645, 2008.

\bibitem[Viviano and Bradic(2024)]{viviano2023fair}
D.~Viviano and J.~Bradic.
\newblock Fair policy targeting.
\newblock \emph{Journal of the American Statistical Association}, 119\penalty0 (545):\penalty0 730--743, 2024.

\bibitem[Wegkamp(2007)]{wegkamp2007lasso}
M.~Wegkamp.
\newblock Lasso type classifiers with a reject option.
\newblock \emph{Electronic Journal of Statistics}, 1:\penalty0 155--168, 2007.

\bibitem[White and Racine(2001)]{white2001statistical}
H.~White and J.~Racine.
\newblock Statistical inference, the bootstrap, and neural-network modeling with application to foreign exchange rates.
\newblock \emph{IEEE Transactions on Neural Networks}, 12\penalty0 (4):\penalty0 657--673, 2001.

\bibitem[Xia et~al.(2017)Xia, Liu, and Liu]{xia2017cost}
Y.~Xia, C.~Liu, and N.~Liu.
\newblock Cost-sensitive boosted tree for loan evaluation in peer-to-peer lending.
\newblock \emph{Electronic Commerce Research and Applications}, 24:\penalty0 30--49, 2017.

\bibitem[Zhang(2004)]{zhang2004statistical}
T.~Zhang.
\newblock Statistical behavior and consistency of classification methods based on convex risk minimization.
\newblock \emph{Annals of Statistics}, pages 56--85, 2004.

\bibitem[Zhao et~al.(2012)Zhao, Zeng, Rush, and Kosorok]{zhao2012estimating}
Y.~Zhao, D.~Zeng, A.~J. Rush, and M.~R. Kosorok.
\newblock Estimating individualized treatment rules using outcome weighted learning.
\newblock \emph{Journal of the American Statistical Association}, 107\penalty0 (499):\penalty0 1106--1118, 2012.

\end{thebibliography}

\clearpage

\appendix

\bigskip

\begin{center}
	{\large \bf ONLINE APPENDIX}
\end{center}

\bigskip

\setcounter{page}{1}
\setcounter{section}{0}
\setcounter{equation}{0}
\setcounter{table}{0}
\setcounter{figure}{0}
\renewcommand{\theequation}{OA.\arabic{equation}}
\renewcommand\thetable{OA.\arabic{table}}
\renewcommand\thefigure{OA.\arabic{figure}}
\renewcommand\thesection{OA.\arabic{section}}
\renewcommand\thepage{Online Appendix - \arabic{page}}
\renewcommand\thetheorem{OA.\arabic{theorem}}


\section{Boosting \label{sec:boosting}}
In this section, we discuss supporting theory for the weighted boosting procedure. Boosting amounts to combining several ``weak'' binary decisions into more sophisticated and powerful decision rules; see \cite{friedman2001elements}, Ch. 10. The weak binary decision is typically constructed with shallow decision trees. An interesting feature of boosting is that the ``weak'' decisions may be only slightly better than a random guess, while their combination may achieve outstanding out-of-sample performance. For the asymmetric loss function, boosting amounts to solving the following empirical risk minimization problem:
\begin{equation*}
	\inf_{f\in\mathscr{F}^{\rm B}}\frac{1}{n}\sum_{i=1}^n\omega(Y_i,X_i)\phi(-Y_if(X_i))
\end{equation*}
with $\mathscr{F}^{\rm B} = \left\{\sum_{j=1}^p\theta_j\varphi_j(x):\;|\theta|_1\leq \lambda,\; \varphi_j\in\mathscr{G},\; p\in\N\right\}$, where $\mathscr{G}$ is a base class of weak predictions. Exponential convexifying function $\phi(z)=\exp(z)$ is a popular choice, but the logistic function is similar; see \cite{friedman2001elements}, Ch. 10. The problem is then solved using a functional version of the gradient descent algorithm, often with additional regularization and tuning. Let $\hat f_n$ be a solution of the empirical risk minimization problem described above, then:
\begin{theorem}\label{thm:boosting}
	Suppose that $\mathscr{G}$ is a measurable class of functions from $\mathcal{X}$ to $[-1,1]$ with $\mathrm{VC}$-dimension $V\geq 1$ and that $\phi$ is a logistic or exponential function. Then under assumptions of Theorem~\ref{thm:oracle inequality}
	\begin{equation*}
		\E\left[\mathcal{R}(\sign(\hat f)) - \mathcal{R}^*\right]  \lesssim \left[\left(\frac{C_V}{n}\right)^{\frac{2+V}{2(1+V)}} + \inf_{f\in\mathscr{F}^{\rm B}}{\mathcal{R}}_{\phi}(f) - \mathcal{R}_\phi^*\right]^\frac{\alpha + 1}{\alpha+2}
	\end{equation*}
	for some constant $C_V>0$.
\end{theorem}
\begin{proof}[Proof of Theorem~\ref{thm:boosting}]	
	By Lemmas~\ref{lemma:logistic}, \ref{lemma:exponential}, and \ref{lemma:curvature_exp_log}, $\gamma = 1/2$ and $\kappa = 1$. By \cite{koltchinskii2011oracle}, Example 5 on p. 87, $\psi_{n,1}^\sharp(\epsilon) \leq \left(C_V/(n\epsilon^2)\right)^{\frac{2+V}{2(1+V)}}$ for some constant $C_V>0$ depending on $V$. Therefore, by Theorem~\ref{thm:oracle inequality}, for every $t>0$ with probability at least $1-c_qe^{-t}$
	{\footnotesize
		\begin{equation*}
			\mathcal{R}(\sign(\hat f)) - \mathcal{R}^* \leq C_\phi\left[\frac{c_\phi}{6}\left(\frac{C_V}{n\epsilon^2}\right)^{\frac{2+V}{2(1+V)}} +  \left(\frac{3(48qLM)^2}{2c_\phi} \vee {6q}\right)\frac{t}{n} +  2\inf_{f\in\mathscr{F}^{\rm B}}\left\{{\mathcal{R}}_{\phi}(f) - \mathcal{R}_\phi^*\right\}\right]^\frac{\alpha + 1}{\alpha+2},
	\end{equation*}}
	where the constants $\epsilon,c_q,C_\phi>0$ and $q>1$ are as in the proof of Theorem~\ref{thm:oracle inequality}. The result follows from integrating the tail bound.
\end{proof}

\noindent In the special case of the symmetric binary classification and correctly specified boosting class, Theorem~\ref{thm:boosting} recovers the bound discussed in Section 5.4 of \cite{boucheron2005theory}. Note that the statistical accuracy of a binary decision is driven by the complexity of the base class $\mathscr{G}$, which should have a small VC dimension to reduce its impact on the first term and, at the same time, generate a sufficiently rich class $\mathscr{F}^B$ to make the approximation error as small as possible.

\begin{example}
	Let $(R_k)_{k=1}^K$ be a tree-structured partition of $\mathcal{X}$ with cuts parallel to coordinate axes. Consider the class of decision trees with $K$ terminal nodes
	\begin{equation*}
		\mathscr{G} = \left\{x\mapsto 2\sum_{k=1}^K\one_{R_k}(x) - 1 \right\}.
	\end{equation*}
	The VC-dimension of $\mathscr{G}$ is $V\leq d\log(2d)$, see \cite{devroye1996probabilistic} and the approximation error tends to zero as $\lambda\to\infty$; see \cite{breiman2000some}.
\end{example}

\section{LASSO\label{appsec:lasso}}
In this section, we consider the uniform excess risk bounds for convexified empirical risk minimization with a LASSO penalty. For a vector $\theta=(\theta_1\dots,\theta_p)^\top\in\R^p$, let $|\theta|_q = \left(\sum_{j=1}^p|\theta_j|^q\right)^{1/q}$ denote its $\ell_q$ norm when $q\geq 1$ and let $|\theta|_0 = \sum_{j=1}^p\one_{\theta_j\ne 0}$.

Recall that the weighted convexified empirical risk is
\begin{equation*}
	\widehat{\mathcal{R}}_{\phi}(f) = \frac{1}{n}\sum_{i=1}^n\omega(Y_i,X_i)\phi(-Y_if(X_i)).
\end{equation*}
For a finite dictionary $\left\{\varphi_1, \ldots, \varphi_p\right\}$ with  $\varphi_j:\mathcal{X}\to\R$, consider the function
\begin{equation*}
	f_{\theta}(x)=\sum_{j=1}^{p} \theta_{j} \varphi_{j}(x), \quad \theta\in\Theta\subset \R^p
\end{equation*}
and the corresponding binary decision rule $\sign(f_\theta(x))$. Note that this setting covers linear decision rules, $f_\theta(x) = \sum_{j=1}^{p}\theta_{j}x_j$, as a special case. Consider the binary decision rule $\sign(f_{\hat\theta})$ with $\hat\theta$ solving the weighted empirical risk minimization problem with a LASSO penalty
\begin{equation*}
	\inf_{\theta}\widehat{\mathcal{R}}_{\phi}(f_\theta)+\lambda_n|\theta|_1,
\end{equation*}
where $\lambda_n\downarrow0$ is a sequence of regularization parameters as described in the following assumption.

\begin{assumption}\label{as:tuning}
	Suppose that for some $c>1$ and $\delta\in(0,1)$, the tuning parameter satisfies
	\begin{equation*}
		\lambda_n \geq 8cLF^*\sqrt{\frac{2\log(2p)}{n}} + 4cLMF^*\sqrt{\frac{2\log(1/\delta)}{n}},
	\end{equation*}
	where $F^*$ is a constant such that $\max_{1\leq j\leq p}|\varphi_j(X)|\leq F^*$ and the constants $L,M$ are as in Assumptions~\ref{as:losses} and \ref{as:phi}.
\end{assumption}

Put $\varphi(X)=(\varphi_1(X),\dots,\varphi_p(X))^\top$. The following assumption states the identification condition, known as the restricted eigenvalue condition; see \cite{Bickel2009simultaneous}.
\begin{assumption}\label{as:grammatrix}
	For an integer $s\leq p$
	\begin{equation*}
		\Phi^2(S) \triangleq \min_{\substack{\Delta\ne 0 \\ |\Delta_{S^c}|_1\leq c_0|\Delta_{S}|_1}} \frac{\Delta^{\top}\E\left[ \varphi(X)\varphi(X)^\top\right]  \Delta}{|\Delta_{S}|_2^{2}}>0,\qquad \forall S\subset\{1,\dots,p\}\text{ with } |S|\leq s,
	\end{equation*}
	where for $\Delta\in\R^p$ and $S\subset\{1,\dots,p\}$, we use $\Delta_S\in\R^p$ to denote the vector with the same coordinates as $\Delta$ on $S$ and zeros on $S^c$, and $c_0 = (2c+1)/(c-1)$ with $c>1$ as in Assumption~\ref{as:tuning}.
\end{assumption}
This condition is not very restrictive and is satisfied whenever the matrix $\E\left[\varphi(X)\varphi(X)^\top\right]$ has its smallest eigenvalue bounded away from zero; see also \cite{Belloni2015some}.
Let $\theta^*$ be a solution to
\begin{equation*}\label{eq:oracle}
	\inf_{\theta:|S_\theta|\leq s}\left\{ 6(2\kappa-1)\left(\frac{2\sqrt{|S_\theta|}\lambda_n}{\Phi(S_\theta)c_\phi^{1/2\kappa}} \right)^\frac{2\kappa}{2\kappa - 1} +  3[\mathcal R_\phi(f_\theta)-\mathcal R_\phi^*]\right\},
\end{equation*}
where $S_\theta=\{1\leq j\leq p:\theta_j\ne 0 \}$ is the support of $\theta\in\R^p$. 

The following result describes the excess risk bounds for binary decisions obtained with asymmetric LASSO; see also \cite{Chetverikov2021analytic}, \cite{Belloni2018high}, \cite{geer2008}, and \cite{wegkamp2007lasso} for related results.

\begin{theorem}\label{thm:lasso}
	Suppose that Assumptions \ref{as:losses}, \ref{as:phi}, \ref{as:tsybakov}, \ref{as:loss_modulus}, \ref{as:tuning}, and \ref{as:grammatrix} are satisfied. Then with probability at least $1-\delta$
	\begin{equation*}
		\mathcal{R}(\sign(f_{\hat\theta})) - \mathcal{R}^* \lesssim \left[\left(s\lambda_n^2 \right)^\frac{\kappa}{2\kappa - 1} + \mathcal{R}_\phi(f_{\theta^*})-\mathcal{R}_\phi^*
		\right]^\frac{\gamma(\alpha+1)}{\gamma\alpha+1}.
	\end{equation*}
\end{theorem}
\begin{proof}
	Since $\hat\theta$ is a minimizer,  for every $\theta\in\Theta$, we have
	\begin{equation}\label{eq:inequality_optimum}
		\widehat{\mathcal{R}}_\phi(f_{\hat{\theta}}) + \lambda_n|\hat\theta|_1 \leq \widehat{\mathcal{R}}_\phi(f_{{\theta}}) + \lambda_n|\theta|_1.
	\end{equation}
	In particular, for $\theta=0$ we have
	\begin{equation*}
		\widehat{\mathcal{R}}_\phi(f_{\hat{\theta}}) + \lambda_n|\hat\theta|_1 \leq \frac{1}{n}\sum_{i=1}^n\omega(Y_i,X_i)\phi(0)\leq 4M.
	\end{equation*}
	where the last inequality follows under Assumptions~\ref{as:losses} (iii) and \ref{as:phi} (i). Therefore, $|\hat\theta|_1\leq 4M/\lambda_n$, and we will restrict the parameter space in the definition of $\hat\theta$ and $\theta^*$ to $\Theta_n=\{\theta:\; |\theta|_1\leq 4M/\lambda_n\}$, so that $|\hat\theta - \theta^*|_1 \leq 8M/\lambda_n\triangleq K_n$.
	
	For $g_\theta(x,y) \triangleq \omega(y,x)\phi(-yf_\theta(x))$, put $P_ng_\theta = \frac{1}{n}\sum_{i=1}^ng_\theta(X_i,Y_i)$ and $Pg_\theta = \E g_\theta(X,Y)$. Consider the following class
	\begin{equation*}
		\mathscr{G}_n = \left\{\frac{g_\theta - g_{\theta^*}}{|\theta-\theta^*|_1}:\;\theta\in\R^p,\; |\theta-\theta^*|_1\leq K_n \right\}
	\end{equation*}
	and let $\|P_n-P\|_{\mathscr{G}_n}=\sup_{g\in\mathscr{G}_n}|(P_n-P)g|$ be the supremum of the empirical process indexed by $\mathscr{G}_n$.
	With this notation, note that $\widehat{\mathcal{R}}_\phi(f_\theta) = P_ng_\theta$ and $\mathcal{R}_\phi(f_\theta)=Pg_\theta$.

	By Lemma~\ref{lemma:concentration} and Assumption~\ref{as:tuning} with probability at least $1-\delta$, we have $\|P_n - P\|_{\mathscr{G}_n} \leq \lambda_n/c$. Therefore,
	\begin{equation*}
		\begin{aligned}
			\mathcal{R}_\phi(f_{\hat\theta}) + \lambda_n|\hat\theta|_1 & = (P_n-P)(g_{\theta^*} - g_{\hat\theta}) + \mathcal{R}_\phi(f_{\theta^*}) + P_n(g_{\hat\theta}-g_{\theta^*}) + \lambda_n|\hat\theta|_1 \\
			& \leq (P_n-P)(g_{\theta^*} - g_{\hat\theta}) + \mathcal{R}_\phi(f_{\theta^*}) + \lambda_n|\theta^*|_1	 \\
			& \leq \|P_n - P\|_{\mathscr{G}_n}|\hat\theta - \theta^*|_1 + \mathcal{R}_\phi(f_{\theta^*}) + \lambda_n|\theta^*|_1 \\
			& \leq \frac{\lambda_n}{c}|\hat\theta - \theta^*|_1 + \mathcal{R}_\phi(f_{\theta^*}) + \lambda_n|\theta^*|_1,
		\end{aligned}
	\end{equation*}
	where the second line follows from equation (\ref{eq:inequality_optimum}) with $\theta=\theta^*$.
	
	Let $S_*=S_{\theta^*}$ be the support of $\theta^*$. Put $\Delta \triangleq \hat\theta - \theta^*$ and note that $|\Delta|_1 = |\Delta_{S_*}|_1 + |\Delta_{S_*^c}|_1,\forall\Delta\in\R^p$. Note also that $|\Delta_{S_*^c}|_1 = |\hat\theta_{S_*^c}|_1$ and that $|\theta^*|_1 = |\theta^*_{S_*}|_1$. Using these properties, we obtain
	\begin{equation*}
		\mathcal{R}_\phi(f_{\hat\theta}) +\lambda_n|\Delta_{S_*^c}|_1 \leq \frac{\lambda_n}{c}\left\{|\Delta_{S_*}|_1 + |\Delta_{S_*^c}|_1\right\} + \mathcal{R}_\phi(f_{\theta^*})  + \lambda_n\left\{|\theta^*_{S_*}|_1 - |\hat\theta_{S_*}|_1 \right\}.
	\end{equation*}
	By the triangle inequality $|\theta^*_{S_*}|_1 - |\hat\theta_{S_*}|_1 \leq |\Delta_{S_*}|_1$, and so
	\begin{equation}\label{eq:inequality}
		c[\mathcal{R}_\phi(f_{\hat\theta})-\mathcal{R}_\phi^*] + (c-1)\lambda_n|\Delta_{S_*^c}|_{1} \leq c[\mathcal{R}_\phi(f_{\theta^*})-\mathcal{R}_\phi^*] + (c+1)\lambda_n|\Delta_{S_*}|_1.
	\end{equation}
	If $\lambda_n|\Delta_{S_*}|_1 < \mathcal{R}_\phi(f_{\theta^*})-\mathcal{R}_\phi^*$, then equation (\ref{eq:inequality}) implies
	\begin{equation}\label{eq:inequality2}		\mathcal{R}_\phi(f_{\hat\theta})-\mathcal{R}_\phi^* \leq \frac{2c+1}{c}[\mathcal{R}_\phi(f_{\theta^*})-\mathcal{R}_\phi^*].
	\end{equation}	
	Suppose now that $\lambda_n|\Delta_{S_*}|_1\geq  [\mathcal{R}_\phi(f_{\theta^*})-\mathcal{R}_\phi^*]$. Then equation (\ref{eq:inequality}) implies
	$(c-1)\lambda_n|\Delta_{S_*^c}|_1$ $\leq$  $(2c+1) \lambda_n|\Delta_{S_*}|_{1}.$
	This shows that $|\Delta_{S_*^c}|_1 \leq  c_0|\Delta_{S_*}|_1$ with $c_0 = (2c+1)/(c-1)$. Therefore, by the Cauchy--Schwarz inequality under Assumption~\ref{as:grammatrix}
	\begin{equation*}
		|\Delta_{S_*}|^2_1 \leq s|\Delta_{S_*}|_2^2 \leq \frac{s}{\Phi_*^2}\Delta^\top\E\left[ \varphi(X)\varphi(X)^\top\right]\Delta = \frac{s}{\Phi_*^2}\left\|\sum_{j=1}^{p} \Delta_{j}\varphi_{j}\right\|^2 = \frac{s}{\Phi_*^2}\|f_{\hat\theta} - f_{\theta^*}\|^{2},
	\end{equation*}
	where we put $\Phi_*=\Phi(S_*)$. Then by adding $(c-1)\lambda_n|\Delta_{S_*}|_1$ both sides, we obtain from equation (\ref{eq:inequality}):
	$c[\mathcal{R}_\phi(f_{\hat\theta})-\mathcal{R}_\phi^*]$ + $(c-1)\lambda_n|\Delta|_{1}$ $\leq$ $3c\lambda_n|\Delta_{S_*}|_1$ $\leq$ $3c\frac{\sqrt{s}\lambda_n}{\Phi_*}\|f_{\hat\theta} - f_{\theta^*}\|.$
	
	By the triangle inequality, Jensen's inequality, and Assumption~\ref{as:loss_modulus}
	\begin{equation*}
		\begin{aligned}
			\|f_{\hat\theta} - f_{\theta^*}\| & \leq \|f_{\hat\theta} - f_\phi^*\| + \|f_{\theta^*} - f_\phi^*\| \\
			& \leq \left\{c_\phi^{-1}\left[\mathcal{R}_\phi(f_{\hat\theta}) - \mathcal{R}_\phi^* \right] \right\}^{1/2\kappa} + \left\{c_\phi^{-1}\left[\mathcal{R}_\phi(f_{\theta^*}) - \mathcal{R}_\phi^* \right] \right\}^{1/2\kappa} \\
			& \leq \frac{2^{1-1/2\kappa}}{c_\phi^{1/2\kappa}} \left[\mathcal{R}_\phi(f_{\hat\theta})-\mathcal{R}_\phi^* + \mathcal{R}_\phi(f_{\theta^*})-\mathcal{R}_\phi^* \right]^{1/2\kappa}.
		\end{aligned}
	\end{equation*}
	Therefore, since $\kappa\geq 1$
	\begin{equation*}
		\begin{aligned}
			c[\mathcal{R}_\phi(f_{\hat\theta})-\mathcal{R}_\phi^*] + (c-1)\lambda_n|\Delta|_{1} & \leq 3c\frac{2\sqrt{s}\lambda_n}{\Phi_*c_\phi^{1/2\kappa}} 2^{-1/2\kappa}\left[\mathcal{R}_\phi(f_{\hat\theta})-\mathcal{R}_\phi^* + \mathcal{R}_\phi(f_{\theta^*})-\mathcal{R}_\phi^* \right]^{1/2\kappa} \\
			& \leq 1.5c(2\kappa-1)\left(\frac{2\sqrt{s}\lambda_n}{\Phi_*c_\phi^{1/2\kappa}}\right)^\frac{2\kappa}{2\kappa - 1} \\
			& \qquad + \frac{3c}{4}\left[\mathcal{R}_\phi(f_{\hat\theta})-\mathcal{R}_\phi^* + \mathcal{R}_\phi(f_{\theta^*})-\mathcal{R}_\phi^* \right],
		\end{aligned}
	\end{equation*}
	where we use the convex conjugate inequality $uv\leq \frac{u^{2\kappa}}{2\kappa}+\frac{(2\kappa-1)}{2\kappa}v^{\frac{2\kappa}{2\kappa-1}}$. Rearranging this inequality
	\begin{equation}
		\label{eq:inequality3}
		\mathcal{R}_\phi(f_{\hat\theta})-\mathcal{R}_\phi^* + \frac{4(c-1)}{c}\lambda_n|\Delta|_{1} \leq 6(2\kappa-1)\left(\frac{2\sqrt{s}\lambda_n}{\Phi_* c_\phi^{1/2\kappa}} \right)^\frac{2\kappa}{2\kappa - 1} + 3[\mathcal{R}_\phi(f_{\theta^*})-\mathcal{R}_\phi^*]
	\end{equation}	
	Combining this equation with the inequality in  equation (\ref{eq:inequality2}), since $(2c+1)/c\leq 3$ we always have
	\begin{equation*}
		\mathcal{R}_\phi(f_{\hat\theta})-\mathcal{R}_\phi^*  \leq 6(2\kappa-1)\left(\frac{2\sqrt{s}\lambda_n}{\Phi_*c_\phi^{1/2\kappa}} \right)^\frac{2\kappa}{2\kappa - 1} +  3[\mathcal{R}_\phi(f_{\theta^*})-\mathcal{R}_\phi^*].
	\end{equation*}
	
	Lastly, under Assumptions~\ref{as:losses}, \ref{as:phi}, and \ref{as:tsybakov}, by Theorem~\ref{thm:convexified_risk_bound_margin}
	\begin{equation*}
		\mathcal{R}(\sign(f_{\hat\theta})) - \mathcal{R}^* \leq C_\phi\left[\mathcal{R}_\phi(f_{\hat\theta}) - \mathcal{R}_\phi^*\right]^\frac{\gamma(\alpha + 1)}{\gamma\alpha + 1}.
	\end{equation*}
	This gives
	\begin{equation*}
		\mathcal{R}(\sign(f_{\hat\theta})) - \mathcal{R}^* \leq C_\phi\left[ 6(2\kappa-1)\left(\frac{2\sqrt{s}\lambda_n}{\Phi_*c_\phi^{1/2\kappa}} \right)^\frac{2\kappa}{2\kappa - 1} +  3[\mathcal{R}_\phi(f_{\theta^*})-\mathcal{R}_\phi^*]
		\right]^\frac{\gamma(\alpha+1)}{\gamma\alpha+1}.
	\end{equation*}
\end{proof}

\begin{lemma}\label{lemma:concentration}
	Suppose that Assumptions \ref{as:losses}, \ref{as:phi}, and \ref{as:tuning} are satisfied. Then with probability at least $1-\delta$
	$\|P_n-P\|_{\mathscr{G}_n}$ $\leq$  $\frac{\lambda_n}{c},$
	where $\|P_n-P\|_{\mathscr{G}_n}$ is as in the proof of Theorem~\ref{thm:lasso}.
\end{lemma}
\begin{proof}
	By \cite{koltchinskii2011oracle}, Theorem 2.5, for every $\delta\in(0,1)$
	\begin{equation*}
		\Pr\left(\|P_n - P\|_{\mathscr{G}_n} < \E\|P_n - P\|_{\mathscr{G}_n} + 4LMF^*\sqrt{\frac{2\log(1/\delta)}{n}}\right) \geq 1-\delta,
	\end{equation*}
	where we use the fact that
	\begin{equation*}
		\begin{aligned}
			\sup_{\theta:|\theta-\theta^*|_1\leq K_n}\frac{|g_\theta - g_{\theta^*}|}{|\theta-\theta^*|_1} & = \sup_{\theta:|\theta-\theta^*|_1\leq K_n}\frac{|\omega(Y,X)|\left| \phi(-Yf_\theta(X)) - \phi(-Yf_{\theta^*}(X))\right|}{|\theta-\theta^*|_1} \\
			&\leq 4LM\frac{\left\vert {f_{\theta^*}(X)-f_{\theta}(X)} \right\vert}{{|\theta-\theta^*|_1}} \\
			&\leq 4LMF^*,
		\end{aligned}
	\end{equation*}
	which follows under Assumptions~\ref{as:losses} (iii), \ref{as:phi} (i)-(ii), and \ref{as:tuning}. Let $(\varepsilon_i)_{i=1}^n$ be i.i.d.\ Rademacher random variables. Then  by the symmetrization, see \cite{koltchinskii2011oracle}, Theorems 2.1,
	\begin{equation*}
		\begin{aligned}
			\E\|P_n - P\|_{\mathscr{G}_n} & \leq 2\mathbb{E} \sup _{\theta:|\theta-\theta^*|_1 \leq K_n}\frac{\left|\frac{1}{n} \sum_{i=1}^{n} \varepsilon_{i}\omega(Y_i,X_i)\left[\phi(-Y_if_\theta(X_i)) - \phi(-Y_if_{\theta^*}(X_i))\right]\right|}{|\theta-\theta^*|_1} \\
			& \leq 8L \mathbb{E} \sup _{\theta:|\theta-\theta^*|_1\leq K_n}\frac{\left|\frac{1}{n} \sum_{i=1}^{n} \varepsilon_{i}\left(f_\theta(X_i)-f_{\theta^*}(X_i)\right)\right|}{|\theta-\theta^*|_1} \\
			& \leq 8L\E\left[\max_{1\leq j\leq p}\left|\frac{1}{n} \sum_{i=1}^{n} \varepsilon_{i} \varphi(X_i)\right|\right],
		\end{aligned}
	\end{equation*}
	where the second line follows by the contraction inequality, see \cite{buhlmann2011statistics}, Theorem 14.4 and \cite{koltchinskii2011oracle}, Theorem 2.3, since under Assumptions~\ref{as:losses} (iii) and \ref{as:phi} (ii), $\left|\omega(Y_i,X_i)\left[\phi(-Y_if_\theta(X_i)) - \phi(-Y_if_{\theta^*}(X_i))\right]\right|\leq 4LM|f_\theta(X_i)-f_{\theta^*}(X_i)|$; and the third by H\"{o}lder's inequality.
	
	By \cite{buhlmann2011statistics}, Lemma 14.14
	\begin{equation*}
		\begin{aligned}
			\E\left[\max_{1\leq j\leq p}\left|\frac{1}{n} \sum_{i=1}^{n} \varepsilon_{i} \varphi_j(X_i)\right|\right] \leq F^*\sqrt{\frac{2\log(2p)}{n}},
		\end{aligned}
	\end{equation*}
	which under Assumption~\ref{as:tuning} shows that with probability at least $1-\delta$,  
	$\|P_n - P\|_{\mathscr{G}_n}$ $<$ $8LF^*\sqrt{\frac{2\log(2p)}{n}}$ + $4LMF^*\sqrt{\frac{2\log(1/\delta)}{n}} \leq \frac{\lambda_n}{c}.$
\end{proof}
\section{Shallow learning \label{appsec:shallow}}

Shallow learning amounts to fitting a neural network with one or two hidden layers. Neural networks are widely used in econometrics at least since \cite{gallant1988there}.\footnote{Conceptually, neural networks can be traced to early mathematical models of the brain; see \cite{mcculloch1943logical} and \cite{rosenblatt1958perceptron}. Among the early econometric applications, we may cite nonlinear time series modeling and forecasting; see \cite{granger1995modelling}, \cite{lee1993testing}, \cite{gallant1992learning}, \cite{white2001statistical}, and \cite{chen2001semiparametric}; and asset pricing, see \cite{hutchinson1994nonparametric} and \cite{chen2009land}, among others.}  We focus on a very simple neural network class with two hidden layers. Following recent trends in big data applications, we refer to this approach as shallow learning, which, in contrast to deep learning, does not allow for the number of layers to scale with the sample size. Consider a single layer neural network class
\begin{equation*}
	\Theta_n^{\rm S} = \left\{x\mapsto \sum_{j=1}^{p_n}b_j\sigma_0(a_j^\top x + a_{0,j}) + b_0,\quad a\in\R^{d+1},\;|b|_1\leq \gamma_n\right\},
\end{equation*}
where $a=(a_0,a_1,\dots, a_d)$ and $b=(b_0,b_1,\dots,b_{p_n})$ are parameters to be estimated,  $\sigma_0$ is some smooth function, and $|.|_1$ is the $\ell_1$ norm. Our shallow learning class is defined as a hybrid network
\begin{equation*}
	\mathscr{F}_n^{S} = \left\{x\mapsto \sigma(\theta(x) + c(x)d + 1) - \sigma(\theta(x) + c(x)d - 1) - 1:\;\theta\in\Theta_n^{\rm S},\;|d|\leq n \right\},
\end{equation*}
where $c(x)$ pertains to asymmetry (see equation (\ref{eq:c-function})) and $\sigma(z)=(z)_+$ is the Rectified Linear Unit (ReLU) activation function. The shallow learning class can be visualized on a directed graph, see Figure~\ref{fig:network_shallow}. All covariates are fed first into the hidden layer 1 corresponding to $\Theta_n^{\rm S}$. This network class consists of $p_n$ neurons with a smooth activation function $\sigma_0$. The output produced by the hidden layer 1 is fed subsequently into the two neurons with the ReLU activation function $\sigma$. Note that this last layer has a single free parameter $b$. The output $f\in[-1,1]$, is obtained by summing the two neurons from the ReLU layer. The final binary decision is $\sign(f)$.

\begin{figure}[htp]
	\centering
	 \resizebox{0.6\textwidth}{!}{
		\begin{tikzpicture}[x=1.5cm, y=1.5cm]
			\tikzset{%
				neuron/.style={
					circle,
					minimum size=0.9cm},
			}
			\foreach \m in {1,2,3}
			\node [neuron, fill=teal!50] (input-\m) at (0,2.5-\m) {$X_\m$};
			
			\foreach \m in {1,...,4}
			\node [neuron, fill=blue!30] (hidden1-\m) at (2,3-\m) {$\sigma_0$};
			
			\foreach \m in {1,2}
			\node [neuron, fill=violet!40] (relu-\m) at (4,2-\m) {$\sigma$};
			
			\node [neuron, fill=orange!50] (c) at (1.8,-2) {$c$};
			
			\foreach \m in {1}
			\node [neuron, fill=red!60] (output-\m) at (6,1.5-\m) {$\hat f_n$};
			
			\foreach \i in {1,2,3}
			\draw [gray, ->] (input-\i) -- (c);
			
			\foreach \i in {1,2,3}
			\foreach \j in {1,...,4}
			\draw [gray, ->] (input-\i) -- (hidden1-\j);
			
			\foreach \i in {1,...,4}
			\foreach \j in {1,2}
			\draw [gray, ->] (hidden1-\i) -- (relu-\j);
			
			\foreach \i in {1,2}
			\foreach \j in {1}
			\draw [gray, ->] (relu-\i) -- (output-\j);
			
			\foreach \i in {1,2}
			\draw [gray, ->] (c) -- (relu-\i);
			
			\foreach \l [count=\x from 0] in {Input, Layer 1, 2 ReLU, Output}
			\node [align=center, above] at (\x*2,2.5) {\l };
	\end{tikzpicture}}
	
	\caption{Directed graph of the shallow learning architecture with $d=3$ covariates, a single hidden layer with $4$ neurons, and $2$ outer ReLU neurons. The orange neuron takes covariates $X\in\R^d$ as input and produces $c(X)\in\R$, which is fed directly into $2$ ReLU neurons.}
	\label{fig:network_shallow}
\end{figure}

The soft shallow learning prediction $\hat f_n:\mathcal{X}\to[-1,1]$ is a solution to the convexified empirical risk minimization problem with hinge convexifying function\footnote{Our focus on the hinge convexification function is motivated by the objective of achieving the minimax optimal convergence rates. The construction with 2 ReLU neurons in conjunction with the hinge function allows approximating the risk $\mathcal{R}^*_\phi$ sufficiently fast. It is not obvious whether logistic convexification can achieve the minimax optimal convergence rate, and we leave more detailed investigation of this for future research.}
\begin{equation*}
	\inf_{f \in \mathscr{F}_n^{\rm S}} \frac{1}{n}\sum_{i=1}^n\omega(Y_i,X_i)(1-Y_if(X_i))_+.
\end{equation*}
To describe the accuracy of the shallow learning binary decision $\sign(\hat f_n)$, consider the Sobolev ball of smoothness $\beta\in\mathbf{N}$ and radius $M\in(0,\infty)$
\begin{equation*}
	W^{\beta,\infty}_M(\mathcal{X}) = \left\{f:\mathcal{X}\to\R:\;\max_{|k|\leq \beta}\mathrm{esssup}_{x\in\mathcal{X}}|D^k f(x)|\leq M\right\},
\end{equation*} 
where we use the multi-index notation $k=(k_1,\dots,k_d)\in\N^d$, $|k|=k_1+\dots+k_d$, and $D^kf = \frac{\partial^{|k|}}{\partial x_1^{k_1}\dots\partial x_d^{k_d}}f$.

The following assumption imposes some mild regularity conditions on the activation function $\sigma_0$ and the Sobolev smoothness of the conditional probability $\eta$.
\begin{assumption}\label{as:shallow_learning}
	(i) $\sigma_0:\R\to[-b,b]$ is non-decreasing and Lipschitz continuous, infinitely differentiable on an open interval containing a point $x_0$ with $D^k\sigma_0(x_0)\ne 0$ for all $k\in\mathbf{Z}_+$; (ii) $\eta\in W^{\beta,\infty}_M(\mathcal{X})$ for some $\beta\in\N$ and $0<M<\infty$, where $\mathcal{X}\subset \R^d$ is a Cartesian product of compact intervals; (iii) $p_n$ and $\gamma_n$ are of polynomial order.
\end{assumption}

Assumption~\ref{as:shallow_learning} (i) rules out polynomial activation functions and allows for the sigmoid function $\sigma_0(x) = (1+e^{-x})^{-1}$. It also rules out the ReLU activation function, which is a more natural choice for deep learning problems.  

\begin{theorem}\label{thm:oracle inequality_S}
	Suppose that $(Y_i,X_i)_{i=1}^n$ is an i.i.d. sample following a distributions satisfying Assumptions~\ref{as:losses}, \ref{as:phi}, \ref{as:tsybakov}, and \ref{as:shallow_learning}, and denoted $\mathcal{P}(\alpha,\beta)$. Then there exist $c,C>0$ such that for every $t>0$ with probability at least $1-ce^{-t}$
	\begin{equation*}
		\mathcal{R}(\sign(\hat f_n)) - \mathcal{R}^* \leq C\left[\left(\frac{p_n\log^2n}{n}\right)^\frac{1+\alpha}{2+\alpha} + p_n^{-(1+\alpha)\beta/d} + \left(\frac{t}{n}\right)^\frac{1+\alpha}{2+\alpha} + \frac{t}{n}\right]
	\end{equation*}
	uniformly over $\mathcal{P}(\alpha,\beta)$. In particular,
	\begin{equation*}
		\sup_{P\in\mathcal{P}(\alpha,\beta)}\E_P\left[\mathcal{R}(\sign(\hat f_n)) - \mathcal{R}^*\right] \lesssim \left(\frac{\log^2 n}{n}\right)^\frac{(1+\alpha)\beta}{(2+\alpha)\beta + d}
	\end{equation*}
	provided that $p_n\sim (n/\log^2 n)^\frac{d}{(2+\alpha)\beta + d}$.
\end{theorem}
\noindent It is worth noting that the convergence rate of the excess risk of a shallow learning prediction can be anywhere between the slow nonparametric rate $O(n^{-\beta/(2\beta+d)})$ and the fast rate $O(n^{-1})$ depending on the margin exponent $\alpha$. In particular, we can partially offset the curse of dimensionality if $\alpha$ is sufficiently large, e.g., for $\alpha=d/\beta$, the rate is $O(n^{-1/2})$. This is another manifestation of the fact that predicting a binary outcome is easier than predicting real-valued variables. Note also that the smoothness of the decision boundary itself, $\{x\in[0,1]^d:\; \eta(x)-c(x)=0\}$, does not directly play a role, and only the smoothness of $\eta$ is important. This is probably not surprising in light of the fact that $c$ is known to the decision-maker.

In the special case when the loss function is symmetric, under the mild assumption that the density of covariates is uniformly bounded, it follows from \cite{audibert2007fast}, Theorem 4.1 that there exists $C>0$ such that for every $n\geq 1$
\begin{equation}\label{eq:lower_bound}
	\inf_{\hat f_n}\sup_{P\in\mathcal{P}(\alpha,\beta)}\E_P\left[\mathcal{R}(\hat f_n) - \mathcal{R}^*\right]\geq Cn^{-\frac{(1+\alpha)\beta}{(2+\alpha)\beta + d}},
\end{equation}
where the infimum is taken over all binary decisions $\hat f_n:\mathcal{X}\to\{-1,1\}$ computed from an i.i.d.\ sample $(Y_i,X_i)_{i=1}^n$. Therefore, apart from a $\log n$ factor, our result shows that binary decisions produced by shallow learning are optimal from the minimax point of view.

\begin{proof}[Proof of Theorem~\ref{thm:oracle inequality_S}]
	By Theorem~\ref{thm:oracle inequality}, for every $t>0$ with probability at least $1-c_qe^{-t}$
	\begin{equation*}
		\mathcal{R}(\sign(\hat f_n)) - \mathcal{R}^* \lesssim \psi_{n,1+1/\alpha}^{\sharp}(\epsilon) + \left(\frac{t}{n}\right)^\frac{1+\alpha}{2+\alpha} + \frac{t}{n} + \inf_{f\in\mathscr{F}_n^{\rm S}}{\mathcal{R}}_{\phi}(f) - \mathcal{R}_\phi^*,
	\end{equation*}
	where we use the fact that $\gamma=1$ and $\kappa = 1+1/\alpha$ by Lemmas~\ref{lemma:hinge} and \ref{lemma:curvature_hinge}. By Lemmas~\ref{lemma:lrc_shallow_learning} and \ref{lemma:lrc_fp}
	\begin{equation*}
		\psi_{n,1+1/\alpha}^{\sharp}(\epsilon) \leq C\left(\frac{p_n\log p_n}{n}\log\left(\frac{n}{p_n\log p_n}\right)\right)^\frac{1+\alpha}{2+\alpha}.
	\end{equation*}
	Next, under Assumption~\ref{as:shallow_learning} (ii), by \cite{mhaskar1996neural}, Theorem 2.1, there exists $\eta_n\in\Theta_n^{\rm S}$ such that
	\begin{equation*}
		\|\eta_n - \eta\|_\infty \leq Cp_n^{-\beta/d}\triangleq\varepsilon_n.
	\end{equation*}
	Define
	\begin{equation*}
		f_n(x) = \left(\frac{\eta_n(x) - c(x)}{\varepsilon_n} + 1\right)_+ - \left(\frac{\eta_n(x) - c(x)}{\varepsilon_n} - 1\right)_+ - 1
	\end{equation*}
	and note that $f_n\in\mathscr{F}_n^{\rm S}$. Note also that
	\begin{equation*}
		f_n(x) = \begin{cases}
			1, & \text{if } \eta_n(x)-c(x)> \varepsilon_n \\
			\frac{\eta_n(x) - c(x)}{\varepsilon_n}, & \text{if }  |\eta_n(x)-c(x)|\leq \varepsilon_n \\
			-1, & \text{if }  \eta_n(x)-c(x)< -\varepsilon_n
		\end{cases}
	\end{equation*}
	and that on the event $\{x\in[0,1]^d:\; |\eta(x)-c(x)|>2\varepsilon_n \}$ we have $f_n=f_\phi^*$. To see this, recall that by Lemma~\ref{lemma:hinge}, $f_\phi^*=\sign(\eta-c)$. Then if $\eta-c>0$, we have $\eta_n-c=(\eta-c)-(\eta-\eta_n)>\varepsilon_n$ while if $\eta-c<0$, we have $\eta_n-c = (\eta-c) + (\eta_n - \eta)<-\varepsilon_n$. Therefore, by Lemma~\ref{lemma:equation_hinge}
	\begin{equation*}
		\begin{aligned}
			\inf_{f\in\mathscr{F}_n^{\rm S}}{\mathcal{R}}_{\phi}(f) - \mathcal{R}_\phi^* & = \inf_{f\in\mathscr{F}_n^{\rm S}}\int_{[0,1]^d}b|f-f_\phi^*||\eta - c|\dx P_X \\
			& \leq \int_{[0,1]^d}b|f_n-f_\phi^*||\eta - c|\dx P_X \\
			& = \int_{|\eta-c|\leq2\varepsilon_n}b|f_n-f_\phi^*||\eta - c|\dx P_X \\
			& \leq 8M\varepsilon_n\int_{|\eta-c|\leq2\varepsilon_n}b|f_n-f_\phi^*|\dx P_X \\ 
			& \leq 16M\varepsilon_nP_X(|\eta-c|\leq 2\varepsilon_n)\\
			& \leq 2^{4+\alpha} MC_m\varepsilon_n^{1+\alpha},
		\end{aligned}
	\end{equation*}
	where the last two lines follow under Assumptions~\ref{as:losses} (iii) and \ref{as:tsybakov}. Therefore, for every $t>0$ with probability at least $1-c_qe^{-t}$
	\begin{equation*}
		\begin{aligned}
			\mathcal{R}(\sign(\hat f_n)) - \mathcal{R}^* & \lesssim \left(\frac{p_n\log^2 n}{n}\right)^\frac{1+\alpha}{2+\alpha} + p_n^{-(1+\alpha)\beta/d} + 2^{4+\alpha} MC_mC^{1+\alpha}\left(\frac{t}{n}\right)^\frac{1+\alpha}{2+\alpha} + \frac{t}{n}.
		\end{aligned}
	\end{equation*}
	The second statement follows from integrating this tail bound in the same way as in the proof of Theorem~\ref{thm:parametric_predictions}. The uniformity follows from the fact that all constants do no depend on the specific distribution of $(X,Y)$ in $\mathcal{P}(\alpha,\beta)$.
\end{proof}

\section{Fixed points of local Rademacher complexities}\label{sec:local_rademacher}
In this section, we obtain useful bounds on the fixed points of local Rademacher complexities for shallow and deep neural network classes. 

Next, we consider the shallow learning class
\begin{equation*}
	\mathscr{F}_n^{\rm S} = \left\{\sigma(\theta + cd + 1) - \sigma(\theta + cd - 1) - 1:\; \theta\in\Theta_n^{\rm S},\; |d|\leq n\right\},
\end{equation*}
where $\sigma(z)=\max\{z,0\}$,
\begin{equation*}
	\Theta_n^{\rm S} = \left\{x\mapsto \sum_{j=1}^{p_n}b_j\sigma_0(a_j^\top x + a_{0,j}) + b_0,\quad a_j\in\R^{d},\quad\sum_{j=0}^{p_n}|b_j|\leq \gamma_n\right\},
\end{equation*}

The following result bounds the local Rademacher complexity of the shallow learning class $\mathscr{F}_n^{\rm S}$ in terms of the number of neurons $p_n$ in the inner neural network class $\Theta_n^{\rm S}$.
\begin{lemma}\label{lemma:lrc_shallow_learning}
	Suppose that (i) $\sigma_0:\R\to[-b,b]$ is non-decreasing with  Lipschitz constant $L_0<\infty$; (ii) $p_n= C_1n^{c_1}$ and $\gamma_n = C_2 n^{c_2}$ for some constants $c_j,C_j>0,j=1,2$; (iii) $\|c\|_\infty\leq \bar c$ and $\gamma_n>2/L_0$. Then for $a,A,K>0$,
	\begin{equation*}
		\psi_n(\delta;\mathscr{F}_n^{\rm S}) \leq K\left[\sqrt{\frac{p_n\delta(1\vee \log(An^a/\delta^{1/2}))}{n}}\bigvee \frac{p_n(1\vee \log(An^a/\delta^{1/2}))}{n}\right].
	\end{equation*}
\end{lemma}
\begin{proof}
	Put
	\begin{equation*}
		\mathscr{F}_n(\delta) \triangleq \left\{f-f_n^*:\; f\in\mathscr{F}_n^{\rm S},\;\|f - f^*_n\|\leq \sqrt{\delta}\right\}
	\end{equation*}
	and $\sigma^2_n\triangleq \sup_{g\in\mathscr{F}_n(\delta)}P_ng^2$. By Dudley's entropy bound, see \cite{koltchinskii2011oracle}, Theorem 3.11, for some numerical constant $C_0$
	\begin{equation}\label{eq:dudley}
		\begin{aligned}
			\psi_n(\delta;\mathscr{F}_n^{\rm S}) & = \E\left[\sup_{g\in\mathscr{F}_n(\delta)}|R_ng|\right]\\
			& \leq \frac{C_0}{\sqrt{n}}\E\int_0^{2\sigma_n}\sqrt{\log N(\mathscr{F}_n(\delta),L_2(P_n),\varepsilon)}\dx\varepsilon\\
		\end{aligned}
	\end{equation}
	and by the symmetrization and contraction inequalities, see \cite{koltchinskii2011oracle}, Theorems 2.1 and 2.3
	\begin{equation*}
		\begin{aligned}	
			\E\sigma^2_n & \leq \E\left[\sup_{g\in\mathscr{F}_n(\delta)} \left|(P_n - P)g^2\right|\right] + \sup_{g\in\mathscr{F}_n(\delta)}Pg^2\\
			& \leq 2\E\left[\sup_{g\in\mathscr{F}_n(\delta)}|R_ng^2|\right] + \delta \\
			& \leq  16\psi_{n}(\delta;\mathscr{F}_n^{\rm S}) + \delta \\
			& \triangleq B,
		\end{aligned}
	\end{equation*}
	where the third line follows since $\|f\|_\infty\leq 1,\forall f\in\mathscr{F}_n^{\rm S}$. Next, for a sequence $x=(x_1,\dots,x_n)\in\mathcal{X}^n$ and a class of functions from $\mathcal{X}$ to $\R$, denoted $\mathscr{F}$, we consider the class $\mathscr{F}|_x \triangleq \left\{(f(x_1),f(x_2),\dots,f(x_n)):\; f\in\mathscr{F} \right\}\subset\R^n$. For $\varepsilon>0$, let $N(\mathscr{F}|_x,\ell_\infty,\varepsilon)$ be the $\varepsilon$-covering number of $\mathscr{F}|_x$ with respect to the $\ell_\infty$ distance. The uniform covering number is defined as
	\begin{equation*}
		N_\infty(\mathscr{F}_n(\delta),n,\varepsilon) \triangleq \max\{N(\mathscr{F}_n(\delta)|_x,\ell_\infty,\varepsilon):\; x\in\mathcal{X}^n \}.
	\end{equation*}
	By \cite{anthony2009neural}, Lemma 10.5, we can bound the $L_2(P_n)$ covering numbers as follows 
	\begin{equation*}
		\begin{aligned}
			N(\mathscr{F}_n(\delta),L_2(P_n),\varepsilon) & = N(\mathscr{F}_n(\delta)|_{X_1,\dots,X_n},\ell_2,\varepsilon) \\
			& \leq \max\{N(\mathscr{F}_n^{\rm S}|_x,\ell_2,\varepsilon):\; x\in\mathcal{X}^n \} \\		
			&\leq N_\infty(\mathscr{F}_n^{\rm S},n,\varepsilon) \\
		\end{aligned}
	\end{equation*}
	Note that for $f_1,f_2\in\mathscr{F}_n^{\rm S}$, we have $f_j(x) = \sigma(\theta_j(x) + c(x)d_j + 1) - \sigma(\theta_j(x)+c(x)d_j - 1) - 1$ for some $\theta_j\in\Theta_n^{\rm S}$ and $|d_j|\leq n$ with $j=1,2$. Since $x\mapsto \sigma(x+1)-\sigma(x-1)-1$ is Lipschitz continuous with Lipschitz constant $1$
	\begin{equation*}
		\max_{1\leq i\leq n}|f_1(X_i) - f_2(X_i)| \leq \max_{1\leq i\leq n}|\theta_1(X_i) - \theta_2(X_i)| + \|c\|_\infty|d_1 - d_2|
	\end{equation*}
	Since the $\varepsilon$-covering number of $[-n,n]$ is $n/\varepsilon$,
	\begin{equation*}
		\begin{aligned}		
			N_\infty(\mathscr{F}_n^{\rm S},n,\varepsilon) & \leq \frac{2\bar c n}{\varepsilon}N_\infty(\Theta_n^{\rm S},n,\varepsilon/2) \\
			& \leq \frac{2\bar c n}{\varepsilon}\left(\frac{8ebn[(d+1)p_n+1](L_0\gamma_n)^2}{\varepsilon(L_0\gamma_n-1)}\right)^{(d+1)p_n+1} \\
			&\leq \frac{2\bar c n}{\varepsilon}\left(\frac{16eb(d+2)L_0np_n\gamma_n}{\varepsilon}\right)^{(d+1)p_n+1},
		\end{aligned}
	\end{equation*}
	where the second inequality follows by \cite{anthony2009neural}, Theorem 14.5 for all $\varepsilon\leq 4b\wedge2\bar c/e$; the third since $(L_0\gamma_n)^2/(L_0\gamma_n-1)\leq 2L_0\gamma_n$ under (iii). In conjunction with the inequality (\ref{eq:dudley}), this shows that
	\begin{equation*}
		\begin{aligned}
			\psi_n(\delta;\mathscr{F}_n^{\rm S}) & \leq  \frac{C_0}{\sqrt{n}}\E\int_0^{2\sigma_n}\sqrt{\log(2\bar c n/\varepsilon) + [(d+1)p_n + 1]\log(16ebL_0(d+2)np_n\gamma_n/\varepsilon)}\dx\varepsilon \\		
			& \leq C_0\sqrt{\frac{2(d+1)p_n}{n}}\E\int_0^{2\sigma_n}\sqrt{\log((2\bar c)\vee (16ebL_0(d+2)p_n\gamma_n)n/\varepsilon)}\dx \varepsilon.
		\end{aligned} 
	\end{equation*}
	Note that for $A>0$
	\begin{equation*}
		\begin{aligned}
			\E\int_0^{2\sigma_n}\sqrt{\log(A/\varepsilon)}\dx \varepsilon & \leq \int_0^{2\sqrt{\E\sigma^2_n}}\sqrt{\log(A/\varepsilon)}\dx \varepsilon \\
			& \leq A\int_0^{2\sqrt{B}/A}\sqrt{\log(1/u)}\dx u \\
			& \leq 4\sqrt{B}\left\{1\vee \sqrt{\log(A/2\sqrt{\delta}) }  \right\} , 
		\end{aligned}
	\end{equation*}
	where the first line follows by Jensen's inequality since $t\mapsto \int_0^th(u)\dx u$ is concave whenever $h$ is non-decreasing; and the last since $\delta\leq B$.
	
	Therefore,
	\begin{equation*}
		\psi_n(\delta;\mathscr{F}_n^{\rm S}) \leq C\sqrt{\frac{p_n}{n}\left\{1\vee \log(An^a/\delta^{1/2})\right\}}\left(4\sqrt{\psi_{n}(\delta;\mathscr{F}_n^{\rm S})} + \sqrt{\delta}\right),
	\end{equation*}
	with $C=4C_0\sqrt{2(d+1)}$, $A=\bar c\vee (8ebL_0(d+2)C_1C_2)$, and $a=c_1+c_2+1$. Solving this inequality for $\psi_n(\delta;\mathscr{F}_n^{\rm S})$ gives
	\begin{equation*}
		\psi_n(\delta;\mathscr{F}_n^{\rm S}) \leq 8C(4C\vee 1)\left[\sqrt{\frac{p_n\delta\{1\vee \log(An^a/\delta^{1/2})\}}{n}}\bigvee  \frac{p_n\{1\vee \log(An^a/\delta^{1/2})\}}{n}\right].
	\end{equation*}	
\end{proof}

Next, we consider the deep learning class
\begin{equation*}
	\begin{aligned}
		\mathscr{F}_n^{\rm DNN} & = \left\{\sigma(\theta+cd+1) - \sigma(\theta+cd-1)-1:\; |d|\leq n,\; \theta\in\Theta_n^{\rm DNN}\right\},
	\end{aligned}
\end{equation*}
where
\begin{equation*}
	\Theta_n^{\rm DNN} = \left\{f(x) = A_L\sigma_{\mathbf{b}_L}\circ A_{L-1}\sigma_{\mathbf{b}_{L-1}}\circ \dots \circ A_1\sigma_{\mathbf{b}_1}\circ A_0x,\; \|f\|_\infty\leq F  \right\}
\end{equation*}
is a deep neural network class such that $\{A_l,\mathbf{b}_l:\;1\leq l\leq L\}$ and $A_0$ are unrestricted free parameters. 

The following result bounds the local Rademacher complexity of the deep learning class $\mathscr{F}_n^{\rm DNN}$ in terms of the pseudo-dimension of the class $\Theta_n^{\rm DNN}$. To define the pseudo-dimension, let $\Theta$ be arbitrary class of functions from $\mathcal{X}$ to $\R$. The pseudo-dimension of $\Theta$ is the largest integer $m$ for which there exists $(x_1,\dots,x_m,y_1,\dots,y_m)\in\mathcal{X}^m\times\R^m$ such that for every $(b_1,\dots,b_m)\in\{0,1\}^m$, there exists $\theta\in\Theta$ such that
\begin{equation*}
	\theta(x_i)>y_i\iff b_i = 1,\qquad \forall 1\leq i\leq m.
\end{equation*}

\begin{lemma}\label{lemma:lrc_deep_learning}
	Suppose that $\|c\|_\infty\leq\bar c$ and that $\Theta_n^{\rm DNN}$ has pseudo-dimension $V\leq n$. Then for $A,K>0$
	\begin{equation*}
		\psi_n(\delta;\mathscr{F}_n^{\rm DNN}) \leq K\left[\sqrt{\frac{V\delta(1\vee \log(An/\delta^{1/2}) )}{n}}\bigvee \frac{V(1\vee \log(An/\delta^{1/2}))}{n}\right].
	\end{equation*}
\end{lemma}
\begin{proof}
	Most of the proof follows from the same steps as in Lemma \ref{lemma:lrc_shallow_learning}. The bound on  uniform covering number is now
	\begin{equation*}
		\begin{aligned}		
			N_\infty(\mathscr{F}_n^{\rm DNN},n,\varepsilon) & \leq \frac{2\bar c n}{\varepsilon}	N_\infty(\Theta_n^{\rm DNN},n,\varepsilon/2) \\
			& \leq \frac{2\bar c n}{\varepsilon}\left(\frac{8en}{\varepsilon V}\right)^{V},
		\end{aligned}
	\end{equation*}
	where we bound the uniform covering numbers relying on \cite{anthony2009neural}, Theorem 12.2.  In conjunction with the inequality (\ref{eq:dudley}), this shows that
	\begin{equation*}
		\begin{aligned}
			\psi_n(\delta;\mathscr{F}_n^{\rm DNN}) & \leq \frac{C_0}{\sqrt{n}}\E\int_0^{2\sigma_n}\sqrt{\log(2\bar c n/\varepsilon) + V\log(16eFn/\varepsilon V)}\dx\varepsilon \\		
			& \leq C_0\sqrt{\frac{2V}{n}}\E\int_0^{2\sigma_n}\sqrt{\log(2(\bar c \vee (8eF))n/\varepsilon)}\dx \varepsilon \\
			& \leq C\sqrt{\frac{V}{n}\left\{1\vee \log(An^a/\delta^{1/2})\right\}}\left(4\sqrt{\psi_{n}(\delta;\mathscr{F}_n^{\rm DNN})} + \sqrt{\delta}\right),
		\end{aligned}
	\end{equation*}
	\end{proof}
	
	\begin{lemma}\label{lemma:lrc_fp}
		Suppose that
		\begin{equation*}
			\psi_n(\delta;\mathscr{F}_n) \leq K\left[\sqrt{\frac{V_n\delta(1\vee \log(An^a/\delta^{1/2}))}{n}}\bigvee \frac{V_n(1\vee \log(An^a/\delta^{1/2}))}{n}\right]
		\end{equation*}
		for some $V_n = o(n)$, $n\geq 1$, and $a,\delta,A>0$. Then there exists $C>0$ such that
		\begin{equation*}
			\psi_{n,\kappa}^{\sharp}(\epsilon) \leq C\left(\frac{V_n\log(n/V_n)}{n}\right)^\frac{\kappa}{2\kappa - 1}.
		\end{equation*}
	\end{lemma}
	\begin{proof}
		Under maintained assumption
		\begin{equation*}
			\psi_{n}^\flat(\sigma) = \sup_{\delta\geq \sigma}\frac{\psi_n(\delta;\mathscr{F}_n)}{\delta} \leq K\left[\sqrt{\frac{V_n (1\vee \log(An^a/\sigma^{1/2}))}{\sigma n}}\bigvee \frac{V_n(1\vee \log(An^a/\sigma^{1/2}))}{\sigma n}\right]
		\end{equation*}
		and whence
		{\footnotesize
			\begin{equation*}
				\begin{aligned}
					\psi_{n,\kappa}^{\sharp}(\epsilon) & = \inf\left\{\sigma>0: \sigma^{1/\kappa - 1}\psi_n^\flat(\sigma^{1/\kappa})\leq \epsilon \right\} \\
					& \leq \inf\left\{\sigma>0: \sqrt{\frac{V_n}{\sigma^{2-1/\kappa} n}}\bigvee \sqrt{\frac{V_n\log(An^a/\sigma^{1/2\kappa})}{\sigma^{2-1/\kappa} n}}\bigvee \frac{V_n}{\sigma n}\bigvee \frac{V_n\log(An^a/\sigma^{1/2\kappa})}{\sigma n}\leq \epsilon/K \right\}.
				\end{aligned}
		\end{equation*}}
		Since the four functions inside the infimum are decreasing in $\sigma$, we have $\psi_{n,\kappa}^{\sharp}(\epsilon) \leq \sigma_1\vee\sigma_2\vee\sigma_3\vee\sigma_4$ with $\sigma_1$ and $\sigma_2$ solving
		\begin{equation*}
			\frac{V_n\log(An^a\sigma_1^{-1/2\kappa})}{n\sigma_1^{2-1/\kappa}} = \epsilon^2/K^2\qquad\text{and}\qquad  \frac{V_n\log(An^a\sigma_2^{-1/2\kappa})}{n\sigma_2} = \epsilon/K.
		\end{equation*}
		and $\sigma_3$ and $\sigma_4$ solving
		\begin{equation*}
			\sqrt{\frac{V_n}{\sigma_3^{2-1/\kappa} n}} = \epsilon/K\qquad \text{and}\qquad \frac{V_n}{\sigma_4 n} = \epsilon/K.
		\end{equation*}	
		To bound $\sigma_1$ and $\sigma_2$, note that
		\begin{equation*}
			\frac{v\log(c/x)}{x^a} = b \iff x = \left(\frac{v}{ab}W_0\left(\frac{abc^a}{v}\right)\right)^{1/a},
		\end{equation*}	
		where $W_0$ is the Lambert $W$-function. Since $W_0(z) \leq \log z,\forall z\geq e$ and $W_0(z)\leq 1$ for all $z\in(0,e]$, this yields
		\begin{equation*}
			x \leq \left(\frac{v}{ab}\right)^{1/a}\bigvee \left(\frac{v}{ab}\log\left(\frac{abc^a}{v}\right)\right)^{1/a}
		\end{equation*}
		Therefore,
		\begin{equation*}
			\begin{aligned}
				\psi^\sharp_{n,\kappa}(\epsilon) & \leq \left(\frac{K^2 V_n}{\epsilon^2n}\right)^\frac{\kappa}{2\kappa - 1} \bigvee \left(\frac{K^2 V_n}{2\epsilon^2n}\log\left(\frac{2(2\kappa-1)\epsilon^2n\left(An^a\right)^{2(2\kappa-1)}}{K^2 V_n}\right)\right)^\frac{\kappa}{2\kappa - 1}\bigvee \\
				&\qquad \bigvee\frac{KV_n}{n\epsilon}\bigvee \frac{KV_n\log(2\kappa\epsilon n(An^a)^{2\kappa}/KV_n)}{2\kappa n\epsilon}.
			\end{aligned}
		\end{equation*}
		and whence since $V_n/n = o(1)$ and $\kappa\geq 1$
		\begin{equation*}
			\psi_{n,\kappa}^{\sharp}(\epsilon) \lesssim \left(\frac{V_n\log(n/V_n)}{n}\right)^\frac{\kappa}{2\kappa - 1}.
		\end{equation*}
	\end{proof}

	\section{Additional Simulation Results} \label{secapp:simulations}
	In  this section, we report additional summary statistics for the simulation comparisons between our weighted logit, the standard symmetric logit model, and the plug-in approach. Tables~\ref{tab:MCresults_summary} corresponds to summary statistics for Table~\ref{tab:MCresults} while Table~\ref{tab:MCresults_plug_in_summary} corresponds to Table~\ref{tab:MCresults_plug_in}. 

	\begin{table}[H]
		\caption{\label{tab:MCresults_summary} Monte Carlo Simulation Results: symmetric vs. asymmetric logit}
		\tablexplain
			This table provides more detailed summary statistics for the ratio $\ell_{logit}/\ell_{w-logit},$ where $\ell_{logit}$ is the social planner loss from using a symmetric logistic model while $\ell_{w-logit}$ is the loss from our weighted logistic model. When the ratio is above 1.00 then the weighted logistic approach is better. See Table~\ref{tab:MCresults} for further details on the simulation design.
		\begin{ctabular}{lccccccccc}
			&	Baseline & $\rho$ 	&  $\tau$  	& $\varphi_0$ & $\varphi_1$	& $\varphi_1$ & $\psi_0$	& $\psi_1$ 	& $\psi_1$\\
			&	case & 0.5	&  1 	& 2 & 2 &3  & 4 & 2	& 3 \\
			&	 & 	& 	&  &  &  &  & &  \\
			& \multicolumn{9}{c}{Sample size $n$ = 1,000} \\
			&	 & 	& 	&  &  &  &  & &  \\
			\multicolumn{10}{l}{Summary statistics for $\ell_{logit}/\ell_{w-logit}$} \\
			&	 & 	& 	&  &  &  &  & &  \\
			Min & 0.00 & 0.00 & 0.77 & 0.00 & 0.00 & 0.00 & 0.00 & 0.00 & 0.00 \\
			25\% & 0.85 & 0.82 & 1.03 & 0.90 & 0.90 & 0.95 & 0.82 & 0.90 & 0.96 \\
			Median & 1.06 & 1.00 & 1.09 & 1.06 & 1.13 & 1.20 & 1.08 & 1.14 & 1.22 \\
			75\% & 1.34 & 1.24 & 1.16 & 1.30 & 1.45 & 1.58 & 1.43 & 1.45 & 1.58 \\
			Max & 10.65 & 8.00 & 1.58 & 11.00 & 14.73 & 13.73 & 13.00 & 9.09 & 18.00 \\
			&	 & 	& 	&  &  &  &  & &  \\
			& \multicolumn{9}{c}{Sample size $n$ = 5,000} \\
			&	 & 	& 	&  &  &  &  & &  \\
			\multicolumn{10}{l}{Summary statistics for $\ell_{logit}/\ell_{w-logit}$} \\
			&	 & 	& 	&  &  &  &  & &  \\
			Min & 0.42 & 0.37 & 0.92 & 0.53 & 0.41 & 0.49 & 0.33 & 0.47 & 0.45 \\
			25\% & 0.91 & 0.89 & 1.03 & 0.93 & 0.94 & 0.97 & 0.91 & 0.94 & 0.97 \\
			Median & 1.03 & 1.00 & 1.06 & 1.02 & 1.07 & 1.11 & 1.08 & 1.07 & 1.11 \\
			75\% & 1.17 & 1.14 & 1.10 & 1.13 & 1.22 & 1.28 & 1.26 & 1.22 & 1.29 \\
			Max & 2.47 & 2.36 & 1.25 & 2.03 & 2.81 & 2.72 & 2.78 & 2.53 & 2.59 \\
			&	 & 	& 	&  &  &  &  & &  \\
			& \multicolumn{9}{c}{Sample size $n$ = 10,000} \\
			&	 & 	& 	&  &  &  &  & &  \\
			\multicolumn{10}{l}{Summary statistics for $\ell_{logit}/\ell_{w-logit}$} \\
			&	 & 	& 	&  &  &  &  & &  \\
			Min & 0.64 & 0.58 & 0.95 & 0.70 & 0.65 & 0.64 & 0.59 & 0.67 & 0.65 \\
			25\% & 0.95 & 0.92 & 1.04 & 0.95 & 0.98 & 1.00 & 0.98 & 0.97 & 1.00 \\
			Median & 1.04 & 1.01 & 1.06 & 1.02 & 1.07 & 1.10 & 1.10 & 1.06 & 1.10 \\
			75\% & 1.14 & 1.11 & 1.08 & 1.09 & 1.17 & 1.22 & 1.23 & 1.17 & 1.21 \\
			Max & 1.95 & 1.81 & 1.19 & 1.56 & 2.19 & 2.38 & 2.15 & 1.98 & 1.94 \\
		\end{ctabular}
	\end{table}

	\begin{table}
		\caption{\label{tab:MCresults_plug_in_summary} Monte Carlo Simulations: plug-in vs. asymmetric logit}
		\tablexplain
			This table provides more detailed summary statistics for the ratio $\ell_{plugin}/\ell_{w-logit},$ where $\ell_{plugin}$ is the social planner loss from using the asymmetric plug-in approach for the logistic model while $\ell_{w-logit}$ is the loss from our weighted logistic model. When the ratio is above 1.00 then the weighted logistic approach is better. See Table~\ref{tab:MCresults_plug_in} for further details on the simulation design.
		\begin{ctabular}{lccccccccc}
			&	Baseline & $\rho$ 	&  $\tau$  	& $\varphi_0$ & $\varphi_1$	& $\varphi_1$ & $\psi_0$	& $\psi_1$ 	& $\psi_1$\\
			&	case & 0.5	&  1 	& 2 & 2 &3  & 4 & 2	& 3 \\
			&	 & 	& 	&  &  &  &  & &  \\
			& \multicolumn{9}{c}{Sample size $n$ = 1,000} \\
			&	 & 	& 	&  &  &  &  & &  \\
			\multicolumn{10}{l}{Summary statistics for $\ell_{plugin}/\ell_{w-logit}$} \\
			&	 & 	& 	&  &  &  &  & &  \\
			Min & 0.10 & 0.00 & 0.54 & 0.00 & 0.20 & 0.21 & 0.27 & 0.18 & 0.23 \\
			25\% & 1.02 & 0.88 & 0.76 & 0.82 & 1.07 & 1.13 & 1.28 & 1.08 & 1.12 \\
			Median & 1.31 & 1.16 & 0.81 & 1.10 & 1.39 & 1.45 & 1.63 & 1.40 & 1.44 \\
			75\% & 1.68 & 1.53 & 0.86 & 1.43 & 1.80 & 1.87 & 2.09 & 1.80 & 1.86 \\
			Max & 11.55 & 10.90 & 1.01 & 15.00 & 11.90 & 12.90 & 18.15 & 17.85 & 21.50 \\
			&	 & 	& 	&  &  &  &  & &  \\
			& \multicolumn{9}{c}{Sample size $n$ = 5,000} \\
			&	 & 	& 	&  &  &  &  & &  \\
			\multicolumn{10}{l}{Summary statistics for $\ell_{plugin}/\ell_{w-logit}$} \\
			& & 	& 	&  &  &  &  & &  \\
			Min & 0.45 & 0.36 & 0.71 & 0.30 & 0.51 & 0.57 & 0.62 & 0.51 & 0.54 \\
			25\% & 1.01 & 0.88 & 0.82 & 0.83 & 1.06 & 1.11 & 1.24 & 1.06 & 1.11 \\
			Median & 1.15 & 1.03 & 0.84 & 0.97 & 1.21 & 1.27 & 1.42 & 1.22 & 1.27 \\
			75\% & 1.33 & 1.19 & 0.86 & 1.12 & 1.39 & 1.45 & 1.63 & 1.39 & 1.45 \\
			Max & 2.76 & 3.10 & 0.94 & 2.70 & 2.59 & 3.18 & 3.68 & 2.74 & 3.47 \\
			&	 & 	& 	&  &  &  &  & &  \\
			& \multicolumn{9}{c}{Sample size $n$ = 10,000} \\
			&	 & 	& 	&  &  &  &  & &  \\
			\multicolumn{10}{l}{Summary statistics for $\ell_{plugin}/\ell_{w-logit}$} \\
			&	 & 	& 	&  &  &  &  & &  \\
			Min & 0.53 & 0.40 & 0.77 & 0.41 & 0.57 & 0.59 & 0.60 & 0.55 & 0.58 \\
			25\% & 0.91 & 0.81 & 0.83 & 0.77 & 0.96 & 1.00 & 1.13 & 0.96 & 1.00 \\
			Median & 1.01 & 0.91 & 0.85 & 0.86 & 1.06 & 1.11 & 1.24 & 1.07 & 1.11 \\
			75\% & 1.12 & 1.02 & 0.86 & 0.97 & 1.18 & 1.23 & 1.38 & 1.18 & 1.23 \\
			Max & 1.78 & 1.83 & 0.92 & 1.70 & 1.75 & 1.92 & 2.29 & 2.09 & 1.99 \\
		\end{ctabular}
	\end{table}

	\clearpage

	\medskip

\section{Empirical Illustration: Further Details \label{appsec:empirics}}
	In this section, we provide details on the empirical application including data cleaning, summary statistics, and additional results. We start with a subsection on data cleaning.

\subsection{Data cleaning \label{appsubsec:cleaning}}

While working with the COMPAS data, we discovered several issues with how the dataset was constructed by ProPublica journalists. These issues include missing charge descriptions and several inconsistencies. Therefore, we reconstructed a more complete dataset starting from the raw ProPublica data files. Reconstructing a clean dataset involved several challenges, including merging COMPAS data with arrest, charge, and individual information.

\medskip

Ideally, for each defendant we should have only one entry in the dataset, with that record containing information on current charges and past crimes to predict recidivism. However, the raw ProPublica files have multiple records for the same arrest in the COMPAS dataset, differing only in charges. A typical example is as follows: we see a suspect who has a single case with one charge for a traffic violation and another for obstructing justice, which in turn could have different charge degrees. However, in the raw data, we have two separate entries for this suspect, one for each charge. This requires careful merging of all current and past charges related to the same arrest into a single record.

\medskip

Considering these limitations, we constructed a new dataset by removing redundancies and merging all charges related to the same arrest. Specifically, we started with the raw COMPAS data, consisting of separate tables for charges, arrests, jail and prison history, personal information, and COMPAS evaluations. We used the latter as an anchor to reconstruct a complete record for each defendant containing the following information: (a) past crimes and future recidivism, and (b) all charges from the same arrest consolidated into a single variable. Following ProPublica's analysis, we considered only COMPAS assessment records for the risk of recidivism that were made within 30 days of an arrest.

Additionally, we carefully reconstructed the recidivism status as the occurrence of another misdemeanor or felony charge within two years of the initial COMPAS screening, and we removed several individuals who were recorded multiple times. These data-cleaning steps resulted in a reduction of the sample size from about 11,757 to 8,227, which is slightly larger than the dataset ProPublica used for its racial-bias analysis, but with detailed records on current and past charges, covering the exact Florida Statute chapters and charge degrees for each criminal charge and ordinance violation.

\medskip

Next, we applied other data processing similar to those originally done by COMPAS. For example, we do not consider municipal or county ordinances or traffic offenses to be recidivism, but we do use the past county ordinances and traffic offenses to predict recidivism. Therefore, in addition to using the original covariates in the COMPAS data, we created new variables based on crime statute degree information and more detailed personal information. The key variables of the cleaned data are as follows:	
	\begin{itemize}
		\item is\_recid: We consider all future crimes (felony and misdemeanor) charged within two years after the first COMPAS score.
		\item Age: A raw number which we standardize for machine learning algorithms. We also create a quadratic age term to capture nonlinearities.
		\item Marital status: Dummies for single, married, divorced, separated, significant other, and unknown as the baseline.
		\item Gender: Male as the baseline and female as a separate dummy.
		\item An African American indicator, using Others as the baseline.
		\item Crime history:
		\begin{itemize}
			\item past juvenile felony, misdemeanor, and other offense charge counts.
			\item dummies with interactions between the current and past charge descriptions and charge degrees (felony or misdemeanor). For example, Theft.Robery\_current\_M dummy equal one means that the current case includes the misdemeanor charge for theft/robbery-related crime and it could be any degree of a misdemeanor. Similarly, Assault.Battery\_past\_F means the suspect has a past felony charge related to assault/battery. We only have these variables for common crimes like theft, robbery, burglary, drug abuse, assault, and traffic-related crimes.
			\item dummies with interactions between the detailed charge description and more detailed degree both for current and past charges. For example, X316\_f2\_current equals one means the current charge (while assessed for future recidivism) includes a 2nd-degree felony corresponding to Florida Statute, Chapter 316 --- ''State Uniform Traffic Control". Similarly, X316\_m1\_past means that the suspect has a past 1st degree misdemeanor charge. 
		\end{itemize}
	\end{itemize}
The reconstructed dataset has a total of 495 variables, and we report the summary statistics for this dataset with 8,227 individual entries in Table~\ref{tab:sumstats_new}.
	
\begin{table}[H]
	\caption{\label{tab:sumstats_new} Summary Statistics}
		\tablexplain
		{\footnotesize
			Summary statistics of recidivism dataset from Broward County, Florida, aggregated to the individual level.}
		{\footnotesize
	\begin{ctabular}{lccccccc}
	\toprule
	\midrule
	is\_recid & No & Yes &  &  &  &  &  \\
	& 5648 & 2579 &  &  &  &  &  \\
	 &  &  &  &  &   & & \\
	Gender & Male & Female &  &  &  &  &  \\
	& 6508 & 1719 &  &  &  &  &  \\
	&  &  &  &  &   & & \\
	Curr. Charge & Assault/ & Theft/ & Burglary & Drug & Driver & Motor & \\
	& Battery & Robbery & & Abuse & License & Vehicles & \\
	& 3824 & 2583 & 1290 & 3677 & 1593 & 1476  \\
	 &  &  &  &  &   & & \\
	 & Fraud & Obst. Justice	 & Other & Forgery &  &  &  \\
	 & 312 & 882 & 2820 & 275 &  &  &  \\
	 &  &  &  &  &   & & \\
	Degree & Felony 1st & Felony 2nd & Felony 3rd+ & Misd. 1st & Misd. 2nd & Ord./Traffic \\
	& 380 & 1951 & 7196 & 6837 & 2377 & 157  \\
	&  &  &  &  &   & & \\
	& Other & & & & & \\
	& 2633 & & & & & \\
	&  &  &  &  &   & & \\
	Race & African & Asian & Caucasian & Hispanic & Native & Other &  \\
	& American &  &  &  &  American &  &  \\
	& 4109 & 45 & 2797 & 720 & 21 & 535 &  \\
	 &  &  &  &  &   & & \\
	Marital Status & Married & Single & Divorced & Separated & Widowed & Unknown &  \\
	& 991 & 6380 & 328 & 194 & 36 & 34 &  \\
	 &  &  &  &  &   & & \\
	 & Min & 25th Pct. & Median & Mean & 75th Pct. & Max &  \\
	Prior Charge Count & 1 & 3 & 7 & 11 & 15 & 192 &  \\
	Age & 16 & 23 & 29 & 32 & 40 & 94 &  \\
	\bottomrule
	\end{ctabular}}
\end{table}
Finally, for replication purposes, we have made the cleaned data available at \url{https://github.com/ababii/cost_sensitive_classification}.

\subsection{Additional Empirical Results}
The empirical model estimates use an 80/20 split for training and testing. We report the statistically significant post-LASSO logit estimates in Table~\ref{tab:post_lasso}. The charges that increase the odds of recidivism include several current and past charges (e.g. obstructing justice, theft/robbery, assault/battery) and past charges (driving without a license, assault). The variables that decrease the odds of recidivism include age, female gender, traffic-related misdemeanors and environmental control charges, among others. It is worth noting that the reported standard errors do not account for the model selection uncertainty and should be understood as potentially optimistic.
	
\begin{table}[H] \centering 
	\caption{Post-LASSO logit estimates} 
	\label{tab:post_lasso}
	{\footnotesize
	\begin{tabular}{@{\extracolsep{5pt}}lcclc} 
	  \\[-1.8ex]\hline 
	  \hline \\[-4ex] 
	  & \multicolumn{1}{c}{\textit{is\_recid}} &  & &  \\ 
		Intercept & $-$1.343$^{***}$ & & age & $-$0.5558$^{***}$ \\
		& (0.0677) & & & (0.0261) \\
		age2 & 0.1468$^{***}$ & & sex\_female & $-$0.3808$^{***}$ \\
		& (0.0387) & & & (0.0710) \\
		mariage\_Married & $-$0.1230 & & mariage\_Significant.Other & 0.3930$^{***}$ \\
		& (0.0943) & & & (0.1421) \\
		mariage\_Unknown & 0.1099 & & mariage\_Widowed & 0.3777 \\
		& (0.4274) & & & (0.4662) \\
		af\_ind & 0.1498$^{***}$ & & Theft.Robery\_current\_M & 0.2374$^{*}$ \\
		& (0.0559) & & & (0.1224) \\
		Motor.vehicles\_current\_M & $-$0.2025$^{***}$ & & Obstructing.justice\_current\_M & 0.2222$^{***}$ \\
		& (0.0534) & & & (0.0752) \\
		Assault.Battery\_past\_M & 0.0911$^{*}$ & & Theft.Robery\_past\_M & 0.1459$^{**}$ \\
		& (0.0487) & & & (0.0673) \\
		Driver.license\_past\_F & $-$0.1061$^{*}$ & & X322\_m2\_current & 0.1926$^{***}$ \\
		& (0.0570) & & & (0.0609) \\
		X509\_m2\_current & 2.3442$^{*}$ & & X741\_m1\_current & 0.5915$^{***}$ \\
		& (1.2621) & & & (0.1979) \\
		X796\_m2\_current & 0.6431$^{**}$ & & X806\_f3\_current & 0.3228$^{*}$ \\
		& (0.3217) & & & (0.1869) \\
		X893\_f1\_current & $-$0.3845$^{*}$ & & X893\_f3\_current & 0.1201$^{**}$ \\
		& (0.2005) & & & (0.0480) \\
		X12569\_co3\_past & 2.3416$^{*}$ & & X316\_tcx\_past & 0.4437$^{**}$ \\
		& (1.2175) & & & (0.1988) \\
		X322\_0\_past & 0.1141$^{***}$ & & X324\_0\_past & 1.9242$^{**}$ \\
		& (0.0356) & & & (0.9326) \\
		X403\_0\_past & $-$0.7651$^{*}$ & & X539\_f3\_past & 0.9116$^{***}$ \\
		& (0.4231) & & & (0.2961) \\
		X787\_f1\_past & 1.2310$^{**}$ & & X790\_f3\_past & 0.2899$^{*}$ \\
		& (0.5793) & & & (0.1489) \\
		X856\_m2\_past & 0.1906$^{**}$ & & X877\_m2\_past & 0.2897$^{**}$ \\
		& (0.0951) & & & (0.1367) \\
		X893\_m1\_past & 0.1374$^{***}$ & & X918\_f3\_past & $-$0.2169$^{*}$ \\
		& (0.0279) & & & (0.1223) \\
		olhd1419\_mo3\_past & 1.4193$^{*}$ & & ofld16113\_mo3\_past & $-$1.9824$^{**}$ \\
		& (0.0690) & & & (0.9803) \\
		ofld16157\_mo3\_past & 0.3444$^{*}$ & & ofld26183\_mo3\_past & $-$2.2661$^{*}$ \\
		& (0.1840) & & & (1.2169) \\
		ofld858b\_mo3\_past & 1.3354$^{*}$ & & ohwd13327c1\_mo3\_past & 2.0361$^{*}$ \\
		& (0.7427) & & & (1.0888) \\
	  \hline 
	  \hline \\[-1.8ex] 
	  \textit{Note:} & \multicolumn{4}{r}{$^{*}$p$<$0.1;\; $^{**}$p$<$0.05;\; $^{***}$p$<$0.01} \\ 
	\end{tabular}} 
  \end{table}

  Next, we look at how the group-specific false positive and false negative rates change when we vary the cost of false negative mistakes. Figure~\ref{fig:fp_fn_application_race} shows that there is a substantial discrepancy between how African American and other defendants are treated for most of the values of $\psi$. Interestingly, balancing the false positive/negative rates is not enough to achieve algorithmic fairness.
  \begin{figure}[h]
	  \centering
	  \begin{subfigure}{0.49\textwidth} 
		  \includegraphics[width=\textwidth]{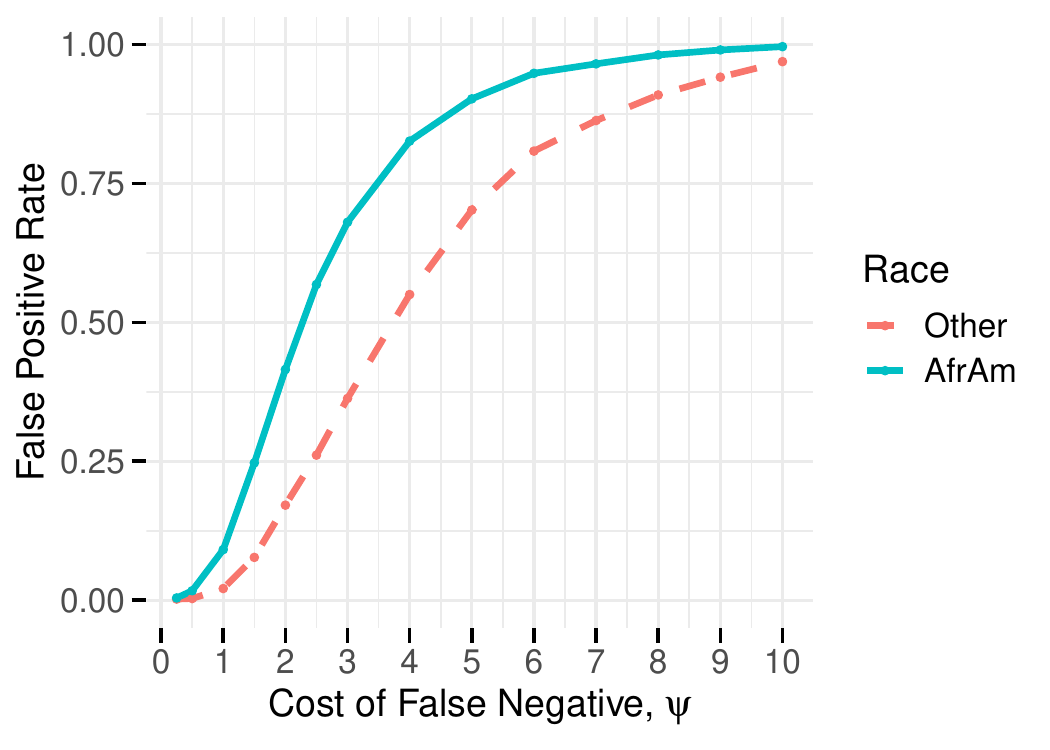}
		  \caption{Logit: FP rates by group} 
	  \end{subfigure}
	  \begin{subfigure}{0.49\textwidth} 
		  \includegraphics[width=\textwidth]{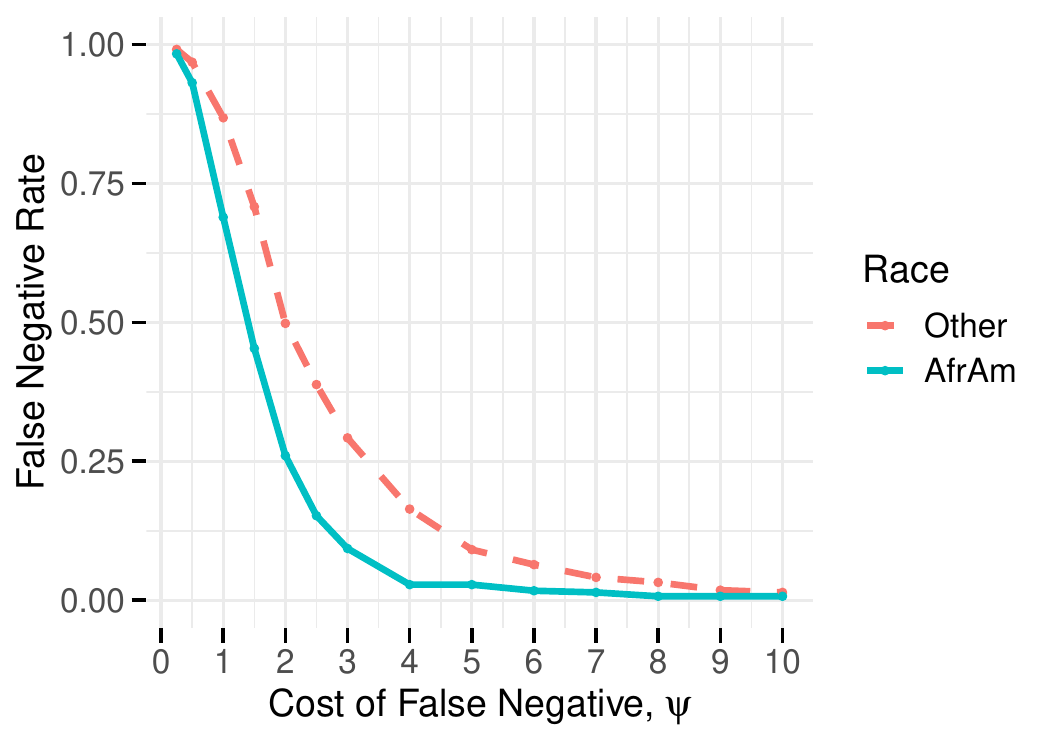}
		  \caption{Logit: FN rates by group} 
	  \end{subfigure}
	  \begin{subfigure}{0.49\textwidth} 
		  \includegraphics[width=\textwidth]{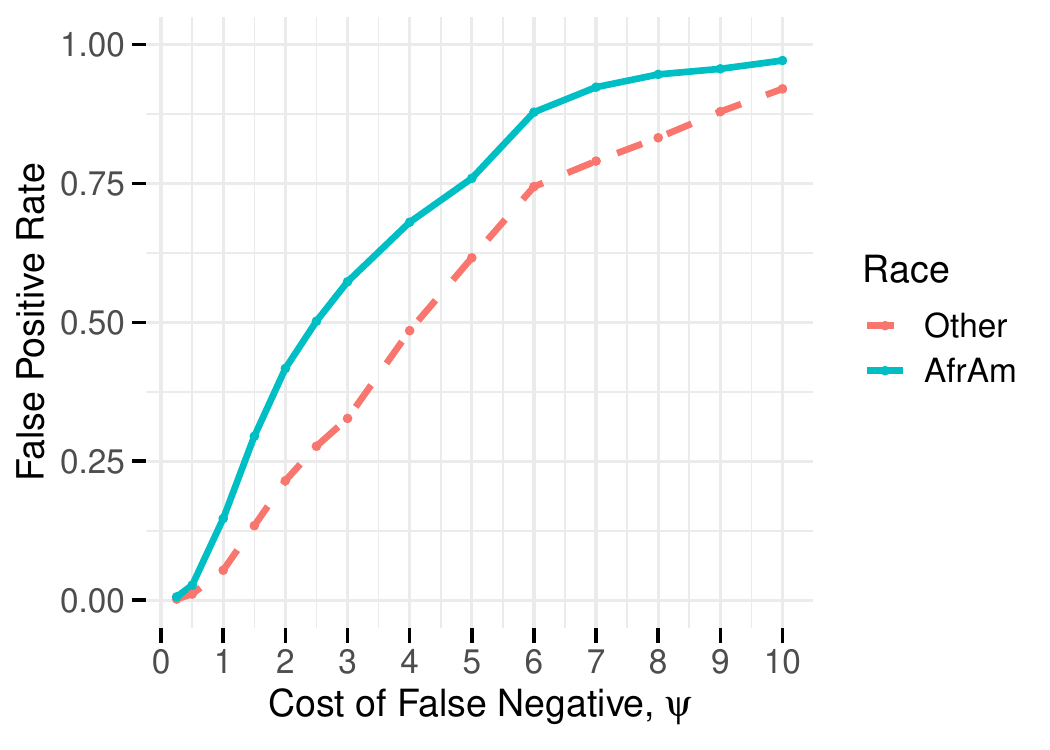}
		  \caption{Boosting: FP rates by group} 
	  \end{subfigure}
	  \begin{subfigure}{0.49\textwidth} 
		  \includegraphics[width=\textwidth]{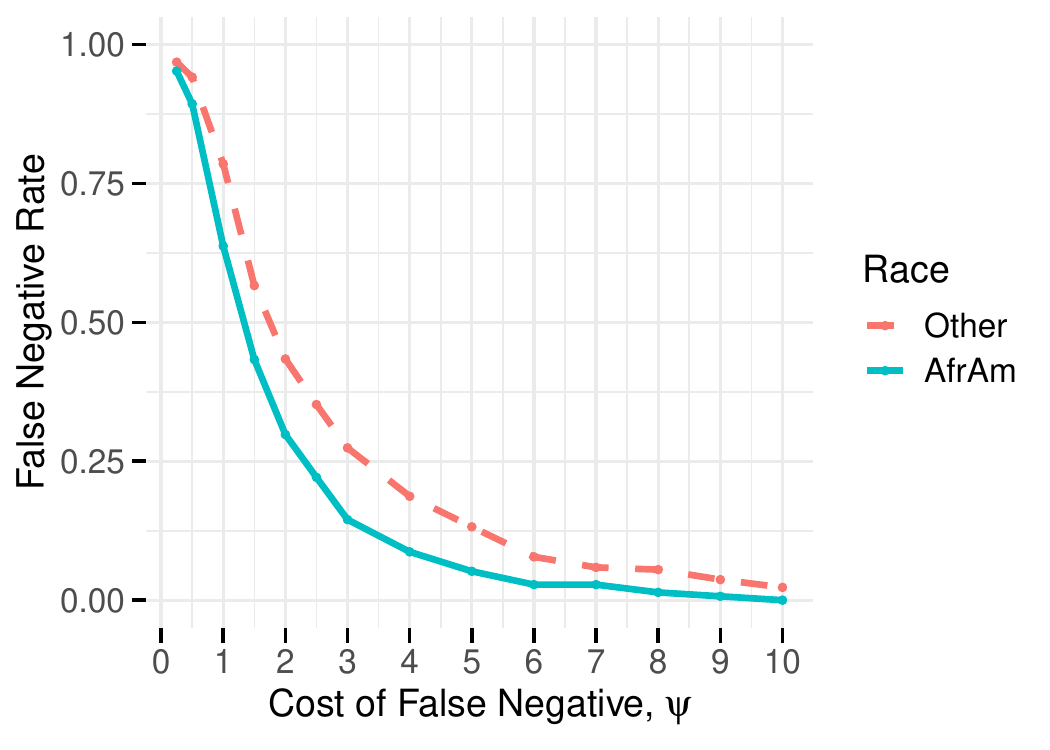}
		  \caption{Boosting: FN rates by group} 
	  \end{subfigure}
	  \caption{Asymmetric Binary Choice: penalizing false positive mistakes. The figure shows that increasing the cost of false negative mistakes in the loss function is not enough to balance the group-specific FP and FN rates.} 
	  \label{fig:fp_fn_application_race}
  \end{figure}
  
  To see the magnitude of this discrepancy, we report the exact values in Table~\ref{tab:balanced}. We find that for a balanced classifier, African American defendants have twice the false positive rate of  the rest of the population, which is similar to the finding reported by the ProPublica article. Interestingly, the discrepancy is slightly smaller with boosting.
  \begin{table}[!htbp] \centering 
	  \caption{Asymmetric Classifiers with Balanced FP/FN Rates. The cost of false negative mistakes is calibrated to $\psi=2.25$.} 
	  \label{tab:balanced} 
	  \begin{tabular}{@{\extracolsep{5pt}} lccc} 
		  \\[-1.8ex]\hline 
		  \hline \\[-1.8ex] 
		  & FP & FN & AUC \\ 
		  & \multicolumn{3}{c}{Logit} \\
		  African American & $0.496$ & $0.201$ & $0.722$ \\ 
		  Others & $0.215$ & $0.425$ & $0.735$ \\ 
		  All & $0.344$ & $0.297$ & $0.735$ \\ 
		& & & \\
		& \multicolumn{3}{c}{Boosting} \\
		African American & $0.469$ & $0.239$ & $0.700$ \\ 
		Other & $0.230$ & $0.411$ & $0.740$ \\ 
		All & $0.339$ & $0.313$ & $0.728$ \\ 
	  \end{tabular} 
  \end{table}

\end{document}